\documentclass[preprint,12pt]{elsarticle}

\usepackage[T1]{fontenc}
\usepackage[utf8]{inputenc}
\usepackage{lmodern}

\usepackage{amsmath,amssymb,amsthm}

\usepackage{graphicx}   % (already loaded by elsarticle; harmless to keep)
\usepackage{booktabs}
\usepackage{tabularx}
\usepackage{caption}
\usepackage[table]{xcolor}
\usepackage{enumitem}
\usepackage{microtype}

\usepackage{tikz}
\usetikzlibrary{positioning}

\usepackage{appendix}

\newtheorem{theorem}{Theorem}[section]
\newtheorem{lemma}[theorem]{Lemma}
\newtheorem{proposition}[theorem]{Proposition}
\newtheorem{corollary}[theorem]{Corollary}

\theoremstyle{definition}
\newtheorem{definition}[theorem]{Definition}
\theoremstyle{remark}
\newtheorem{remark}[theorem]{Remark}
\numberwithin{equation}{section}

\AtBeginDocument{}

\usepackage[hyphens]{url}      % defines \url and \urldef
\usepackage[hidelinks]{hyperref}  % optional but recommended; loads url too
\providecommand{\doi}[1]{\href{https://doi.org/#1}{doi:#1}}
\begin{document}
\makeatletter
\def\ps@pprintTitle{%
 \let\@oddhead\@empty
 \let\@evenhead\@empty
 \let\@oddfoot\@empty
 \let\@evenfoot\@empty
}
\makeatother

\begin{frontmatter}

\title{Formal Foundations for Controlled Stochastic Activity Networks}

\author[sharif,umich]{Ali Movaghar\corref{cor1}}
\cortext[cor1]{Corresponding author.}
\ead{movaghar@sharif.edu}
\ead{movaghar@umich.edu}

\address[sharif]{Sharif University of Technology, Tehran, Iran}
\address[umich]{Computer Science and Engineering Division, University of Michigan, Ann Arbor, MI, USA}

\begin{abstract}
We introduce Controlled Stochastic Activity Networks (Controlled SANs), a formal extension of classical Stochastic Activity Networks that integrates explicit control actions and policies into a unified framework for modeling controlled stochastic discrete-event and real-time systems. Controlled SANs capture nondeterminism, probabilistic branching, and stochastic timing while providing rigorous semantics for policy-driven decision-making.

We develop a hierarchical, automata-theoretic semantics for Controlled SANs through controlled automata, controlled probabilistic automata, and controlled stochastic automata. At the nondeterministic level, we distinguish structurally available controls from behaviorally admissible controlled transitions. We establish computability properties of automata realized by controlled activity networks and prove a constructive modeling-power result: every decidable, effectively finitely branching controlled automaton with finite activity and control-action sets and a finite nonempty initial-state set admits a computable, well-behaved controlled activity network realization whose reachable restriction is effectively isomorphic to that of the original automaton. We formalize a hierarchy of control policies, ranging from memoryless and finite-memory policies to stack- and tape-augmented policies, and investigate their expressive power through associated language classes. We also develop behavioral equivalences, including bisimulation and stochastic isomorphism, to support abstraction and compositional reasoning.

A central semantic mechanism is the use of contingent activity-indexed control assignments, in which admissible controls are provisionally assigned before the realized activity is known. At the probabilistic level, an activity-generation law together with this contingent-action semantics yields a discrete-time Markov decision process (DTMDP) representation of a controlled probabilistic automaton, while every conventional DTMDP admits a single-activity controlled probabilistic automaton representation. At the stochastic level, the Markovian subclass admits, under the stated assumptions, a continuous-time Markov decision process (CTMDP) projection, while every conventional CTMDP admits a single-activity controlled Markovian automaton representation. For randomized policies, the corresponding decision-process representations are subject to the adopted product-form policy convention.

These relationships connect the formal semantics to established methods for reward optimization and policy synthesis when the relevant objectives, information structures, and policy classes are preserved. The framework provides a mathematical basis for specification, verification, performance analysis, optimization, and synthesis of controlled stochastic discrete-event systems, including applications involving reinforcement learning and safety-critical control. It also suggests potential applications to the formal modeling and analysis of agentic AI systems, where autonomous decision-making interacts with concurrency, uncertainty, and stochastic timing.
\end{abstract}


\end{frontmatter}
\setcounter{equation}{0}
\section{Introduction}

Safety-critical systems---such as those used in avionics,
autonomous vehicles, medical devices, and industrial
automation---are characterized by stringent correctness
requirements, continuous operation, and the need to make
timely decisions under uncertainty. These systems typically
involve distributed architectures with interacting components
that operate concurrently and must respond predictably to
real-time stimuli, probabilistic events, and external commands.
Any error in their behavior can have catastrophic consequences,
making formal modeling, verification, and assurance indispensable
to their development and deployment
\cite{hoare1969axiomatic,lamport1977proving,
newcombe2015aws,woodcock2009survey,
seshia2022cyberphysical}.

To analyze, control, and synthesize such systems, we require
modeling frameworks that are expressive enough to capture their
complex dynamic behavior, yet structured enough to support
rigorous reasoning and tool-based analysis. Classical
paradigms---such as finite-state machines \cite{Sipser2013},
timed automata \cite{alur1994theory}, Petri nets
\cite{Peterson1981}, and stochastic process algebras
\cite{hillston1996compositional,hermanns2002interactive}---
have each contributed key insights in this direction. In
particular, \emph{Stochastic Activity Networks} (SANs)
\cite{movaghar1984stochastic,meyer1985stochastic,
movaghar2001stochastic} generalize timed Petri nets by
associating activities with probabilistic durations and
gate-driven transitions. SANs provide an intuitive graphical
representation of concurrent stochastic behavior and support
compositional modeling, making them useful for performance and
dependability analysis.

Despite these strengths, classical SANs do not provide an
explicit semantic treatment of \emph{control}: decision logic
that governs system evolution according to the current state,
observed activities, history, or other available information.
Such control is pervasive in safety-critical and adaptive
systems, including real-time scheduling, fault-tolerance
strategies, resource allocation, adaptive control loops, and
fail-safe mechanisms. Integrating these decisions directly into
the stochastic model is important not only for verification and
performance analysis, but also for systematic controller and
policy synthesis. This motivates a framework in which control,
probabilistic behavior, stochastic timing, and concurrency share
a common formal semantics.

In this paper, we introduce \emph{Controlled Stochastic Activity
Networks} (Controlled SANs), a formal extension of SANs that
incorporates explicit control actions and policies into the
modeling process. Our objective is to provide a unified formalism
for systems in which control decisions interact with
nondeterministic, probabilistic, and timed behavior. The resulting
framework supports reasoning about controller synthesis,
reachability under uncertainty, behavioral equivalence, reward
optimization, and related qualitative and quantitative objectives
within a compositional and mathematically rigorous setting.

The controlled models developed here are motivated by reactive
systems, including cyber-physical and autonomous systems, in
which control decisions interact with observable activities and
system state. An activity may represent, for example, the
arrival of a customer, the receipt of a packet, the completion
of a task, or the detection of a fault. A central semantic idea
of the framework is the use of \emph{contingent activity-indexed
control assignments}. At a decision point, the controller can
provisionally assign an admissible control to each relevant
activity, forming a contingent-action map; when an activity
occurs, only the control assigned to that activity is executed.
At the probabilistic level, an activity-generation law determines
the activity conditionally on the current state and selected
contingent action. At the stochastic level, the provisional
assignments may additionally affect activity completion rates
and hence the timed race itself. For randomized contingent-action
policies, the provisional activity-specific assignments are
conditionally independent given the information available to the
controller and therefore induce a product-form distribution over
contingent-action maps. This common contingent-action view provides
a bridge from the automata semantics to standard decision-process
representations, subject to preservation of the relevant policy
class.

\subsection{Motivation from Safety-Critical Domains}

Consider an autonomous drone navigating an urban environment.
Its mission requires sensing obstacles, adjusting altitude and
velocity in real time, avoiding collisions, and ensuring safe
return. Each of these behaviors involves uncertain timings
(e.g., sensor latencies), probabilistic outcomes (e.g., wind
disturbances), and control decisions (e.g., evasive maneuvers).
Similarly, a medical infusion pump must adjust dosages based on
monitored vitals, react to hardware faults, and enforce safety
constraints on flow rates---all while managing concurrent
activities with varied timing profiles. In both cases, a formal
model must:
\begin{itemize}
    \item represent concurrent activities with different timing
    semantics;
    \item express control actions and policies that guide system
    evolution based on state, history, activities, or other
    available information;
    \item support probabilistic modeling of failures, delays, and
    other random phenomena;
    \item enable quantitative and qualitative reasoning about
    safety, liveness, performance, and reward objectives; and
    \item provide a semantic basis for systematic policy synthesis
    and scalable analysis.
\end{itemize}

Controlled SANs address these needs by extending SANs with
policy-driven control and by connecting their formal semantics
with established stochastic decision-process models. These
connections are important computationally as well as
semantically: when an appropriate DTMDP or CTMDP representation
is available, established optimization and policy-synthesis
methods can be brought to bear under their corresponding
assumptions, provided that the relevant objective, information
structure, and admissible policy class are preserved.

As explicit state-space analysis becomes infeasible, the same
semantic foundation can support approximation, simulation,
and reinforcement-learning approaches. Of particular interest
are learning-enabled controllers for safety-critical systems,
where policy optimization must be accompanied by rigorous
reasoning about admissible behavior, safety constraints, and
quantitative performance. The underlying activity-network
structure may also provide useful information for state
aggregation, compositional analysis, and structure-aware
optimization and learning. Establishing formal safety and
learning guarantees for such controllers remains an important
direction for future research.

Another potential application concerns agentic artificial
intelligence (AI) systems, in which autonomous agents
select actions, interact with their environments, and
execute possibly concurrent workflows under uncertainty
and timing constraints. Controlled SANs offer a possible
formal modeling framework for such systems by representing
agent activities and environmental interactions through
controlled state transitions, with explicit policies
governing admissible actions. This perspective may
support the analysis of reachability, action admissibility,
and stochastic or time-dependent behavior in agent
workflows. Developing concrete modeling abstractions,
verification methods, and software tools for such
applications remains a direction for future research.

\subsection{Overview of the Model and Contributions}

We develop Controlled SANs through a layered approach that
introduces increasing levels of semantic detail. At each level,
we formalize system behavior using corresponding
automata-theoretic models, define classes of control policies,
and study expressive power and behavioral equivalence
\cite{Milner1989,baier2008}.

\textbf{1. Nondeterministic Semantics.}
At the foundational level, Controlled SANs extend classical
activity networks by introducing control actions that mediate
system evolution. Marking transitions induced by activity
completions and gate functions are represented through
\emph{controlled automata}, whose transitions are labeled by
activities and control actions. These automata provide a
state-based semantic foundation for reasoning about controlled
nondeterministic executions, language acceptance, and behavioral
equivalence.

At this level, we distinguish structurally available control
actions from behaviorally admissible controlled transitions.
The former are specified by the activity-network structure,
whereas the latter are determined by the existence of
corresponding transitions in the realized controlled automaton.
We also establish a constructive modeling-power result showing
that, under the stated decidability and finiteness assumptions,
a well-behaved computable controlled activity network can
realize the reachable behavior of a decidable, effectively
finitely branching controlled automaton up to effective
isomorphism. This result provides a formal connection between
the operational behavior of controlled activity networks and
the computational expressiveness of controlled automata.

\textbf{2. Probabilistic Semantics.}
To incorporate probabilistic branching, we introduce
\emph{controlled probabilistic activity networks} and their
automata-theoretic semantics as \emph{controlled probabilistic
automata} (CPAs). When an activity-generation law is supplied,
a contingent-action map assigns an admissible control to each
enabled activity before the next activity is generated. The
activity is then generated conditionally on the current state
and selected contingent action, and the control assigned to that
activity governs the resulting probabilistic transition. This
semantics yields a conventional \emph{discrete-time Markov
decision process} (DTMDP) representation in which
contingent-action maps serve as actions
\cite{puterman1994markov}. Deterministic contingent-action
policies correspond directly, while randomized CPA policies
correspond to the product-form DTMDP policies induced by the
conditionally independent activity-specific assignments.
Conversely, every conventional DTMDP admits a single-activity
CPA representation. The resulting relationship provides a
semantic bridge between the CPA framework and established
methods for discrete-time stochastic control, reward
optimization, and policy synthesis, subject to preservation
of the relevant objective, information structure, and
admissible policy class.

\textbf{3. Stochastic Semantics.}
For real-time systems, we enrich the framework with stochastic
timing by assigning activity requirements and
control-dependent processing rates to timed activities. The
resulting \emph{controlled stochastic activity networks} are
given automata-theoretic semantics by \emph{controlled
stochastic automata}. Before each timed race, the controller
selects a contingent-action map assigning an admissible control
to each enabled timed activity; these assignments may affect
the corresponding completion rates. When an activity completes,
its already assigned control is executed and governs the
probabilistic successor-state distribution. A new contingent
action is selected before the next race.

For the Markovian subclass with exponential timing, the
physical-state behavior admits, under the stated assumptions,
a conventional \emph{continuous-time Markov decision process}
(CTMDP) projection in which contingent-action maps serve as
CTMDP actions \cite{puterman1994markov}. The projection
preserves the conditional holding-time and successor-state
laws at the physical-state level under corresponding
deterministic policies and, for randomized policies, under
the adopted product-form correspondence, while marginalizing
over the identity of the completed activity. Conversely,
every conventional CTMDP admits a single-activity controlled
Markovian automaton representation. General non-exponential
timing remains within the broader stochastic semantics and
may require additional timing-state information or more
general stochastic-process analysis.

Across these models, we define a spectrum of policy
types---including memoryless, history-dependent,
finite-memory, stack-augmented, and tape-augmented
policies---and classify their expressive power using formal
language theory. This taxonomy provides a rigorous way to
relate controller memory and computational resources to
behavioral expressiveness and is relevant to both verification
and policy synthesis.

The DTMDP representation and CTMDP projection also connect the
formal semantics to established optimization methods. Under
the conditions developed in this paper, corresponding reward
representations permit the preservation of expected
finite-horizon and discounted accumulated rewards. For models
permitting explicit state-space analysis, classical
dynamic-programming and stochastic-control techniques can be
used under their standard assumptions. For larger models,
approximate dynamic programming, sampling and simulation, and
reinforcement-learning methods provide possible alternatives.
When the model is explicitly available, it can support
model-based learning and planning; for very large systems, the
compositional activity-network structure may further support
state aggregation, feature construction, guided exploration,
and structure-aware policy synthesis.

These connections are particularly relevant to learning-enabled
control of safety-critical systems, where stochastic behavior,
control constraints, and quantitative performance objectives
must be considered within a common formal framework. They
motivate a progression from exact optimization toward
approximation and reinforcement learning as explicit
state-space analysis becomes infeasible, without implying
automatic optimality, convergence, safety, or learning
guarantees. The applicability of corresponding optimization,
verification, and learning results requires preservation of
the relevant objective, information structure, and admissible
policy class.

\subsection{Contributions}

Our contributions in this paper include:
\begin{enumerate}
    \item A unified formalism for modeling controlled stochastic
    discrete-event and real-time systems using Controlled SANs,
    integrating concurrency, nondeterminism, probabilistic
    branching, stochastic timing, and policy-driven control.

    \item Formal automata-theoretic semantics for Controlled SANs,
    including controlled automata, controlled probabilistic
    automata, and controlled stochastic automata. At the
    nondeterministic level, the semantics distinguish structurally
    available control actions from behaviorally admissible
    controlled transitions.

    \item A characterization of the modeling power of computable
    controlled activity networks, including a constructive
    realization result for decidable, effectively finitely
    branching controlled automata with finite activity and
    control sets and a finite, nonempty initial-state set.
    Under these conditions, a well-behaved computable controlled
    activity network realizes the reachable behavior of the
    corresponding controlled automaton up to effective
    isomorphism.

    \item A classification of control policies by memory and
    computational structure, ranging from memoryless and
    finite-memory policies to stack- and tape-augmented
    policies, together with a study of their expressive power
    through associated language classes.

    \item Behavioral-equivalence notions, including bisimulation
    and stochastic isomorphism, supporting abstraction,
    compositional reasoning, and comparison of controlled
    models.

    \item A formal relationship between controlled probabilistic
    automata (CPAs) and conventional discrete-time Markov
    decision processes (DTMDPs): an activity-generation law
    and contingent-action semantics yield a DTMDP representation
    of a CPA, with direct correspondence for deterministic
    policies and the adopted product-form correspondence for
    randomized policies. Conversely, every conventional DTMDP
    admits a single-activity CPA representation.

    \item A formal relationship between the Markovian subclass
    and conventional continuous-time Markov decision processes
    (CTMDPs): every conventional CTMDP admits a single-activity
    controlled Markovian automaton representation, while a
    controlled Markovian automaton admits, under the stated
    assumptions, a CTMDP projection in which contingent-action
    maps serve as actions. The projection preserves the
    corresponding physical-state holding-time and successor-state
    laws, with randomized policy correspondence restricted to
    the adopted product-form subclass.

    \item A semantic basis for connecting Controlled SAN models
    with reward optimization and policy synthesis, including
    corresponding reward representations and methods ranging
    from exact dynamic programming to approximation and
    reinforcement learning, subject to preservation of the
    relevant objectives, information structures, and admissible
    policy classes.
\end{enumerate}

Together, these results establish Controlled SANs as a formal
foundation for the specification, verification, performance
analysis, optimization, and synthesis of controlled stochastic
discrete-event and real-time systems. The framework provides
a basis for reasoning about qualitative properties, including
safety, liveness, and controllability, as well as quantitative
objectives involving expected reward, risk, and deadline-related
performance. A principal motivation is the development of
learning-enabled controllers for safety-critical systems, where
control decisions must account for concurrency, uncertainty,
stochastic timing, and formally specified behavioral constraints.

The present paper focuses on formal and semantic foundations
rather than the development of a complete optimization,
reinforcement-learning, or verification toolchain. Nevertheless,
the framework provides a basis for future automated analysis
in which Controlled SAN models can support decision-process
generation, policy synthesis, simulation, and verification.
In particular, it motivates the integration of temporal-logic
and probabilistic model checking, approximate optimization,
reinforcement learning, and structure-aware policy synthesis.
An important direction is the development of learning-based
controllers whose performance and safety properties can be
analyzed or verified under explicitly stated assumptions,
including the establishment of statistical or PAC-style
guarantees where applicable.
\subsection{Organization of the Paper}
The rest of the paper is organized as follows. Section 2
discusses related work. Section 3 introduces the syntax and
semantics of Controlled SANs in a nondeterministic setting,
together with their interpretation through controlled automata.
Section 4 develops controlled probabilistic activity networks
and controlled probabilistic automata, policy types, language
classes, probabilistic semantics, and the relationship with
DTMDPs. Section 5 develops stochastic timing semantics,
controlled stochastic automata and their Markovian subclass,
the relationship with CTMDPs, reward models, and the connections
to optimization and learning. Section 6 summarizes the main
results and outlines directions for future research.
\section{Related Work}

The Controlled SAN framework developed in this paper brings
together several classical and modern strands of research in
stochastic modeling, formal verification, automata theory,
stochastic decision processes, optimization, and
learning-enabled control. In this section, we review the most
relevant foundational and applied models and clarify their
relationship to our approach.

\subsection{Petri Nets and SANs}

Stochastic Activity Networks (SANs)
\cite{movaghar1984stochastic,meyer1985stochastic,
movaghar2001stochastic} extend timed Petri nets by associating
activities with probabilistic timing and gate-based transition
logic. SANs have been extensively applied to performance and
dependability modeling of distributed and fault-tolerant
systems. However, their classical semantics do not provide the
explicit policy structure developed here for
\emph{control policies}---that is, policy-dependent selection
among alternative control actions based on system state,
history, or other available information. Controlled SANs extend
SANs by integrating control actions and policy-based semantics,
with the classical SAN model recovered as the special case in
which no nontrivial control choice is present.

A closely related and influential extension of Petri nets is
the model of \emph{Generalized Stochastic Petri Nets} (GSPNs)
\cite{Marsan1984,Balbo1997}. GSPNs enrich classical Petri nets
with both immediate (zero-time) and timed stochastic firings
and provide a well-established framework for performance and
dependability evaluation.

While GSPNs support stochastic timing, immediate transitions,
and priority mechanisms, their standard semantics do not
provide the explicit policy structure developed here for
Controlled SANs. Controlled SANs instead incorporate control
actions and policies directly into the semantic model,
allowing systematic study of controlled executions, policy
classes, behavioral equivalences, optimization, and policy
synthesis.

A more direct connection to decision-oriented stochastic net
models is provided by \emph{Markov Decision Petri Nets}
(MDPNs) and \emph{Markov Decision Well-formed Nets} (MDWNs)
\cite{BeccutiFranceschinisHaddad2007}. These formalisms provide
high-level Petri-net representations of distributed systems
combining probabilistic and nondeterministic behavior, with an
underlying connection to Markov decision processes. More
recently, \emph{Stochastic Decision Petri Nets} (SDPNs)
\cite{WittboldBernemannHeckelHeindelKonig2023} extend
stochastic Petri nets with rewards and controllable
transitions that may be deactivated by an external decision
maker, and admit a translation into MDPs.

These models are closely related in motivation to Controlled
SANs, since they also combine stochastic net behavior with
explicit decision making. The emphasis of the present
framework, however, is different. Controlled SANs are
developed together with controlled-automaton, controlled
probabilistic-automaton, and controlled stochastic-automaton
semantics in which control policies are treated as first-class
mathematical objects. This makes it possible to study
systematically policy classes with different memory and
computational capabilities, their induced language-theoretic
expressiveness and behavioral equivalences, and their
relationships with classical stochastic decision processes.
Accordingly, decision-oriented Petri-net models and Controlled
SANs address overlapping questions of stochastic control,
concurrency, and policy synthesis from complementary formal
perspectives.

The formal underpinnings of SANs are rooted in Petri-net
theory \cite{Peterson1981}, which has also been enriched
through control-oriented extensions such as \emph{inhibitor
arcs} and \emph{priority transitions}, and through the
application of \emph{supervisory control theory} to Petri nets
\cite{Ramadge1987,Cassandras2008}. These approaches provide
mechanisms for restricting, supervising, or controlling
Petri-net behavior. Controlled SANs take a complementary
approach by incorporating explicit control actions and
policy-dependent behavior into a unified activity-network and
automata-based semantic framework, enabling systematic
analysis of system evolution under different classes of
control policies.

\subsection{DTMDPs, CTMDPs, GSMPs, and GSMDPs}

Discrete-Time Markov Decision Processes (DTMDPs)
\cite{puterman1994markov} provide a standard framework for
sequential decision making under probabilistic state
transitions. Their established theory includes dynamic
programming, reward optimization, policy synthesis, and
learning-based methods. The relationship between DTMDPs and
the probabilistic level of the present framework is expressed
through controlled probabilistic automata (CPAs).

Given an activity-generation law, a CPA admits a DTMDP
representation in which an action at state $q$ is a complete
contingent-action map assigning an admissible control to each
enabled activity. The contingent action is selected first; the
next activity is then generated conditionally on the current
state and selected contingent action, and the control already
assigned to the generated activity is executed. The resulting
DTMDP transition kernel marginalizes over the generated
activity. Conversely, every conventional DTMDP admits a
single-activity CPA representation. Deterministic
contingent-action policies correspond directly. For randomized
policies, the CPA convention assumes conditionally independent
activity-specific assignments and therefore corresponds to the
product-form subclass of randomized DTMDP policies; an
unrestricted randomized DTMDP policy may be more general. Thus,
the relationship is not based on identifying the two formalisms
syntactically, but on an explicit semantic representation
preserving the corresponding state-process law under the stated
conditions and policy correspondence.

Continuous-Time Markov Decision Processes (CTMDPs)
\cite{puterman1994markov} provide a foundational model for
decision making under continuous-time Markovian dynamics.
Under the usual exponential-sojourn formulation, an action
selected at a decision epoch determines the holding-time and
successor-state laws, and established theory provides
optimization and verification methods for a variety of reward
and reachability objectives.

The relationship between CTMDPs and the Markovian subclass of
Controlled SANs is most naturally expressed through their
automata semantics. Every conventional CTMDP admits a
single-activity controlled Markovian automaton representation.
Conversely, under the Markovian assumptions developed in this
paper, a controlled Markovian automaton admits a conventional
CTMDP projection by taking a complete contingent-action map as
the CTMDP action. In the controlled Markovian automaton, this
map assigns an admissible control to each enabled activity
before the timed race; the assigned controls determine the
corresponding completion rates, and the control assigned to
the activity that completes is the one that is executed. The
CTMDP projection preserves the conditional holding-time and
successor-state laws at the physical-state level under
corresponding deterministic policies and, for randomized
policies, under the product-form correspondence induced by the
conditionally independent activity-specific assignments. An
unrestricted randomized CTMDP policy may therefore be more
general than the randomized contingent-action policy class
adopted for the controlled Markovian automaton. The projection
marginalizes over the identity of the completed activity.
Consequently, objectives or policies depending essentially on
activity identities, activity-labeled histories, or other
activity-specific information require a separate preservation
argument or an information-retaining state augmentation.

These DTMDP and CTMDP connections also provide an interface
between the formal semantics developed here and established
optimization methods. For sufficiently small explicitly
represented decision processes, classical dynamic-programming
methods such as value iteration and policy iteration may be
used under their standard assumptions. As explicit
state-space analysis becomes infeasible, approximate dynamic
programming, function approximation, sampling, simulation,
and reinforcement-learning methods provide possible
alternatives. When the decision-process model is available,
it may further support model-based reinforcement learning,
planning, and simulation. Controlled SANs retain additional
activity-network and compositional structure that may be
useful for state aggregation, feature construction, guided
exploration, and structure-aware policy synthesis. These
computational possibilities motivate the semantic connections
to decision processes but are not, by themselves, guarantees
of optimality, convergence, or scalability. Transfer of
optimization or learning results additionally requires
preservation of the relevant objective, information structure,
and admissible policy class.

Generalized Semi-Markov Processes (GSMPs)
\cite{Matthes1962,Schassberger1977} provide a broader
stochastic-process framework in which multiple concurrently
enabled events may have general, rather than necessarily
exponential, timing distributions and compete according to
race semantics. This is closely related to the timing
structure represented by Controlled SANs. In particular,
general non-exponential activity requirements in a Controlled
SAN may require the stochastic state to include residual
requirements, elapsed processing, activity status, or other
timing information. Comparisons between the two formalisms
therefore depend on common conventions concerning clock
retention and reactivation, simultaneous events,
instantaneous behavior, and state augmentation.

Generalized Semi-Markov Decision Processes (GSMDPs)
\cite{YounesSimmons2004} add decision making to generalized
semi-Markov dynamics and therefore provide an especially
relevant point of comparison. Both GSMDPs and Controlled SANs
can combine general stochastic timing with control. Their
control semantics, however, need not coincide: a GSMDP action
that determines or changes a future event-time law need not
have the same information structure or decision timing as the
pre-race contingent-action assignments used in the stochastic
semantics developed here. Consequently, neither a general
equivalence nor a subsumption relation follows merely from the
presence of general timing and control in both formalisms.
Precise comparisons require explicit conventions for
observable histories, admissible policies, clock semantics,
simultaneous events, and the quantities to be preserved.

From a modeling perspective, Controlled SANs retain the
activity-network structure of SANs, including places, gates,
concurrent activities, and explicit control actions, while
their controlled stochastic automata provide a state-based
semantic representation. GSMPs and GSMDPs instead provide
process-oriented frameworks for general event timing and, in
the latter case, decision making. These formalisms therefore
offer related but distinct ways of representing controlled
stochastic discrete-event systems.

\subsection{Controlled Automata and Language Expressiveness}

The study of controlled automata and their language-theoretic
properties under various policy types has a rich history in
formal methods, including supervisory control theory
\cite{Ramadge1987}, game-theoretic controller synthesis
\cite{Pnueli1989}, and probabilistic automata
\cite{Segala1995}. Our work builds on these foundations by
systematically classifying control policies (e.g.,
memoryless, finite-memory, stack-augmented, and
tape-augmented) and characterizing their expressive power in
terms of accepted languages over finite and infinite words.

The hierarchy of language families induced by these policies
connects naturally with classical distinctions in automata
theory, including regular, context-free, and computable
language classes, and with foundational results on
$\omega$-automata and B\"uchi acceptance
\cite{Thomas1990,Vardi1986Automata,Rabin1969}. The precise
language family obtained depends on the controlled model,
acceptance semantics, and policy class.

This perspective provides a formal way to study the trade-off
between controller resources and behavioral expressiveness.
Increasing controller memory may enlarge the set of behaviors
that can be realized or recognized, but it may also increase
the difficulty or undecidability of associated analysis and
synthesis problems.

\begin{remark}
The expressiveness--complexity relationship should be
understood relative to the particular model and decision
problem. Finite-memory, stack-augmented, and tape-augmented
policies provide increasingly rich computational mechanisms,
but complexity and decidability conclusions require the
corresponding assumptions on the plant, policy representation,
acceptance condition, and problem formulation.
\end{remark}

\subsection{Pushdown Automata: Probabilistic and Controlled Extensions}

Pushdown automata (PDAs) provide a classical foundation for
modeling infinite-state systems with stack-based control. A
well-studied subclass is \emph{Basic Process Algebra} (BPA),
which can be viewed as a special case of pushdown processes
with a structured representation of recursive behavior.

\emph{Probabilistic BPA} (pBPA) extends BPA by introducing
probability distributions over rewrite rules, thereby
inducing infinite-state stochastic processes with structured
stack discipline. These models have been extensively
investigated in the literature and admit decidability results
for several qualitative properties.

More general classes, such as \emph{stochastic BPA games}
\cite{brazdil2010qualitative} and \emph{probabilistic
pushdown automata} (pPDA)
\cite{etessami2005recursive,etessami2009recursive,
esparza2014probabilistic}, capture combinations of recursion,
nondeterminism, and probabilistic branching. Such models
provide important points of comparison for the
stack-augmented policy classes considered in this paper and
for the analysis of infinite-state controlled systems.

The Controlled SAN framework approaches unbounded controller
memory from a complementary direction. Rather than identifying
the stochastic plant itself with a pushdown process, we
distinguish the underlying controlled model from the
computational resources available to its policy. In
particular, stack- and tape-augmented policies make it
possible to study how increasing controller memory changes
language-theoretic expressiveness. The stochastic layer then
adds general timing semantics to this controlled-automata
framework. This perspective connects Controlled SANs with
probabilistic pushdown models without requiring a blanket
subsumption claim between formalisms having different state,
timing, and control semantics.

\subsection{Formal Verification, Optimization, and Learning}

Formal modeling and verification of safety-critical systems
has increasingly emphasized the integration of stochastic
timing, uncertainty, and control
\cite{seshia2022cyberphysical,Clarke1999}. Tools such as
PRISM \cite{Kwiatkowska2011}, MODEST \cite{Hahn2010}, and
Storm \cite{Dehnert2017} support important classes of
probabilistic and stochastic models and associated
verification and quantitative-analysis techniques. Controlled
SANs provide a complementary formalism that combines
activity-network structure, explicit policy-based control,
stochastic timing, and automata-theoretic semantics, with the
aim of supporting both qualitative and quantitative reasoning.

One of the most prominent tools for modeling and analyzing
SAN-based systems is the M\"obius Modeling Framework
\cite{Deavours2002Mobius}, which provides a compositional and
extensible environment for performance and dependability
analysis. M\"obius has significantly influenced the SAN
modeling community and provides relevant context for future
tool support and implementation of extensions such as
Controlled Stochastic Activity Networks.

The decision-process representations developed in this paper
suggest a broader tool architecture connecting formal modeling
with optimization and learning. A Controlled SAN model can
serve as a structured system model from which an appropriate
decision-process representation is derived and supplied to a
policy-synthesis engine. Depending on model size and the
available assumptions, such an engine could employ exact
dynamic programming, approximate dynamic programming,
model-based reinforcement learning, or other approximation and
learning techniques. The resulting policy could then be
evaluated through simulation, quantitative performance
analysis, and formal verification.

This direction is complementary to purely learning-based
approaches. The objective is not to replace formal semantics
with reinforcement learning, but to use the formal model as a
foundation for optimization, simulation, learning, and
verification. In very large systems, compositional and
activity-network structure may further guide approximation and
learning without requiring explicit enumeration of the entire
state space. Developing such scalable, structure-aware methods
and establishing their approximation, convergence, statistical,
or PAC-style guarantees remain directions for future work.
\section{Nondeterministic Models}
Controlled activity networks are a generalization of activity networks \cite{movaghar1984stochastic, movaghar2001stochastic}, incorporating control mechanisms in a nondeterministic setting.  
\subsection{Model Structure}
Throughout this paper, ${\mathcal N}$ denotes the set of natural
numbers and ${\mathcal R_{+}}$ represents the set of non-negative real
numbers.
\begin{definition} \label{def3.1}
A \emph{controlled activity network} is a 10-tuple
\[
K=(P,IA,TA,CA,\Gamma,IG,OG,IR,IOR,TOR),
\]
where:
\begin{itemize}

\item $P$ is a finite set of \emph{places};

\item $IA$ is a finite set of \emph{instantaneous activities};

\item $TA$ is a finite set of \emph{timed activities};

\item $CA$ is a finite set of \emph{control actions};

\item $\Gamma:\mathbb{N}^{P}\times TA\to 2^{CA}$
is the \emph{control-availability function}.
For each $\mu\in\mathbb{N}^{P}$ and $a\in TA$,
$\Gamma(\mu,a)$ specifies the set of control
actions available for activity $a$ under
the configuration $\mu$.
For each $a\in TA$, define
\[
\Gamma(a):=
\bigcup_{\mu\in\mathbb{N}^{P}}\Gamma(\mu,a),
\]
the set of all control actions associated
with activity $a$.

\item $IG$ is a finite set of \emph{input gates}.
Each input gate has a finite number of \emph{inputs}.
To each $G\in IG$ with $m$ inputs are associated a function
\[
f_G:\mathcal{N}^m\longrightarrow\mathcal{N}^m,
\]
called the \emph{function} of $G$, and a predicate
\[
g_G:\mathcal{N}^m\longrightarrow\{\mathrm{true},\mathrm{false}\},
\]
called the \emph{enabling predicate} of $G$;

\item $OG$ is a finite set of \emph{output gates}.
Each output gate has a finite number of \emph{outputs}.
To each $G\in OG$ with $m$ outputs is associated a function
\[
f_G:\mathcal{N}^m\longrightarrow\mathcal{N}^m,
\]
called the \emph{function} of $G$;

\item
\[
IR\subseteq
P\times\{1,\ldots,|P|\}\times IG\times(IA\cup TA)
\]
is the \emph{input relation}.
The relation $IR$ satisfies the following conditions:
\begin{itemize}

\item for any $(P_1,i,G,a)\in IR$, if $G$ has $m$ inputs,
then $i\leq m$;

\item for any $G\in IG$ with $m$ inputs and any
$i\in\{1,\ldots,m\}$,
there exist
$a\in IA\cup TA$ and $P_1\in P$ such that
\[
(P_1,i,G,a)\in IR;
\]

\item for any
\[
(P_1,i,G_1,a),(P_1,j,G_2,a)\in IR,
\]
we have $i=j$ and $G_1=G_2$.

\end{itemize}

\item
\[
IOR\subseteq IA\times OG\times\{1,\ldots,|P|\}\times P
\]
is the \emph{instantaneous output relation};

\item
\[
TOR\subseteq
TA\times CA\times OG\times\{1,\ldots,|P|\}\times P
\]
is the \emph{timed output relation}. 
For every $(a,c,G,i,P')\in TOR$,
we require $c\in\Gamma(a)$.

The relations $IOR$ and $TOR$ satisfy the following conditions:
\begin{itemize}

\item for any $(a,G,i,P_1)\in IOR$, if $G$ has $m$ outputs,
then $i\leq m$;

\item for any $(a,c,G,i,P_1)\in TOR$, if $G$ has $m$ outputs,
then $i\leq m$;

\item for any $G\in OG$ with $m$ outputs and any
$i\in\{1,\ldots,m\}$,
exactly one of the following holds:
\begin{itemize}
\item there exist $a\in IA$ and $P_1\in P$ such that
\[
(a,G,i,P_1)\in IOR;
\]
\item there exist $a\in TA$, $c\in CA$, and $P_1\in P$ such that
\[
(a,c,G,i,P_1)\in TOR.
\]
\end{itemize}

\item for any
\[
(a,G_1,i,P_1),(a,G_2,j,P_1)\in IOR,
\]
we have $i=j$ and $G_1=G_2$;

\item for any
\[
(a,c_1,G_1,i,P_1),(a,c_2,G_2,j,P_1)\in TOR,
\]
we have
\[
i=j,\qquad c_1=c_2,\qquad G_1=G_2.
\]

\end{itemize}

\end{itemize}
\end{definition}
Graphically, a controlled activity network is represented as follows.
\vspace{2ex} 
A place is depicted as \put(12,3){\circle{12}}\hspace{5ex}.
An instantaneous activity 
\vspace{2ex} 
is represented as \put(15,-3){\line(0,1){13}} \hspace{6ex}   
and a timed activity as \put(15,-3){\rule{1mm}{5mm}}  \hspace{6ex}. 
A control action is shown as \put(6,3){\circle{5}}\hspace{2ex} connected to the right side of a timed activity.
An input gate
\vspace{2ex}
with  $m$ inputs
is shown as
\\
\put(12,0){\line(1,0){10}} \put(12,10){\line(1,0){10}}
\put(5,2){.} \put(5,5){.} \put(5,8){.}
\put(5,-2){\tiny $m$} \put(5,10){\tiny $1$}
\put(10,2){.} \put(10,5){.} \put(10,8){.}
\put(22,-3){\line(0,1){16}}
\put(22,13){\line(3,-2){12}}
\put(22,-3){\line(3,2){12}}
\put(34,5){\line(1,0){5}}
\hspace{10ex}
and an output gate with  
\vspace{2ex} 
$m$ outputs
as
\put(18,-3){\line(0,1){16}}
\put(18,13){\line(3,-2){12}}
\put(18,-3){\line(3,2){12}}
\put(11,5){\line(1,0){7}}
\put(23,0){\line(1,0){12}} \put(23,10){\line(1,0){12}}
\put(33,2){.} \put(33,5){.} \put(33,8){.}
\put(38,-2){\tiny m} \put(38,10){\tiny 1}
\put(38,2){.} \put(38,5){.} \put(38,8){.} 
\hspace{12ex}.

Except for the inclusion of control actions, the structure of controlled activity networks is the same as the structure of activity networks \cite{movaghar2001stochastic}.  
As an example, consider the graphical representation 
of a controlled  
activity network as in Figure 1. 
$P1$, $P2$, $P3$, $P4$, $P5$, $P6$, $P7$, and $P8$ are
places.
$T1$, $T2$,  $T3$, 
and $T4$ are timed activities, and $I1$, and $I2$ are instantaneous activities.
$c1$, $c2$, $c3$, and $c4$ are control actions.
$G1$ is an input gate with only one input, and $G2$ is
an output gate with only one output.  The enabling predicates, if any, and functions 
of these gates are indicated in a table called 
a ``Gate Table''  as depicted in Figure 1. 

A directed line from a place to an activity represents a  special input gate with a single input and 
an enabling predicate $g$ and a function
$f$ such that 
$g(x) = true$, iff $x \geq 1$, and $f(x) = x -1$ (e.g.,
directed lines from $P1$ to $I1$ and from $P2$
to $I2$).  A directed line from a timed activity and a control action to a 
place represents a special output gate with a single output and 
a function $f$ such that 
$f(x) = x + 1$ (e.g., the directed lines from $T1$ and $c1$ to $P1$ and 
the directed lines from $T1$ and $c2$ to $P3$).
These special gates are referred to as ``standard''  
gates.

Consider a controlled activity network as in Definition ~\ref{def3.1}.
Suppose, $(P_{k}, k, G, a) \in IR$.
Then, in a graphical representation, place $P_{k}$
is linked to the $k$-th input of an input gate $G$ whose output is connected to activity $a$.
$P_{k}$ is said to be an {\it input place} of $a$ and $G$ 
is referred to as an {\it input gate} of $a$.
For example, let $IR$ be the input relation of the model of Figure 1.
Then, $(P4, 1, G2, T2) \in IR$, $P4$, $P5$, and $P6$
are input places of $T2$, and $G2$ is an input gate of $T2$. 
Similarly, suppose 
$(a, c, G, k, P_{k}) \in TOR$.
Then, in a graphical representation, the timed activity $a$ via control action $c$ is connected to the input of an output gate $G$ whose $k$-th output is linked to a place $P_k$.  $P_k$ is said to be 
an {\it output place} of the timed activity $a$, and gate $G$ is referred to as an {\it output gate} of $a$ via control action $c$.  
For example, in the model of Figure 1, $(T1, c1, G1, 1, P2) \in TOR$, that is, a timed activity $T1$ via control action $c1$ is linked to the 
input of an output gate $G1$ whose only output is 
connected to place $P2$.
$G1$ is an output gate of $T1$ via $c1$, and $P2$ is 
an output place of $T1$ via $c1$.
Note also that
$(T1, c1, G', 1, P1), (T1, c2, G'', 1, P3) \in TOR$, where
$G'$ and $ G''$ are standard output gates, which 
are the output gates of $T1$ via $c1$ and $c2$, respectively,
and that $P1$ and $P3$ are the output places of $T1$ via $c1$ and $c2$, respectively.
For the CAN illustrated in Figure~1,
the set of control actions is
\[
CA=\{c_1,c_2,c_3,c_4,c_d\},
\]
and the sets of structurally available
control actions are
\[
\Gamma(T_1)=\{c_1,c_2\},\qquad
\Gamma(T_2)=\{c_3,c_4\},
\]
\[
\Gamma(T_3)=\Gamma(T_4)=\{c_d\},
\]
where \(c_d\) is the unique control action associated with each of \(T_3\) and \(T_4\).
Note that
$T_3$ and $T_4$ have no output gates.

Like Petri nets, we have a notion of ``marking'' for controlled activity networks.
\begin{definition}
\label{def3.2}
Consider a controlled activity network as in Definition ~\ref{def3.1}.
A marking is a function
\[
\mu: \; P \longrightarrow {\mathcal N}.
\]
It is often convenient to characterize a marking $\mu$ as a vector, that is, 
$\mu = (\mu_{1}, \ldots , \mu_{n})$,
where $\mu_{i} = \mu (P_{i})$,                               
$i = 1, \ldots , n$, and $P = \{ P_{1}, \ldots , P_{n} \}$.
In a graphical representation, a marking is characterized by 
{\it tokens} (dots) inside places.
The number of tokens in a place represents the marking of that place
(e.g., marking $(0,0,0,1,1,1,0,0)$ in Figure 1).
An activity is ``enabled'' in a marking if the enabling predicates 
of its input gates are true in that marking.
\end{definition}
More formally, we have:
\begin{definition}
\label{def3.3}
Consider a controlled activity network as in Definition ~\ref{def3.1}.
$a \in (IA \cup TA)$ is {\it enabled} in a marking $\mu$ if
for any input gate $G$ of $a$ with $m$ inputs and an enabling predicate
$g_{G}$,
\[
g_{G} (\mu_{1}, \dots, \mu_{m}) = true,
\]
where $\mu_{k} = \mu (P_{k} )$,
for some $P_{k} \in P$ such that $(P_{k}, k, G, a) \in IR$,
$ k = 1, \dots , m$.
\end{definition}
An activity is {\it disabled} in a marking if it is not enabled in that
marking.  A marking is {\it stable} if no instantaneous 
activity is enabled in that marking.
A marking is {\it unstable} if it is not stable. For example,  $(0,,0,1,0,1,1,0,0)$ is an unstable marking but $(0,,0,0,1,1,1,0,0)$
is a stable marking in the model of Figure 1.  In the former marking, timed activity $T1$ and instantaneous activity $I3$ are enabled, but in the latter marking,  timed activities $T1$ and $T2$ are enabled.

We impose the following consistency condition
on the control-availability function $\Gamma$
of Definition~\ref{def3.1}: for every marking
$\mu\in\mathcal N^P$ and timed activity $a\in TA$,
\[
\text{if $\mu$ is stable and $a$ is enabled
at $\mu$, then }\Gamma(\mu,a)\neq\varnothing.
\]
\subsection{Model Behavior}
A controlled activity network 
with a 
marking is a dynamic system. 
A marking changes only if an 
activity {\it completes}.  In a stable marking, only one of the enabled 
timed activities is allowed to complete. 
When there is more than one enabled timed activity, the choice of which activity to complete first 
is done nondeterministically.
In an unstable marking, only one of the enabled 
instantaneous activities may complete
(i.e., enabled instantaneous activities have priority
over enabled timed activities for completion). 
When there is more than one enabled instantaneous activity, the choice 
of which activity to complete first is also done nondeterministically. 
When an activity completes, it may change the marking of its 
input and output places.  This change is governed by the functions of 
its input gates and output gates, and is done in two steps as 
follows.  First, the marking of its input places may change due to the 
functions of its input gates, resulting in an intermediary marking.
Next, in this latter marking, the marking of its output places
may also change due to the functions of 
its output gates, resulting in a
final marking after the completion of that activity. 
More specifically, consider a controlled activity network as in
Definition~\ref{def3.1}. Suppose an activity $a$ completes in a
marking $\mu$. The resulting marking $\mu'$ is determined in two
steps. First, an intermediary marking $\mu''$ is obtained from $\mu$
by applying the functions of the input gates of $a$. The marking
$\mu'$ is then determined according to whether $a$ is an
instantaneous or a timed activity. If $a$ is instantaneous, $\mu'$
is obtained from $\mu''$ by applying the functions of the output
gates of $a$. 
If $a$ is timed, a control action
$c\in\Gamma(\mu,a)$ is first chosen.
When more than one control action is
available, the choice is made according
to a given policy, to be defined later
in this section.
The marking $\mu'$ is then obtained from $\mu''$ by applying
the functions of the output gates of $a$ corresponding to the chosen
control action.

It is important to distinguish the enabledness
of an activity from the admissibility of a
controlled transition. The former is determined
by the enabling predicates of the activity's
input gates, whereas the latter is a behavioral
property of the realized controlled automaton.
Once an enabled timed activity completes under
a control action selected by a policy, its
corresponding gate functions are applied,
regardless of whether the subsequent
instantaneous-activity execution terminates.
Such an execution may reach a stable marking
or may diverge through infinitely many
instantaneous-activity completions.
These possibilities are formalized in the
realized controlled automaton of
Definition~\ref{def:can-realization}.

More formally, $\mu''$ and  $\mu'$ are defined as follows:
\begin{itemize}
\item for any $P_{1} \in P$ which is not an input or output place of $a$, 
\[
\mu'' (P_{1} ) = \mu'(P_{1}) = \mu (P_{1}) ,
\]
\item
for any input gate $G$ of $a$ with $m$ inputs and a 
function $f_{G}$,
\[
f_{G} (\mu_{1}, \ldots , \mu_{m}) = (\mu''_{1}, \ldots , \mu''_{m} ),
\]
where $\mu_{k} = \mu (P_{k} )$ and $\mu''_{k} = 
\mu'' (P_{k} )$ such that
$(P_{k}, k, G, a) \in IR$,
$k = 1, \ldots , m$, 
\item 
for any output gate $G$ of an instantaneous activity $a$ with $m$ outputs and a 
function $f_{G}$,
\[
f_{G} (\mu''_{1}, \dots , \mu''_{m} ) = (\mu'_{1}, \ldots , \mu'_{m} ), 
\]
where $\mu''_{k} = \mu'' (P_{k} )$ and
$\mu'_{k} = \mu' (P_{k} )$ such that 
$(a, G, k, P_{k}) \in IOR$, $k = 1, \ldots , m$.
\item 
for any output gate $G$ of a timed activity $a$ under a chosen control action $c$ with $m$ outputs and a 
function $f_{G}$,
\[
f_{G} (\mu''_{1}, \dots , \mu''_{m} ) = (\mu'_{1}, \ldots , \mu'_{m} ), 
\]
where $\mu''_{k} = \mu'' (P_{k} )$ and
$\mu'_{k} = \mu' (P_{k} )$ such that 
$(a, c, G, k, P_{k}) \in TOR$, $k = 1, \ldots , m$.
\end{itemize}

The above summarizes the behavior of a controlled activity network.  As an example, consider the model of Figure 1 with a marking 
$\mu = (0,0,0,1,1,1,0,0)$ and  available control actions $\Gamma(\mu, T1) = \{c1,c2\}.$
In this marking, the timed activities $T1$ and $T2$ are enabled.  Either of these
activities may complete.  
Suppose $T1$ completes first. Since
$\Gamma(\mu,T1)=\{c1,c2\}$, one of these
two control actions is chosen according
to the given policy. Suppose $c1$ is chosen next.
Then, $P1$ and $P2$ each gain a token, and the marking changes to $(1,1,0,1,1,1,0,0)$, as in Figure 2, which is an unstable marking.  
In this marking, timed activities $T1$ and $T2$ and instantaneous activities $I1$ and $I2$ are enabled.  Only $I1$ or $I2$ is allowed to be completed next.  Suppose $I2$ completes next.  The resulting marking will be $(1,0,0,1,2,1,0,0)$, as in Figure 3, which is an unstable marking.  In this marking, only $I1$ is allowed to complete, resulting in the stable marking $(0,0,0,2,2,1,0,0)$, as in Figure 4.

Following the above discussion, a marking $\mu'$ is said to be {\it reachable} from 
a marking $\mu$ under a string of activities $a_{1}  \dots a_{n}$, if the successive 
completion of $a_{1}, \ldots, a_{n}$ changes the marking of the network from $\mu$ to 
$\mu'$ under some possible choice of a sequence of control actions.
$\mu' $ is said to be {\it reachable} from $\mu$ if $\mu'$
is reachable from $\mu$ under a string of activities or 
$ \mu' = \mu$. 
For example, $(0,0,0,2,2,1,0,0)$ is reachable from  $(0,0,0,1,1,1,0,0)$ under
$T1 I2 I1$, and $(1,0,0,1,2,1,0,0)$ is reachable from $(0,0,0,1,1,1,0,0)$.
 \subsection{Semantic Models}
The behavior of a controlled activity network is concerned with how various stable markings are reached from each other due to 
the completion of timed activities. To study this behavior
more formally, we take advantage of the notion of a controlled automaton.
\begin{definition}  \label{def3.4}
A controlled automaton is a tuple
\[
\mathcal{A}=(Q,\Sigma,C,\delta,Q_0),
\]
where
\begin{itemize}
    \item $Q$ is the set of {\it states};
    \item $\Sigma$ is the {\it observable activity} alphabet;
    \item $C$ is the nonempty set of {\it control actions};
    \item $\delta\subseteq Q\times\Sigma\times C\times Q$
          is the {\it controlled transition} relation;
    \item $Q_0\subseteq Q$ is the set of {\it initial states}.
\end{itemize}
For $(q,a,c,q')\in\delta$, we also write
\[
q\xrightarrow{a,c}q'.
\]
\end{definition}
In this case, we say that state $q'$ is immediately reachable from
state $q$ under observable activity $a\in\Sigma$ and control action
$c\in C$, and that $a$ is enabled in state $q$.

For $q_i\in Q$, $a_i\in\Sigma$, and $c_i\in C$, $0\leq i<n$,
\[
q_0 \xrightarrow{a_0,c_0} q_1
\xrightarrow{a_1,c_1}\cdots
\xrightarrow{a_{n-1},c_{n-1}} q_n
\]
is called a finite execution of $\mathcal A$ of length $n$.

The sequence
\[
(q_0,a_0)(q_1,a_1)\cdots(q_{n-1},a_{n-1})q_n
\]
is called the corresponding finite run.

For each $q\in Q$, $a\in\Sigma$, and $c\in C$,
define the successor set
\[
\operatorname{Succ}(q,a,c)
=
\{q'\in Q:(q,a,c,q')\in\delta\}.
\]
A controlled transition labeled by $(a,c)$
is said to be admissible from state $q$
if and only if
\[
\operatorname{Succ}(q,a,c)\neq\varnothing.
\]
Thus, if $\operatorname{Succ}(q,a,c)=\varnothing$,
no transition labeled by $(a,c)$ is
admitted from $q$.

For each state $q\in Q$ and observable activity
$a\in\Sigma$, define the set of admissible
control actions by
\[
C(q,a)=
\{c\in C:\operatorname{Succ}(q,a,c)\neq\varnothing\}.
\]
Whenever a control policy selects an action for activity \(a\) at state \(q\), 
it must select an action from \(C(q,a)\).
Consequently, a control action $c$ for which
$\operatorname{Succ}(q,a,c)=\varnothing$
cannot be selected for that activity at that state.

\subsection{Bisimulation}
Next, we define a notion of equivalence for controlled automata based on the concept of bisimulation, which is similar to the one proposed for automata \cite{Milner1989,baier2008}.

\begin{definition} \label{def3.5}
Let
\[
\mathcal{A}=(Q,\Sigma,C,\delta,Q_0)
\quad\text{and}\quad
\mathcal{A}'=(Q',\Sigma,C,\delta',Q_0')
\]
be two controlled automata with the same observable
activity alphabet $\Sigma$ and set of control actions $C$.

The controlled automata $\mathcal{A}$ and $\mathcal{A}'$
are said to be equivalent if there exists a relation
\[
\gamma\subseteq Q\times Q'
\]
such that:

\begin{enumerate}
\item For every $q\in Q$, there exists $q'\in Q'$
such that $(q,q')\in\gamma$, and for every
$q'\in Q'$, there exists $q\in Q$ such that
$(q,q')\in\gamma$.

\item For every $q_0\in Q_0$, there exists
$q_0'\in Q_0'$ such that $(q_0,q_0')\in\gamma$,
and for every $q_0'\in Q_0'$, there exists
$q_0\in Q_0$ such that $(q_0,q_0')\in\gamma$.

\item Whenever $(q,q')\in\gamma$, the following
conditions hold for every $a\in\Sigma$ and
$c\in C$:
\begin{enumerate}
\item If $(q,a,c,r)\in\delta$, then there exists
$r'\in Q'$ such that $(q',a,c,r')\in\delta'$
and $(r,r')\in\gamma$.

\item If $(q',a,c,r')\in\delta'$, then there exists
$r\in Q$ such that $(q,a,c,r)\in\delta$
and $(r,r')\in\gamma$.
\end{enumerate}
\end{enumerate}

Such a relation $\gamma$ is called a
\emph{bisimulation} between $\mathcal{A}$ and
$\mathcal{A}'$.

The two controlled automata are \emph{isomorphic}
if there exists a bijection $f:Q\to Q'$ that maps
$Q_0$ onto $Q_0'$ and preserves and reflects
controlled transitions, that is,
\[
(q,a,c,r)\in\delta
\iff
(f(q),a,c,f(r))\in\delta'
\]
for all $q,r\in Q$, $a\in\Sigma$, and $c\in C$.
\end{definition}
\begin{proposition}
Let $\mathcal{E}_{\mathcal{A}}$ be a relation on the set of all
controlled automata such that
\[
(\mathcal{A}_1,\mathcal{A}_2)\in\mathcal{E}_{\mathcal{A}}
\]
if and only if $\mathcal{A}_1$ and $\mathcal{A}_2$ are equivalent
controlled automata in the sense of Definition~\ref{def3.5}.
Then $\mathcal{E}_{\mathcal{A}}$ is an equivalence relation.
\end{proposition}
\begin{proof}
Reflexivity follows from the identity relation on the state set of
each controlled automaton. Symmetry follows directly from the
symmetry requirement in Definition~\ref{def3.5}.

For transitivity, suppose that $\mathcal{A}_1$ is equivalent to
$\mathcal{A}_2$ under a bisimulation $\gamma_{12}$ and that
$\mathcal{A}_2$ is equivalent to $\mathcal{A}_3$ under a
bisimulation $\gamma_{23}$. Define
\[
\gamma_{13}
=
\{(q_1,q_3)\mid
\text{there exists }q_2\text{ such that }
(q_1,q_2)\in\gamma_{12}
\text{ and }
(q_2,q_3)\in\gamma_{23}\},
\]
together with its symmetric counterpart.
The state-coverage and initial-state conditions follow from those
of $\gamma_{12}$ and $\gamma_{23}$. Moreover, any transition from a
state $q_1$ can first be matched in $\mathcal{A}_2$ using
$\gamma_{12}$ and then in $\mathcal{A}_3$ using $\gamma_{23}$,
with the same observable activity and control action. The converse
direction follows symmetrically. Hence $\gamma_{13}$ is a
bisimulation between $\mathcal{A}_1$ and $\mathcal{A}_3$.
Therefore $\mathcal{E}_{\mathcal{A}}$ is reflexive, symmetric, and
transitive.
\end{proof}
\begin{proposition} \label{3.7}
Let
\[
\mathcal{A}=(Q,\Sigma,C,\delta,Q_0)
\qquad\text{and}\qquad
\mathcal{A}'=(Q',\Sigma,C,\delta',Q'_0)
\]
be equivalent controlled automata under a bisimulation $\gamma$.
Then every finite execution
\[
q_0 \xrightarrow{a_0,c_0} q_1
\xrightarrow{a_1,c_1} \cdots
\xrightarrow{a_{n-1},c_{n-1}} q_n
\]
of $\mathcal{A}$ has a corresponding finite execution
\[
q'_0 \xrightarrow{a_0,c_0} q'_1
\xrightarrow{a_1,c_1} \cdots
\xrightarrow{a_{n-1},c_{n-1}} q'_n
\]
of $\mathcal{A}'$ such that
\[
(q_i,q'_i)\in\gamma,
\qquad 0\leq i\leq n,
\]
and conversely.
\end{proposition}
\subsection{Operational Semantics}
We are now in a position to formalize the notion of 
the behavior of a controlled activity network as follows.
\begin{definition}[Realization of a CA] \label{def:can-realization}
Let $(K,\mu_0)$ be a controlled activity network with initial marking
$\mu_0$, where $K$ is defined as in Definition~\ref{def3.1}.
The pair $(K,\mu_0)$ is said to \emph{realize} the controlled automaton
\[
\mathcal{A}_K=(Q,\Sigma,C,\delta,Q_0),
\]
where:
\begin{itemize}
\item $Q$ consists of all stable markings of $K$,
together with a divergence state $\Delta$ if an
infinite sequence of instantaneous activities can
be completed from some marking of $K$;

\item $\Sigma=TA_K$, the set of timed activities of $K$;

\item $C=CA_K$, the set of control actions of $K$;

\item for $\mu,\mu'\in Q$, $a\in\Sigma$, and $c\in\Gamma_K(\mu,a)$,
\[
(\mu,a,c,\mu')\in\delta
\]
if and only if, starting from the stable marking $\mu$, the timed
activity $a$ completes under the chosen control action $c$, resulting
in a marking from which the stable marking $\mu'$ is reached through
a possibly empty finite sequence $x$ of instantaneous activities.
The completion of $a$ consists of applying its input-gate functions,
choosing the control action $c$, and then applying the output-gate
functions corresponding to $c$.

Moreover,
\[
(\mu,a,c,\Delta)\in\delta
\]
if and only if the completion of $a$
under the chosen control action $c$
results in a marking from which an
infinite sequence of instantaneous
activities can be completed;

\item $Q_0$ consists of all stable markings of $K$
reachable from $\mu_0$ through a possibly empty
finite sequence of instantaneous activities,
together with $\Delta$ if an infinite sequence
of instantaneous activities can be completed
from $\mu_0$.
\end{itemize}
For each state $q$, activity $a$, and control
action $c$, a controlled transition labeled
by $(a,c)$ is admissible at $q$ precisely when
\[
\operatorname{Succ}(q,a,c)\neq\varnothing.
\]
This is a behavioral notion, as in
Definition~\ref{def3.4}, and does not impose an
additional enabling condition on the timed
activity or restrict the execution of its
associated gate functions. In particular,
a completed timed activity may be followed
either by a finite sequence of instantaneous
completions reaching a stable marking or
by an infinite sequence represented by
$\Delta$, as specified above.

This admissibility condition is part of the
operational semantics of controlled transitions.
Consequently, an empty successor set does not
give rise to an observable transition. The admissibility of controlled transitions
is understood as in Definition~\ref{def3.4}.

\end{definition}
The above definition implies a notion of equivalence for controlled
activity networks as follows.
\begin{definition} \label{def3.7}
Two controlled activity networks with initial markings
$(K_1,\mu_{01})$ and $(K_2,\mu_{02})$ are \emph{equivalent} if they
realize equivalent controlled automata.
\end{definition}

\begin{proposition}
Let $\mathcal{E}_K$ be the relation on the set of all controlled
activity networks with initial markings defined by
\[
\bigl((K_1,\mu_{01}),(K_2,\mu_{02})\bigr)\in\mathcal{E}_K
\]
if and only if $(K_1,\mu_{01})$ and $(K_2,\mu_{02})$ are equivalent
in the sense of Definition~\ref{def3.7}. Then $\mathcal{E}_K$ is an
equivalence relation.
\end{proposition}

\subsection{Modeling Power}
The following concepts help specify the modeling power
of controlled activity networks.

\begin{definition}[Computable controlled activity network]
\label{Computable controlled activity network}
A controlled activity network is said to be
\emph{computable} if the enabling predicates
and functions of all its input gates and the
functions of all its output gates are total
computable functions.
In addition, the control-availability function
$\Gamma$ is required to be computable, in the
sense that membership
$c\in\Gamma(\mu,a)$ is decidable uniformly
in $\mu$, $a$, and $c$.

A computable controlled activity network may
admit infinite sequences of instantaneous-activity
firings. Such sequences constitute
\emph{instantaneous divergence}, interpreted
according to Definition~\ref{def:can-realization}.

A controlled activity network is said to be
\emph{instantaneous-divergence-free} if no
infinite sequence of consecutive
instantaneous-activity firings is possible
from any marking of the network.
\end{definition}

We formalize termination of instantaneous-activity
sequences along executions from a designated
initial marking through the following notion
of well-behavedness.
\begin{definition}[Well-behaved controlled activity network]
\label{def:well-behaved-can}
Let $K$ be a controlled activity network with initial marking
$\mu_0$. The initialized network $(K,\mu_0)$ is said to be
\emph{well behaved} if no infinite sequence of instantaneous
activities can be completed from $\mu_0$ or from any marking
reachable immediately after a feasible timed-activity completion
from a reachable stable marking. 
Equivalently, every maximal sequence of
consecutive instantaneous-activity completions
starting from such a marking is finite and
terminates at a stable marking.
\end{definition}
Instantaneous-divergence-freeness is a global
property, whereas well-behavedness is relative
to the designated initial marking. Thus,
instantaneous-divergence-freeness implies
well-behavedness, but not conversely.

\begin{definition}[Recursively enumerable and decidable
controlled automata]
\label{def:computability-ca}
A controlled automaton
$A=(Q,\Sigma,C,\delta,Q_0)$ is said to be
\emph{recursively enumerable} (r.e.) if
$Q$ is an effectively encoded countable set,
and both its initial-state set $Q_0$
and its controlled transition relation
\[
\delta\subseteq Q\times\Sigma\times C\times Q
\]
are recursively enumerable.

The controlled automaton $A$ is said to be
\emph{decidable} if both $Q_0$ and $\delta$
are decidable.

A controlled automaton is
\emph{effectively finitely branching} if,
for every $(q,a,c)\in Q\times\Sigma\times C$,
the finite successor set
\[
\operatorname{Succ}(q,a,c)
=
\{q'\in Q:(q,a,c,q')\in\delta\}
\]
can be effectively computed.
\end{definition}

Although the enabling predicates of a computable
controlled activity network are decidable, the
enabledness of an observable activity in its realized
controlled automaton need not be decidable. Indeed,
determining whether an observable transition is
reachable after finitely many instantaneous-activity
firings may require an unbounded computation.

\begin{lemma}[Computability properties of realized controlled automata]
\label{lem:finite-execution-transitions}
Let $(K,\mu_0)$ be a computable controlled activity network
with initial marking $\mu_0$, and let
\[
A_K=(Q_K,\Sigma,C,\delta_K,Q_{0,K})
\]
be its realized controlled automaton in the sense of
Definition~\ref{def:can-realization}. Then:
\begin{enumerate}

\item The finite-completion transition relation
\[
\delta_K^{\mathrm{fin}}
=
\{(\mu,a,c,\mu')\in\delta_K :
\mu \text{ and } \mu' \text{ are stable markings}\}
\]
and the initial-state set $Q_{0,K}$ are
recursively enumerable.

\item If $K$ is additionally instantaneous-divergence-free,
then $A_K$ is decidable and effectively finitely branching.
In particular, for every
$(\mu,a,c)\in Q_K\times\Sigma\times C$, the successor set
\[
\operatorname{Succ}_K(\mu,a,c)
=\{\mu'\in Q_K:(\mu,a,c,\mu')\in\delta_K\}
\]
is finite and effectively computable.
\end{enumerate}
\end{lemma}

\begin{proof}
By Definitions~\ref{def3.1} and~\ref{def3.2}, $K$ has a finite set of places
and an effectively encoded countable marking space.
Since all gate predicates and functions are total computable
by Definition~\ref{Computable controlled activity network}, 
the enabling conditions of individual
activities are decidable, and the marking resulting from
any enabled activity completion under a specified control
is effectively computable. Moreover, stability of a marking
is decidable, because $K$ contains finitely many instantaneous
activities with decidable enabling conditions.

We first establish recursive enumerability of
$\delta_K^{\mathrm{fin}}$. Enumerate all stable markings $\mu$
and all pairs $(a,c)\in\Sigma\times C$. 
Whenever the timed activity $a$ is enabled at $\mu$
and $c\in\Gamma(\mu,a)$, compute the marking $\nu$
resulting from its completion under $c$, and systematically
enumerate all finite admissible sequences of instantaneous
activities starting from $\nu$. Whenever such a sequence
reaches a stable marking $\mu'$, enumerate
$(\mu,a,c,\mu')$. Each finite sequence can be checked
effectively, and every finite-completion transition has
such a finite witness. Thus $\delta_K^{\mathrm{fin}}$ is
recursively enumerable. This argument does not require
instantaneous-divergence-freeness.

For any marking $\nu$, let
\[
\operatorname{Stab}_K(\nu)
=
\left\{
\mu\in\mathbb{N}^{P} :
\mu \text{ is stable and }
\nu\xrightarrow{x}\mu
\text{ for some }x\in IA^*
\right\},
\]
where $P$ is the set of places of $K$ and
$\xrightarrow{x}$ denotes an admissible sequence
of instantaneous-activity firings, possibly empty.

Similarly, the initial-state set
$Q_{0,K}=\operatorname{Stab}_K(\mu_0)$
is recursively enumerable, since all finite
sequences of instantaneous-activity firings
from $\mu_0$ can be effectively enumerated,
and stability of the resulting markings is
decidable.

Now suppose that $K$ is instantaneous-divergence-free.
Definition~\ref{def:can-realization} then introduces no
divergence state $\Delta$. Consequently,
\[
Q_K=\{\mu:\mu\text{ is a stable marking of }K\}.
\]
Thus $Q_K$ is an effectively encoded countable set with
decidable membership.

For any marking $\nu$, consider the rooted tree whose
nodes represent finite admissible sequences of
instantaneous-activity firings starting from $\nu$.
A node is extended by one child for each instantaneous
activity enabled at its final marking. Because $K$
contains finitely many instantaneous activities, this
tree is finitely branching. Furthermore, since $K$ is
instantaneous-divergence-free, the tree has no infinite
branch.

By K\"onig's lemma \cite{Boolos2007}, every such tree is finite.
Since the enabling conditions and marking updates are
computable, the entire tree can be effectively explored,
and the exploration terminates.

By the preceding argument,
$\operatorname{Stab}_K(\nu)$ is finite and
effectively computable for every marking $\nu$.

By Definition~\ref{def:can-realization}, the initial-state
set is
\[
Q_{0,K}=\operatorname{Stab}_K(\mu_0).
\]
Since this set is finite and effectively computable,
membership in $Q_{0,K}$ is decidable.

It remains to establish decidability of the controlled
transition relation $\delta_K$. Consider an arbitrary tuple
\[
(\mu,a,c,\mu')\in Q_K\times\Sigma\times C\times Q_K.
\]
First, determine whether the timed activity
$a$ is enabled at $\mu$ and whether
$c\in\Gamma_K(\mu,a)$.
If either condition fails, then
$(\mu,a,c,\mu')\notin\delta_K$.
Otherwise, compute the marking
$\nu_{a,c}(\mu)$ resulting from the
completion of $a$ under $c$, using the
corresponding input- and output-gate
functions. By Definition~\ref{def:can-realization},
\[
(\mu,a,c,\mu')\in\delta_K
\quad\Longleftrightarrow\quad
\mu'\in\operatorname{Stab}_K(\nu_{a,c}(\mu)).
\]
Since $\operatorname{Stab}_K(\nu_{a,c}(\mu))$
is finite and effectively computable,
membership in $\delta_K$ is decidable.
Moreover, these sets give the complete
finite successor sets effectively.

Therefore $Q_K$ is effectively encoded and countable,
membership in $Q_{0,K}$ and in $\delta_K$ is decidable, and
every labeled successor set is finite and effectively
computable. Hence $A_K$ is decidable and effectively
finitely branching in the sense of
Definition~\ref{def:computability-ca}.
\end{proof}
By Lemma~\ref{lem:finite-execution-transitions}, every controlled automaton effectively
isomorphic to one realized by an instantaneous-divergence-free
computable controlled activity network is decidable and
effectively finitely branching. Moreover, its initial-state
set is finite and nonempty. The following theorem establishes
a converse realization result, up to effective isomorphism
of reachable restrictions.
\begin{theorem}[Modeling power of computable
controlled activity networks]
\label{thm:modeling-power}
Let $A=(Q,\Sigma,C,\delta,Q_0)$ be a
decidable and effectively finitely branching
controlled automaton, where $\Sigma$ and
$C$ are finite and $Q_0$ is finite and
nonempty.

Then there exists a computable controlled
activity network $K$ with initial marking
$\mu_0$ such that $(K,\mu_0)$ is well
behaved and the reachable restriction of
the controlled automaton realized by
$(K,\mu_0)$ in the sense of Definition~\ref{def:can-realization}
is effectively isomorphic to the reachable
restriction of $A$.
\end{theorem}
\begin{proof}
We begin by considering a restricted class of activity networks
\cite{movaghar2001stochastic} consisting only of
\emph{instantaneous activities}, \emph{standard gates}, and a special
class of gates known as \emph{inhibitor gates}.
An inhibitor gate is an input gate with an enabling predicate $g$
and the identity function as its input function, where
$g(x)=\mathit{true}$ if and only if $x=0$.

This class of activity networks is equivalent, at the structural level,
to the class of \emph{extended Petri nets} with inhibitor arcs
\cite{Peterson1981}.
Such nets strictly generalize ordinary Petri nets and form the basis of
several expressive stochastic Petri net models, including
Generalized Stochastic Petri Nets (GSPNs)
\cite{Marsan1984,Balbo1997} when timing is added.

Extended Petri nets with inhibitor arcs can simulate Turing-complete
models of computation. Via register machines \cite{Shepardson1963}
or Minsky's program machines \cite{Minsky1967}, they can simulate
Turing-machine computations \cite{Hack1975,Peterson1981,Sipser2013}.
We use this fact only for the terminating computational procedures
specified below; in particular, each computational branch used in the
construction is required to terminate.

Let $A=(Q,\Sigma,C,\delta,Q_0)$ satisfy the hypotheses of the theorem.
Fix an effective injective encoding $\iota:Q\to\mathbb N$ of its
states.  
We choose an effective state encoding
$\iota$ with decidable range and
computable decoding on valid codes.
Consequently, whether a marking is
canonical is decidable by examining
its state code and auxiliary places.
The designated state place $P_A$ is used to
encode states of $A$, with $\iota(q)$
representing state $q\in Q$.
The encoding and its decoding on valid state
codes are taken to be effective, as required by the chosen effective
presentation of $Q$.
Let $\mu$ denote an arbitrary marking of $K$,
where $\mu(P_A)$ is the number of tokens in
the designated state place $P_A$.
For each $a\in\Sigma$ and $c\in C$, put
\[
\operatorname{Succ}(q,a,c)
=\{q'\in Q:(q,a,c,q')\in\delta\},
R_{ac}
=\{(\iota(q),\iota(q')):(q,a,c,q')\in\delta\}.
\]
As in Definition~\ref{def3.4}, for each $q\in Q$
and $a\in\Sigma$, let
\[
C(q,a)=
\{c\in C:\operatorname{Succ}(q,a,c)\neq\varnothing\}.
\]
For each $a\in\Sigma$, define
\[
G_a(\mu)=
\begin{cases}
\mathrm{true},
& \text{if } \mu(P_A)=\iota(q)
  \text{ for some }q\in Q
  \text{ and }
  \exists c\in C:
  \operatorname{Succ}(q,a,c)\neq\varnothing,\\
\mathrm{false},
& \text{otherwise}.
\end{cases}
\]
 Finally, let
\[
R_{Q_0}=\{(1,\iota(q_0)):q_0\in Q_0\}.
\]
Because $C$ is finite and the successor sets are effectively
computable and finite, $G_a$ is a total computable enabling predicate.
Thus $G_a$ is implemented directly by the triangular input gate of
timed activity $a$ in Figure~5. 

For each $(a,c)\in\Sigma\times C$, choose a
terminating procedure which, on a valid input
state code $\iota(q)$, computes the complete
finite set $\operatorname{Succ}(q,a,c)$ and,
if this set is nonempty, nondeterministically
selects an element
$q'\in\operatorname{Succ}(q,a,c)$.
If the set is empty, the controlled transition
labeled by $(a,c)$ is inadmissible at $q$,
in accordance with Definition~\ref{def3.4}, and is
excluded from the operational behavior.
The corresponding instantaneous computational
modules are denoted by $K_{R_{ac}}$.
The nondeterministic selection is exhaustive:
every state in $\operatorname{Succ}(q,a,c)$
can be selected along some terminating
instantaneous firing sequence, and no state
outside this set can be selected.

A separate terminating initialization module
$K_{Q_0}$ nondeterministically selects an
element $q_0$ of the finite nonempty set $Q_0$.
The selection is exhaustive: every $q_0\in Q_0$
can be selected along some terminating
instantaneous firing sequence. 
Upon completion, the module stores $\iota(q_0)$
in $P_A$ and clears its auxiliary places,
producing the \emph{canonical} stable marking
$\widehat{\mu}_{q_0}$.
These computational modules may be implemented
using standard and inhibitor gates. They are
arranged with disjoint internal places and
activities and communicate only through their
designated interface places.

We construct a controlled activity network $K$ with the set of places $P$, timed activities
$\Sigma$, and control actions $C$, and an initial marking $\mu_0$. 
For each timed activity $a\in\Sigma$ and
marking $\mu\in\mathbb{N}^{P}$, define
\[
\Gamma(\mu,a)=
\begin{cases}
C(q,a),
& \text{if }\mu(P_A)=\iota(q)
  \text{ for some }q\in Q,\\
\varnothing,
& \text{otherwise}.
\end{cases}
\]
Thus, for every canonical stable marking
$\widehat{\mu}_q$,
\[
\Gamma(\widehat{\mu}_q,a)=C(q,a).
\]
Moreover, by the definition of the enabling
predicate $G_a$, whenever a stable marking
$\mu$ enables $a$, we have
$\Gamma(\mu,a)\neq\varnothing$,
as required by the control-availability consistency
condition following Definition~3.3.
Finally,
\[
\Gamma(a)=
\bigcup_{\mu\in\mathbb{N}^{P}}\Gamma(\mu,a)
=\bigcup_{q\in Q}C(q,a).
\]
For each control action
$c \in \Gamma(\widehat{\mu}_q,a)=C(q,a)$
at a canonical stable marking $\widehat{\mu}_q$
encoding $q$, Figure~5 depicts the associated
transition subnetwork $K_{R_{ac}}$.
The state code is stored in $P_A$. When a
timed activity $a$ completes under $c$, its output gate deposits a
control token in $P_{ac}$, activating $K_{R_{ac}}$.
For a nonempty successor set, this module terminates after selecting
exactly one $q'\in\operatorname{Succ}(q,a,c)$ and updates $P_A$ to
$\iota(q')$. Once its computation and state update have finished, the
module deposits one token in $P_{\mathrm{done}}$.

The token in $P_{\mathrm{done}}$ enables the instantaneous cleanup
activity $\tau$. Its standard input consumes that token, and its
triangular output gate $H$ clears the temporary computational and
synchronization places while preserving the state code in $P_A$.
The module and cleanup are synchronized so that $\tau$ is enabled
only after the selected successor has been fully written, and no
other instantaneous activity can interfere with cleanup. The gate
function $H$ is a total computable function of the marking, so its
reset is performed in a single instantaneous completion. 
The output gate $H$ of the cleanup activity
$\tau$ has all temporary computational and
synchronization places in its output interface;
its total computable output function resets
their markings to zero while leaving the
marking of $P_A$ unchanged.

Initialization is handled by the subnetwork $K_{Q_0}$ in Figure~6.
Initially, $P_0$ contains one token and $P_A$ contains zero tokens.
The initialization module consumes the token in $P_0$, selects some
$q_0\in Q_0$, and writes $\iota(q_0)$ to $P_A$. It then signals
completion through $P_{\mathrm{done}}$, whereupon $\tau$ and $H$
clear the initialization auxiliaries. During initialization, all
transition modules remain disabled. Conversely, each $q_0\in Q_0$
can be selected.

\emph{Synchronization of computational modules.}
The computational subnetworks are equipped with
disjoint internal places and instantaneous activities.
Each module has a designated activation place, and
its initial instantaneous activity is enabled only
when that place contains an activation token.
The activation token is consumed upon entry into
the module and cannot be regenerated during the
same instantaneous phase.

The initialization module $K_{Q_0}$ and the
successor-computation modules $K_{R_{ac}}$ are
coordinated so that at most one computational
module is active at any time. All internal
instantaneous activities of an inactive module
are disabled by their enabling predicates.
Furthermore, the cleanup activity $\tau$ is
enabled only after the active module has
completed its computation and deposited a
token in $P_{\mathrm{done}}$.

The output gate $H$ of $\tau$ removes all
temporary computational and synchronization
tokens while preserving the state code in
$P_A$. Consequently, no computational module
can interfere with another, and no module can
be reactivated before a subsequent timed
completion.

For each $q\in Q$, define the canonical marking
\[
\widehat\mu_q(P)=
\begin{cases}
\iota(q),&P=P_A,\\
0,&\text{otherwise}.
\end{cases}
\]

By construction, every instantaneous activity
requires either an initialization token, an
active computational control token, or a
completion token in $P_{\mathrm{done}}$.
All such tokens are absent from a canonical
marking $\widehat{\mu}_q$.
Therefore every canonical marking is stable.

The preceding initialization produces exactly
$\{\widehat\mu_{q_0}:q_0\in Q_0\}$ as its stable outcomes. Every
successful transition computation, followed by cleanup, produces a
canonical stable marking $\widehat\mu_{q'}$.

The construction ensures that every reachable
stable marking at which a timed activity is
enabled is canonical. During the execution
of a successor-computation or cleanup module,
the marking is unstable so that no timed activity
can complete. Each such computation terminates
after finitely many instantaneous completions
and restores a canonical stable marking.
Consequently, the control-availability
function is nonempty whenever a timed
activity is enabled at a reachable stable
marking.  The computational modules are constructed so that
every reachable intermediate marking has at least
one enabled instantaneous activity, and no
computational branch can terminate without
producing its completion signal.

\emph{Control-admissibility requirement.}
By the construction of $\Gamma$, for every
canonical marking $\widehat{\mu}_q$ and timed
activity $a\in\Sigma$,
\[
\Gamma(\widehat{\mu}_q,a)=C(q,a).
\]
Consequently, a timed completion under control
$c$ is permitted at $\widehat{\mu}_q$ if and
only if $\operatorname{Succ}(q,a,c)$ is
nonempty. Thus, every activated
successor-computation module has a nonempty
finite successor set, independently of the
choice of control policy.

\emph{Termination of instantaneous computations.}
Every instantaneous phase consists of
the execution of at most one computational
module, followed by exactly one cleanup
completion. The deterministic computation
terminates, and the subsequent finite
nondeterministic selection has only
terminating branches. The synchronization
protocol prevents any other module from
being activated during this phase.

Consequently, every maximal sequence of
consecutive instantaneous completions
from the initial marking or following an
admissible timed completion is finite
and ends at a canonical stable marking.
Therefore $(K,\mu_0)$ is well behaved
in the sense of Definition~\ref{def:well-behaved-can}.

Each successor-computation module $K_{R_{ac}}$
is implemented in two stages. In the first
stage, a deterministic terminating procedure
computes the complete finite set
\[
S=\operatorname{Succ}(q,a,c).
\]
In the second stage, provided that $S$ is
nonempty, the module makes a single
nondeterministic choice of an element
$q'\in S$, updates $P_A$ to $\iota(q')$,
and signals completion through
$P_{\mathrm{done}}$.

The deterministic stage is implemented by
an instantaneous activity subnetwork whose
firing sequences faithfully simulate the
steps of the terminating procedure. Its
enabling predicates ensure that no
alternative computational firings are
possible at any intermediate configuration.
The selection stage consists of a finite
choice among terminating branches, each
of which leads to the completion signal.

The initialization module $K_{Q_0}$ is
constructed analogously, using a terminating
procedure to enumerate the finite nonempty
set $Q_0$, followed by a single finite
nondeterministic selection.

For every control action selected in accordance
with Definition~\ref{def3.4}, the computed successor
set $S$ is nonempty; hence the selection
stage is always enabled and terminates
after a single nondeterministic choice.
Thus every permitted computation within
either type of module terminates after
finitely many instantaneous firings.
Since module activation is single-use
and cleanup cannot reactivate a module,
every maximal sequence of consecutive
instantaneous firings starting from a
marking encountered along an initialized
execution is finite.

Let
\[
A_K=(Q_K,\Sigma,C,\delta_K,Q_{0,K})
\]
denote the controlled automaton realized by
$(K,\mu_0)$ in the sense of
Definition~\ref{def:can-realization}.

Let $\operatorname{Reach}_A(Q_0)$ and
$\operatorname{Reach}_{A_K}(Q_{0,K})$
denote the respective reachable-state sets.
Define
\[
\Phi:\operatorname{Reach}_A(Q_0)
\longrightarrow
\operatorname{Reach}_{A_K}(Q_{0,K}),
\qquad
\Phi(q)=\widehat{\mu}_q.
\]
The initialization construction yields
\[
Q_{0,K}=
\{\widehat{\mu}_q:q\in Q_0\}.
\]
Moreover, for every state $q$ reachable from $Q_0$
and every $a\in\Sigma$,
$c\in C$, and $q'\in Q$,
\[
(q,a,c,q')\in\delta
\quad\Longleftrightarrow\quad
(\widehat{\mu}_q,a,c,
 \widehat{\mu}_{q'})\in\delta_K.
\]
By induction on the length of finite
executions, $\Phi$ maps the reachable
states of $A$ bijectively onto the
reachable states of $A_K$ and preserves
and reflects every labeled transition.

The map $\Phi$ is computable from
the effective encoding $\iota$,
and its inverse is computable by
reading and decoding the state code
stored in $P_A$. Hence the reachable
restrictions of $A$ and $A_K$ are
effectively isomorphic.
\end{proof}
The above result shows that, in terms of behavioral equivalence, 
every decidable, effectively finitely branching controlled automaton 
with finite activity and control sets and a finite nonempty initial-state set
can be represented using a controlled activity network comprising standard and inhibitor gates. 
However, such a representation may require substantial structural expansion—potentially 
leading to models that are less intuitive or harder to construct manually.
\begin{remark}
Generalized Stochastic Petri Nets (GSPNs) and classical
Stochastic Activity Networks (SANs) provide alternative
representations of stochastic discrete-event systems.
Although they have equivalent modeling capabilities,
SANs offer a convenient activity-network representation
through input gates, output gates, and associated
probabilistic and timing specifications.
Controlled SANs extend this representation by introducing
explicit control actions and policies, thereby supporting
the systematic analysis of policy-dependent behavior,
controller expressiveness, and policy synthesis.
\end{remark}
Although the preceding modeling-power result
establishes the existence of well-behaved
controlled activity network representations
for a broad class of controlled automata,
it does not imply that well-behavedness can
be decided for an arbitrary computable
controlled activity network. In fact, the
following theorem establishes that this
property is undecidable in general.
\begin{theorem}[Undecidability of well-behavedness]
\label{thm:well-behavedness-undecidable}
For the class of computable controlled activity networks,
the problem of determining whether an initialized network
$(K,\mu_0)$ is well behaved is undecidable.
\end{theorem}

\begin{proof}
By the Turing-completeness of extended Petri nets
with inhibitor arcs, as employed in the proof of
Theorem~\ref{thm:modeling-power}, computable controlled activity
networks can simulate Turing-machine computations.
Given a Turing machine $M$ and an input $x$, construct a
computable controlled activity network whose instantaneous
activities simulate successive computation steps of $M$ on $x$,
and which reaches a stable marking exactly when the simulated
computation halts.

Thus, the constructed initialized network is well-behaved if and
only if $M$ halts on $x$. Hence, a decision procedure for
well-behavedness would yield a decision procedure for the halting
problem. Therefore, well-behavedness is undecidable.
\end{proof}
\subsection{Policy Types}
\label{Policy Types Nondeterministic}
We next define policies and several policy types for controlled
automata.  
Throughout the following language-theoretic analysis,
we work over a fixed finite nonempty observable alphabet
$\Sigma$, unless otherwise stated. The state sets and
nonempty control-action sets of individual automata may
vary.

We use $\Pi$ to denote a policy class, with
$\Pi\in\{0,F,\mathrm{Stack},\mathrm{Tape},H\}$.
The notation $\mathfrak C$ denotes collections of
automaton-specific language families, whereas
$\mathcal C$ denotes global classes of individual
languages. The superscript $\omega$ indicates
infinite-word acceptance.

The classes of finite controlled automata and finite
controlled B\"uchi automata over $\Sigma$ are denoted
by $\mathfrak A_{\mathrm{fin}}(\Sigma)$ and
$\mathfrak A^\omega_{\mathrm{fin}}(\Sigma)$,
respectively. In both cases, the finite state sets
and finite nonempty control-action sets may vary.
\begin{definition} \label{3.16}
Let
\[
\mathcal{A}=(Q,\Sigma,C,\delta,Q_0)
\]
be a controlled automaton. A \emph{history-dependent policy} admissible for
$\mathcal A$ is a function
\[
\pi:(Q\times\Sigma)^*\longrightarrow C.
\]
Consider a finite execution
\[
q_0\xrightarrow{a_0,c_0}q_1
\xrightarrow{a_1,c_1}\cdots
\xrightarrow{a_{n-1},c_{n-1}}q_n
\]
of $\mathcal{A}$. For $0\leq i<n$, let
\[
h_i=(q_0,a_0)(q_1,a_1)\cdots(q_i,a_i).
\]
The execution is said to be \emph{under} policy $\pi$ if
\[
c_i=\pi(h_i),\qquad 0\leq i<n.
\]

A policy $\pi$ is \emph{memoryless} if the control action selected
after observing the current state $q$ and completed timed activity
$a$ depends only on $q$ and $a$, i.e.,
\[
\pi:Q\times\Sigma\longrightarrow C.
\]
\end{definition}
\begin{definition} \label{equivpolicy}
Let
\[
\mathcal{A}=(Q,\Sigma,C,\delta,Q_0)
\qquad\text{and}\qquad
\mathcal{A}'=(Q',\Sigma,C,\delta',Q'_0)
\]
be equivalent controlled automata under a bisimulation $\gamma$.
Let $\pi$ and $\pi'$ be history-dependent policies admissible for
$\mathcal{A}$ and $\mathcal{A}'$, respectively.

Consider corresponding finite executions
\[
q_0\xrightarrow{a_0,c_0}q_1
\xrightarrow{a_1,c_1}\cdots
\xrightarrow{a_{n-1},c_{n-1}}q_n
\]
and
\[
q'_0\xrightarrow{a_0,c_0}q'_1
\xrightarrow{a_1,c_1}\cdots
\xrightarrow{a_{n-1},c_{n-1}}q'_n
\]
such that
\[
(q_i,q'_i)\in\gamma,\qquad 0\leq i\leq n.
\]
For $0\leq i<n$, let
\[
h_i=(q_0,a_0)\cdots(q_i,a_i),
\qquad
h'_i=(q'_0,a_0)\cdots(q'_i,a_i).
\]

The policies $\pi$ and $\pi'$ are said to be
\emph{equivalent under $\gamma$} if, for every such pair of
corresponding executions,
\[
\pi(h_i)=\pi'(h'_i),
\qquad 0\leq i<n.
\]
\end{definition}

% Insert after Definition 3.17 in the Policy Types subsection.
% Requires the existing definition environment and amsmath/amssymb.
% These definitions use deterministic controller transitions.

\noindent
In the following definitions, let
$\mathcal A=(Q,\Sigma,C,\delta,Q_0)$ be a controlled automaton,
let $O=Q\times\Sigma$, and assume that $O$ and $C$ are finite and
$C\neq\varnothing$. A controller receives observations
$o_i=(q_i,a_i)\in O$ sequentially. After processing $o_i$, it emits
$c_i$; the automaton may then take a transition
$(q_i,a_i,c_i,q_{i+1})\in\delta$. The controller retains its resulting
configuration for the next observation. Each controller below must
produce an action after every nonempty finite observation history,
including histories that are not feasible executions of $\mathcal A$.
The induced policy $\pi:O^*\to C$ assigns an independently fixed
value to $\pi(\varepsilon)$, which is not used in selecting actions
along executions.

\begin{definition}[Finite-memory policy] \label{3.18}
A \emph{finite-memory controller} for $\mathcal A$ is a tuple
\[
\mathcal M=(M,m_0,\delta_M,\gamma_M),
\]
where $M$ is a finite set of memory states, $m_0\in M$ is the initial
memory state, and
\[
\delta_M:M\times O\to M,
\qquad
\gamma_M:O\times M\to C
\]
are total functions. On receiving $o_i=(q_i,a_i)$ in memory state
$m_i$, the controller emits
\[
c_i=\gamma_M(o_i,m_i)
\]
and updates its memory to
\[
m_{i+1}=\delta_M(m_i,o_i).
\]
For each nonempty history $h=o_0\cdots o_i$, the induced policy is
$\pi(h)=c_i$. A history-dependent policy is called
\emph{finite-memory} if it is induced by a finite-memory controller.
\end{definition}

\noindent
For the next two definitions, an \emph{observation transition}
consumes one observation $o\in O$, whereas an \emph{internal
transition} (or $\varepsilon$-transition) consumes no observation.
After one observation transition, the controller executes internal
transitions until none is enabled and then emits its control action.
The transition functions below are deterministic and partial;
whenever an observation is presented at an observation boundary,
its observation transition must be defined. Every ensuing internal
computation must terminate after finitely many steps at a
configuration where the output is defined. These requirements apply
to every finite observation history.

\begin{definition}[Stack-augmented policy] \label{3.19}
A \emph{stack-augmented controller} for $\mathcal A$ is a tuple
\[
\mathcal P=(M,\Gamma,\bot,m_0,s_0,\delta_P,
             \delta_P^{\varepsilon},\gamma_P),
\]
where $M$ is a finite set of controller states, $\Gamma$ is a finite
stack alphabet containing a bottom-of-stack symbol $\bot$, and
$m_0\in M$ is the initial controller state, and
$s_0\in(\Gamma\setminus\{\bot\})^*\bot$ is the specified
finite initial stack. Its partial deterministic
transition functions are
\[
\delta_P:M\times\Gamma\times O
  \rightharpoonup M\times\Gamma^*,
\qquad
\delta_P^{\varepsilon}:M\times\Gamma
  \rightharpoonup M\times\Gamma^*.
\]
A transition replaces the top stack symbol by a finite word and
changes the controller state; it may neither inspect nor modify the
remainder of the stack. The bottom symbol is preserved, so the stack
is never empty. Initially the configuration is $(m_0,s_0)$.

The partial terminal-output function has the form
\[
\gamma_P:O\times M\times\Gamma\rightharpoonup C.
\]
It is evaluated after all internal transitions have terminated and
uses only the current observation, terminal controller state, and top
stack symbol. Information about deeper stack contents may be obtained
through permitted pushdown transitions and retained in the finite
controller state. In particular, the controller may temporarily pop
and restore a symbol to determine whether the bottom marker lies
immediately beneath it. The function $\gamma_P$ itself performs no
additional stack inspection or modification.

On receiving $o_i$, the controller applies $\delta_P$ to its current
state and top stack symbol, then repeatedly applies
$\delta_P^{\varepsilon}$ until it is undefined. If the resulting
terminal configuration is $(m,s)$, it emits
\[
c_i=\gamma_P(o_i,m,\operatorname{top}(s)).
\]
The terminal configuration persists until the next observation.
For every nonempty finite observation history, the specified
observation transitions exist, the internal computation terminates,
and the terminal output is defined. A designated dead mode may
emit a fixed default action when the intended simulation has no
continuation; the plant must have no transitions under that action
from the corresponding dead configurations. The induced policy is
$\pi(o_0\cdots o_i)=c_i$. A history-dependent policy is called
\emph{stack-augmented} if it is induced by a stack-augmented
controller.
\end{definition}

\begin{definition}[Tape-augmented policy] \label{3.20}
A \emph{tape-augmented controller} for $\mathcal A$ is a tuple
\[
\mathcal T=(M,\Gamma,\square,m_0,\tau_0,j_0,
             \delta_T,\delta_T^{\varepsilon},\gamma_T),
\]
where $M$ is a finite set of controller states, $\Gamma$ is a finite
tape alphabet containing a blank symbol $\square$, $m_0\in M$ is
the initial controller state, $\tau_0$ is the initial tape contents,
and $j_0$ is the initial head position. The tape is an ordinary
unbounded Turing-machine tape, initially blank except for a finite
specified region. Its partial deterministic transition functions
are
\[
\delta_T:M\times\Gamma\times O
  \rightharpoonup M\times\Gamma\times\{L,R,S\},
\qquad
\delta_T^{\varepsilon}:M\times\Gamma
  \rightharpoonup M\times\Gamma\times\{L,R,S\}.
\]
Each transition changes the controller state, overwrites the
currently scanned symbol, and moves the tape head one cell left,
one cell right, or not at all. An observation transition also
consumes the current observation. The partial terminal-output
function
\[
\gamma_T:O\times M\times\Gamma\rightharpoonup C
\]
is a finite table evaluated using the current observation,
controller state, and scanned tape symbol.

On receiving $o_i$, the controller applies $\delta_T$ and then
repeatedly applies $\delta_T^{\varepsilon}$ until it is undefined.
At the resulting terminal configuration $(m,\tau,j)$, it emits
\[
c_i=\gamma_T(o_i,m,\tau(j)).
\]
The resulting controller state, tape contents, and head position
persist until the next observation. For every nonempty finite
observation history, the specified observation transitions exist,
the internal computation terminates, and the terminal output is
defined. The induced policy is $\pi(o_0\cdots o_i)=c_i$. A
history-dependent policy is called \emph{tape-augmented} if it is
induced by a tape-augmented controller.
\end{definition}
\subsection{Nondeterministic Languages on Finite Words}
In the remainder of this section, to further explore formal language recognition by controlled automata, we assume that the automaton has a finite set of states along with a designated subset 
$F$, referred to as the set of \emph{accepting states}.
\begin{definition} \label{def3.21}
Let
\[
\mathcal{A}=(Q,\Sigma,C,\delta,Q_0)
\]
be a controlled automaton with a finite set of states, and let
$F\subseteq Q$ be a set of accepting states. Given a policy $\pi$
on $\mathcal{A}$, the language accepted by $\mathcal{A}$ under
$\pi$ is
\[
L_\pi(\mathcal{A})
=
\left\{
a_0a_1\cdots a_{n-1}\in\Sigma^*
\;\middle|\;
\begin{array}{l}
\text{there exists a finite execution}\\[1mm]
q_0\xrightarrow{a_0,c_0}q_1
\xrightarrow{a_1,c_1}\cdots
\xrightarrow{a_{n-1},c_{n-1}}q_n\\[1mm]
\text{of $\mathcal{A}$ under $\pi$ such that }
q_0\in Q_0\text{ and }q_n\in F
\end{array}
\right\}.
\]
\end{definition}
\begin{definition} 
\label{def:deterministic.controlled.automaton}
Let
\[
\mathcal{A}=(Q,\Sigma,C,\delta,Q_0)
\]
be a controlled automaton. The automaton $\mathcal{A}$ is
\emph{deterministic} if for every $q\in Q$, $a\in\Sigma$, and
$c\in C$, there exists at most one state $q'\in Q$ such that
\[
(q,a,c,q')\in\delta.
\]
Equivalently, the transition relation $\delta$ can be regarded as
a partial function
\[
\delta:Q\times\Sigma\times C\rightharpoonup Q.
\]
If this function is total, $\mathcal{A}$ is said to be
\emph{completely deterministic}.
\end{definition}

\begin{proposition} \label{3.23}
Let
\[
\mathcal{A}=(Q,\Sigma,C,\delta,Q_0)
\]
be a deterministic controlled automaton, and let
\[
\pi:(Q\times\Sigma)^*\longrightarrow C
\]
be a history-dependent policy. Then, for every initial state
$q_0\in Q_0$ and every string
\[
a_0a_1\cdots a_{n-1}\in\Sigma^*,
\]
there exists at most one execution of $\mathcal{A}$ under $\pi$
of the form
\[
q_0\xrightarrow{a_0,\pi(h_0)}q_1
\xrightarrow{a_1,\pi(h_1)}\cdots
\xrightarrow{a_{n-1},\pi(h_{n-1})}q_n,
\]
where
\[
h_i=(q_0,a_0)(q_1,a_1)\cdots(q_i,a_i),
\qquad 0\leq i<n.
\]
\end{proposition}

\begin{proof}
The result follows by induction on the length of the input string.
For a fixed state $q_i$ and activity $a_i$, the policy $\pi$
uniquely determines the control action
\[
c_i=\pi(h_i).
\]
Since $\mathcal{A}$ is deterministic, there is at most one state
$q_{i+1}$ such that
\[
(q_i,a_i,c_i,q_{i+1})\in\delta.
\]
Thus each successive state, if it exists, is uniquely determined.
Hence, for a fixed initial state $q_0$ and activity string
$a_0a_1\cdots a_{n-1}$, there is at most one execution of
$\mathcal{A}$ under $\pi$.
\end{proof}
We now introduce the language classes induced by
the different policy types and examine their
relationships.
\begin{definition}[Policy-induced language classes]
\label{Policy-induced language classes}
Let $\mathcal A$ be a controlled automaton with
accepting states, and let
$\Pi\in\{0,F,\mathrm{Stack},\mathrm{Tape},H\}$.
Denote by $\Pi(\mathcal A)$ the set of all admissible
policies of class $\Pi$ for $\mathcal A$.
The policy-induced language class of $\mathcal A$
is defined by
\[
L_\Pi(\mathcal A)
=
\{L_\pi(\mathcal A)\mid\pi\in\Pi(\mathcal A)\}.
\]
\end{definition}
\begin{proposition}[Policy-induced language hierarchy]
\label{Policy-induced language hierarchy}
For every controlled automaton $\mathcal A$,
\[
L_0(\mathcal A)
\subseteq L_F(\mathcal A)
\subseteq L_{\mathrm{Stack}}(\mathcal A)
\subseteq L_{\mathrm{Tape}}(\mathcal A)
\subseteq L_H(\mathcal A).
\]
\end{proposition}

\begin{proof}
The inclusions follow immediately from the
corresponding inclusions among the policy classes.
\end{proof}
\begin{definition} \label{equivfinite}
Let
\[
\mathcal{A}=(Q,\Sigma,C,\delta,Q_0)
\qquad\text{and}\qquad
\mathcal{A}'=(Q',\Sigma,C,\delta',Q'_0)
\]
be equivalent controlled automata under a bisimulation $\gamma$ as
in Definition~\ref{def3.5}, and let $F\subseteq Q$ and
$F'\subseteq Q'$ be their respective sets of accepting states.

The controlled automata with accepting states $(\mathcal{A},F)$
and $(\mathcal{A}',F')$ are said to be \emph{equivalent under
$\gamma$} if acceptance is preserved by $\gamma$, i.e.,
for every $(q,q')\in\gamma$ with $q\in Q$ and $q'\in Q'$,
\[
q\in F
\quad\Longleftrightarrow\quad
q'\in F'.
\]
\end{definition}
\begin{proposition}\label{3.26}
Let $(\mathcal A,F)$ and $(\mathcal A',F')$ be
equivalent controlled automata with accepting states
under a bisimulation $\gamma$, as in
Definition~\ref{equivfinite}.
If $\pi$ and $\pi'$ are equivalent policies on
$\mathcal A$ and $\mathcal A'$, respectively, as in
Definition~\ref{equivpolicy}, then
\[
L_\pi(\mathcal A)=L_{\pi'}(\mathcal A').
\]

Moreover, let
$\Pi\in\{0,F,\mathrm{Stack},\mathrm{Tape},H\}$
be any of the policy classes considered in
Definition~\ref{Policy-induced language classes}.
If every policy of class $\Pi$ on $\mathcal A$
has an equivalent policy of the same class on
$\mathcal A'$, and conversely, then
\[
L_\Pi(\mathcal A)=L_\Pi(\mathcal A').
\]
\end{proposition}

\begin{proof}
Let
\[
w=a_0a_1\cdots a_{n-1}\in L_\pi(\mathcal A).
\]
Then there exists an accepting execution
\[
q_0
\xrightarrow{a_0,c_0}
q_1
\xrightarrow{a_1,c_1}
\cdots
\xrightarrow{a_{n-1},c_{n-1}}
q_n
\]
of $\mathcal A$ under $\pi$, with
$q_0\in Q_0$ and $q_n\in F$.

By Proposition~\ref{3.7}, this execution has
a corresponding execution
\[
q'_0
\xrightarrow{a_0,c_0}
q'_1
\xrightarrow{a_1,c_1}
\cdots
\xrightarrow{a_{n-1},c_{n-1}}
q'_n
\]
of $\mathcal A'$ such that
\[
(q_i,q'_i)\in\gamma,
\qquad 0\leq i\leq n.
\]
Since $\pi$ and $\pi'$ are equivalent,
corresponding observation histories select
the same control actions. Hence the matched
execution is an execution under $\pi'$.

Moreover, acceptance preservation under
$\gamma$ implies that $q'_n\in F'$.
Consequently,
\[
w\in L_{\pi'}(\mathcal A'),
\]
and therefore
\[
L_\pi(\mathcal A)
\subseteq L_{\pi'}(\mathcal A').
\]
The reverse inclusion follows symmetrically.
Thus
\[
L_\pi(\mathcal A)=L_{\pi'}(\mathcal A').
\]

Now fix a policy class $\Pi$ satisfying the
stated bidirectional policy-transport condition.
If
\[
L\in L_\Pi(\mathcal A),
\]
then, by
Definition~\ref{Policy-induced language classes},
there exists a policy $\pi$ of class $\Pi$
on $\mathcal A$ such that
\[
L=L_\pi(\mathcal A).
\]
By assumption, there exists an equivalent
policy $\pi'$ of the same class on
$\mathcal A'$.

By the first part of the proposition,
\[
L=L_{\pi'}(\mathcal A'),
\]
so
\[
L_\Pi(\mathcal A)
\subseteq L_\Pi(\mathcal A').
\]
The reverse inclusion follows from the
converse policy-transport assumption.
Therefore,
\[
L_\Pi(\mathcal A)=L_\Pi(\mathcal A').
\]
\end{proof}

By Definition~\ref{3.20}, tape-augmented policies
coincide with total computable history-dependent
policies. Consequently, for every controlled
automaton $\mathcal A$,
\[
L_{\mathrm{Tape}}(\mathcal A)
=
L_H^{\mathrm{comp}}(\mathcal A).
\]
\subsection{Language Families and Global Classes for Controlled Automata on Finite Words}
We now consider two related notions obtained by ranging over all finite
controlled automata and policies of a specified type. First, we collect
the automaton-specific language families introduced in Definition~\ref{Policy-induced language classes},
thereby retaining the distinction between the sets
of languages induced by different policy classes on each 
automaton. Second, by taking the union of these families over all
finite controlled automata, we obtain global language classes whose
elements are individual finite-word languages. The former notion
captures plant-specific policy expressiveness, whereas the latter
characterizes the overall language-theoretic expressive power of each
policy class.

\begin{definition}[Language families and global classes]
\label{Language families and global classes finite words}
Consider the class $\mathfrak A_{\mathrm{fin}}(\Sigma)$
of finite controlled automata
\[
\mathcal A=(Q,\Sigma,C,\delta,Q_0)
\]
equipped with a designated accepting set
$F\subseteq Q$, where the finite state set $Q$
and finite nonempty control-action set $C$ may vary.

For each policy class
\[
\Pi\in\{0,F,\mathrm{Stack},\mathrm{Tape},H\},
\]
define the corresponding collection of automaton-specific
language families by
\[
\mathfrak C_\Pi
=
\{
\mathcal L_\Pi(\mathcal A)
\mid
\mathcal A\in\mathfrak A_{\mathrm{fin}}(\Sigma)
\},
\]
where $\mathcal L_\Pi(\mathcal A)$ is defined as in
Definition~\ref{Policy-induced language classes}.
The corresponding global language class is
\[
\mathcal C_\Pi
=
\bigcup_{\mathcal A\in\mathfrak A_{\mathrm{fin}}(\Sigma)}
\mathcal L_\Pi(\mathcal A)
=
\bigcup\mathfrak C_\Pi.
\]

Thus, $\mathfrak C_\Pi$ is a collection of
automaton-specific language families, whereas
$\mathcal C_\Pi$ is the class of individual
finite-word languages over $\Sigma$ accepted by some
finite controlled automaton with control-action set
$C$ under a policy of type $\Pi$.

For two policy classes $\Pi_1$ and $\Pi_2$, write
\[
\mathfrak C_{\Pi_1}\preceq\mathfrak C_{\Pi_2}
\]
if, for every
$\mathcal A\in\mathfrak A_{\mathrm{fin}}(\Sigma)$,
\[
\mathcal L_{\Pi_1}(\mathcal A)
\subseteq
\mathcal L_{\Pi_2}(\mathcal A).
\]

We write
\[
\mathfrak C_{\Pi_1}\prec\mathfrak C_{\Pi_2}
\]
if
\[
\mathfrak C_{\Pi_1}\preceq\mathfrak C_{\Pi_2}
\]
and there exists some
$\mathcal A\in\mathfrak A_{\mathrm{fin}}(\Sigma)$
such that
\[
\mathcal L_{\Pi_1}(\mathcal A)
\subsetneq
\mathcal L_{\Pi_2}(\mathcal A).
\]
Equivalently,
\[
\mathfrak C_{\Pi_1}\prec\mathfrak C_{\Pi_2}
\]
if
\[
\mathfrak C_{\Pi_1}\preceq\mathfrak C_{\Pi_2}
\qquad\text{and}\qquad
\mathfrak C_{\Pi_2}\not\preceq\mathfrak C_{\Pi_1}.
\]
\end{definition}
The automaton-specific inclusions of Proposition~\ref{Policy-induced language hierarchy} immediately imply
the global inclusions
\[
\mathcal{C}_{0}
\subseteq
\mathcal{C}_{F}
\subseteq
\mathcal{C}_{\mathrm{Stack}}
\subseteq
\mathcal{C}_{\mathrm{Tape}}
\subseteq
\mathcal{C}_{H}.
\]
By Definition~\ref{3.20}, we also have
\[
\mathfrak{C}_{\mathrm{Tape}}
=
\mathfrak{C}_{H}^{\mathrm{comp}}.
\]
\begin{lemma}[Realization of context-free languages by stack-augmented policies]
\label{3.29}
For every context-free language \(L\subseteq\Sigma^*\), there exists a
finite controlled automaton
\[
\mathcal A=(Q,\Sigma,C,\delta,Q_0)
\]
with a designated accepting set \(F\subseteq Q\), and a deterministic
stack-augmented policy \(\pi\), such that
\[
L_\pi(\mathcal A)=L.
\]
Consequently,
\[
\mathrm{CFL}\subseteq\mathcal C_{\mathrm{Stack}}.
\]
\end{lemma}

\begin{proof}
Let $L\subseteq\Sigma^*$ be context-free. By the Greibach
normal-form theorem, there exists a context-free grammar
\[
G=(V,\Sigma,P,S)
\]
such that $L(G)=L\setminus\{\varepsilon\}$ and every production
in $P$ has the form
\[
\rho:A\longrightarrow a\alpha,
\qquad A\in V,\quad a\in\Sigma,\quad\alpha\in V^*.
\]
We use the convention that the leftmost symbol of a stack word
is its top.

The standard real-time nondeterministic pushdown automaton
associated with $G$ has initial stack $S$ and, for every
production $\rho:A\to a\alpha$, a transition that consumes $a$
and replaces the top symbol $A$ by $\alpha$. Thus, from stack
$A\beta$, this transition produces $\alpha\beta$. The automaton
accepts by empty stack after consuming the entire input.

We construct a finite controlled automaton
\[
\mathcal A=(Q,\Sigma,C,\delta,Q_0)
\]
and a deterministic stack-augmented policy $\pi$ that together
simulate this pushdown automaton. The policy maintains the
simulated stack, whereas nondeterministic choices among
applicable productions are made by the controlled automaton.
The plant state reached after such a choice identifies the
selected production, which becomes observable to the policy
at the next input position.

Set
\[
Q=\{q_0\}\cup
  \{q_{\rho,b}:\rho\in P,\ b\in\{0,1\}\},
\qquad Q_0=\{q_0\},
\]
and
\[
C=\{c_{A,b}:A\in V,\ b\in\{0,1\}\}
  \cup\{c_{\mathrm{dead}}\}.
\]
For a nonempty simulated stack $A\beta$, the bit $b$ indicates
whether $A$ is its only symbol:
\[
b=1\quad\Longleftrightarrow\quad\beta=\varepsilon.
\]
Thus $c_{A,b}$ reports both the current top variable and
whether that variable is the last one on the simulated stack.

For every $q\in Q$, every production
$\rho:A\to a\alpha\in P$, and every $b\in\{0,1\}$, include
\[
(q,a,c_{A,b},q_{\rho,b})\in\delta.
\]
There are no other transitions. In particular, no transition
is available under $c_{\mathrm{dead}}$. Consequently, when
the policy emits $c_{A,b}$ on input $a$, the possible successors
are exactly the states $q_{\rho,b}$ corresponding to productions
$\rho:A\to a\alpha$. The plant's nondeterminism therefore
reproduces the choice among applicable grammar productions precisely.

We next specify the deterministic stack-augmented controller
inducing $\pi$. Its stack alphabet is
\[
\Gamma=V\cup\{\bot\},
\]
where $\bot\notin V$ is a permanent bottom marker, and its
initial stack is $S\bot$. Its finite-control states include
an initial mode, finitely many modes for processing the
productions in $P$, modes for testing whether the simulated
stack contains exactly one variable, and a designated dead
mode. They also retain the current top variable and the
one-bit indicator whenever these are needed for output.
Because $V$ and $P$ are finite, only finitely many such states
are required.

At each observation, the controller first takes its
observation transition, as required by Definition~\ref{3.16}.
It then performs a finite deterministic sequence of internal
$\varepsilon$-transitions and emits an action at the resulting
terminal configuration.

On the initial observation $(q_0,a)$, the observation
transition enters a mode for preparing the initial output
without changing the simulated stack. The subsequent internal
transitions determine whether the initial top variable $A=S$
is the only variable above $\bot$. To do so, the controller
temporarily pops $A$, inspects the newly exposed top symbol,
and restores $A$. More explicitly, if the exposed symbol is
$X\in V\cup\{\bot\}$, restoration replaces $X$ by $AX$;
this preserves the bottom marker and restores the original
stack. The controller stores $A$ and the resulting bit $b$
in its finite-control state. Its terminal-output function
$\gamma_P$, evaluated using the current observation, that
finite-control state, and the top stack symbol, then emits
$c_{A,b}$. In particular, $\gamma_P$ does not inspect any
deeper stack symbol.

Suppose the plant subsequently takes a transition
\[
q\xrightarrow{a,c_{A,b}}q_{\rho,b},
\qquad \rho:A\to a\alpha.
\]
The stack replacement prescribed by $\rho$ is deliberately
deferred until the next observation. If a next observation
$(q_{\rho,b},a')$ arrives, its observation transition reads
the plant state $q_{\rho,b}$ and enters the finite-control
mode associated with $\rho$. Using that transition and, if
needed, finitely many subsequent internal transitions, the
controller replaces the top stack variable $A$ by $\alpha$.
Thus a simulated stack $A\beta\bot$ becomes
$\alpha\beta\bot$. This replacement is deterministic because
$\rho$ is determined by the observed plant state.

If the resulting simulated stack is nonempty, say
$B\gamma\bot$, the controller again determines, by a
temporary pop-and-restore operation, whether $\gamma$ is
empty. It stores $B$ and
\[
b'=1\quad\Longleftrightarrow\quad\gamma=\varepsilon
\]
in finite control and emits $c_{B,b'}$ at its terminal
configuration. Every such computation uses finitely many
pushdown transitions and terminates.

If the pending replacement instead leaves only $\bot$, the
controller enters its designated dead mode and emits
$c_{\mathrm{dead}}$. It remains in that mode on all subsequent
observations, leaving $\bot$ intact and emitting the same
action. Hence the controller is defined and terminates on
every finite observation history, including histories that
do not correspond to feasible plant executions. Since the
plant has no transitions under $c_{\mathrm{dead}}$, entering
this mode cannot create any additional plant execution.
The construction therefore satisfies the totality and
termination requirements of Definition~\ref{3.16}.

We verify the simulation. Let
\[
w=a_0a_1\cdots a_{n-1}\in\Sigma^*,
\qquad n\geq 1,
\]
and consider a computation of the grammar's pushdown
automaton that selects productions
\[
\rho_i:A_i\to a_i\alpha_i,
\qquad 0\leq i<n.
\]
Write $\sigma_i\in V^*$ for its stack immediately before
reading $a_i$. We show inductively that, immediately before
the plant processes $a_i$:
\begin{enumerate}
\item the controller's simulated stack, excluding $\bot$,
      is exactly $\sigma_i$;
\item if $\sigma_i=A_i\beta_i$, the controller emits
      $c_{A_i,b_i}$, where
      \[
      b_i=1\quad\Longleftrightarrow\quad
      \beta_i=\varepsilon;
      \]
\item the plant can move to $q_{\rho_i,b_i}$ exactly when
      $\rho_i$ is a pushdown move enabled from $\sigma_i$
      on $a_i$.
\end{enumerate}

For $i=0$, both machines start with simulated stack $S$.
The controller's initial output and the definition of
$\delta$ establish all three assertions.

Assume they hold at position $i<n-1$. If the plant chooses
\[
q\xrightarrow{a_i,c_{A_i,b_i}}q_{\rho_i,b_i},
\qquad \rho_i:A_i\to a_i\alpha_i,
\]
then, on receiving the next observation, the controller
performs the pending replacement $A_i\mapsto\alpha_i$.
This is exactly the pushdown automaton's stack update, so
the resulting simulated stack is $\sigma_{i+1}$. If this
stack is nonempty, the controller reports its new top
variable and last-symbol bit, establishing the assertions
at position $i+1$. If it is empty, neither machine has a
continuation on the remaining input: the controller enters
the dead mode, and the plant has no transition under
$c_{\mathrm{dead}}$.

It follows that every pushdown computation on a nonempty
word determines a corresponding plant execution under
$\pi$, and conversely every such plant execution determines
the same sequence of grammar productions and hence a
pushdown computation.

It remains to characterize acceptance, since the replacement
selected on the final input symbol is not performed at a
subsequent observation. Define
\[
F=\{q_{\rho,1}:\rho\text{ is a production }A\to a\}.
\]
If the plant ends in $q_{\rho,1}\in F$, the bit $1$ records
that the simulated stack immediately before the final move
contained only $A$, and the production $\rho:A\to a$
replaces $A$ by the empty word. Thus the corresponding
pushdown computation empties its stack exactly after
consuming the input. Conversely, every accepting empty-stack
computation must finish by applying a production $A\to a$
to its sole remaining stack variable, so its corresponding
plant execution ends in a state belonging to $F$.
Therefore, for every nonempty $w\in\Sigma^*$,
\[
w\in L(G)\quad\Longleftrightarrow\quad
w\in L_\pi(\mathcal A).
\]

Finally, if $\varepsilon\in L$, add $q_0$ to $F$; otherwise
leave $q_0\notin F$. Since $q_0$ is the unique initial
state, this handles the empty word. Consequently,
\[
L_\pi(\mathcal A)=L,
\]
and hence
\[
\mathrm{CFL}\subseteq\mathcal C_{\mathrm{Stack}}.
\]
\end{proof}
\begin{corollary}
\label{3.30}
For controlled automata on finite words,
\[
\mathcal C_{\mathrm{Stack}}=\mathrm{CFL}.
\]
\end{corollary}

\begin{proof}
Lemma~\ref{3.29} gives
\[
\mathrm{CFL}\subseteq\mathcal C_{\mathrm{Stack}}.
\]
Conversely, the product of a finite controlled automaton with the
deterministic pushdown transducer implementing a stack-augmented policy
is a nondeterministic pushdown automaton: the transducer determines the
stack evolution and control action deterministically, while any
remaining branching is finite-state plant nondeterminism. Hence every
language in \(\mathcal C_{\mathrm{Stack}}\) is context-free, so
\[
\mathcal C_{\mathrm{Stack}}\subseteq\mathrm{CFL}.
\]
Therefore
\[
\mathcal C_{\mathrm{Stack}}=\mathrm{CFL}.
\]
\end{proof}

\begin{theorem}[Pointwise and global language hierarchies]
\label{3.31}
For controlled automata on finite words, the collections
of automaton-specific language families satisfy the
strict pointwise hierarchy
\[
\mathfrak C_0
\prec
\mathfrak C_F
\prec
\mathfrak C_{\mathrm{Stack}}
\prec
\mathfrak C_{\mathrm{Tape}}
\prec
\mathfrak C_H.
\]
Moreover, the global language classes satisfy
\[
\mathcal C_0
=
\mathcal C_F
=
\mathrm{REG}
\subsetneq
\mathcal C_{\mathrm{Stack}}
=
\mathrm{CFL}
\subsetneq
\mathcal C_{\mathrm{Tape}}
=
\mathcal C_H^{\mathrm{comp}}
=
\mathrm{REC}
\subsetneq
\mathcal C_H
=
\mathrm{ALL},
\]
where $\mathrm{REG}$, $\mathrm{CFL}$, and $\mathrm{REC}$
denote the classes of regular, context-free, and
recursive languages, respectively, and
\[
\mathrm{ALL}
=
\bigcup_{\Sigma\text{ finite}}2^{\Sigma^*}
\]
denotes the class of all languages over alphabets.
\end{theorem}

\begin{proof}
By Proposition~\ref{Policy-induced language hierarchy}, 
for every finite controlled automaton $\mathcal A$,
\[
\mathcal L_0(\mathcal A)
\subseteq
\mathcal L_F(\mathcal A)
\subseteq
\mathcal L_{\mathrm{Stack}}(\mathcal A)
\subseteq
\mathcal L_{\mathrm{Tape}}(\mathcal A)
\subseteq
\mathcal L_H(\mathcal A).
\]
Consequently,
\[
\mathfrak C_0
\preceq
\mathfrak C_F
\preceq
\mathfrak C_{\mathrm{Stack}}
\preceq
\mathfrak C_{\mathrm{Tape}}
\preceq
\mathfrak C_H.
\]
We establish the strictness of each comparison.

First, consider the unary controlled automaton
\[
\mathcal A_1
=
(\{q,f\},\{a\},\{c_0,c_1\},\delta,\{q\}),
\qquad F=\{f\},
\]
with transitions
\[
q\xrightarrow{a,c_0}q,
\qquad
q\xrightarrow{a,c_1}f,
\]
and
\[
f\xrightarrow{a,c_0}q,
\qquad
f\xrightarrow{a,c_1}q.
\]
Under a memoryless policy, the action selected at
$(q,a)$ is fixed. If it is $c_0$, no nonempty word
is accepted; if it is $c_1$, exactly the odd-length
words are accepted. Hence no memoryless policy on
$\mathcal A_1$ accepts
\[
L_3=\{a^{3k}\mid k\geq 1\}.
\]
A finite-memory policy can keep the input length
modulo three and select $c_1$ exactly when the
current length is divisible by three, thereby
accepting $L_3$. Thus
\[
\mathcal L_0(\mathcal A_1)
\subsetneq
\mathcal L_F(\mathcal A_1),
\]
which gives
\[
\mathfrak C_0\prec\mathfrak C_F.
\]

Next, consider the nonregular context-free language
\[
L_{\mathrm{pal}}
=
\{ww^R\mid w\in\{0,1\}^*\}.
\]
By Lemma~\ref{3.29}, there exists a finite
controlled automaton $\mathcal A_{\mathrm{pal}}$
and a stack-augmented policy accepting
$L_{\mathrm{pal}}$. Every language induced by
a finite-memory policy on a finite controlled
automaton is regular. Therefore
\[
\mathcal L_F(\mathcal A_{\mathrm{pal}})
\subsetneq
\mathcal L_{\mathrm{Stack}}
(\mathcal A_{\mathrm{pal}}),
\]
and hence
\[
\mathfrak C_F
\prec
\mathfrak C_{\mathrm{Stack}}.
\]

For the next comparison, consider
\[
L_{abc}
=
\{a^n b^n c^n\mid n\geq 0\},
\]
which is recursive but not context-free.
A deterministic Turing transducer can decide,
after each observed input prefix, whether that
prefix belongs to $L_{abc}$ and output a control
action accordingly. Coupled with a finite
controlled automaton whose accepting status
reflects the selected action, this yields
a tape-augmented policy accepting $L_{abc}$.
Since Corollary~\ref{3.30} establishes
\[
\mathcal C_{\mathrm{Stack}}=\mathrm{CFL},
\]
no stack-augmented policy on any finite controlled
automaton can accept $L_{abc}$. Consequently,
for a suitable automaton $\mathcal A_{abc}$,
\[
\mathcal L_{\mathrm{Stack}}(\mathcal A_{abc})
\subsetneq
\mathcal L_{\mathrm{Tape}}(\mathcal A_{abc}),
\]
and therefore
\[
\mathfrak C_{\mathrm{Stack}}
\prec
\mathfrak C_{\mathrm{Tape}}.
\]

Finally, let $L_{\mathrm{nr}}$ be a nonrecursive
language over an alphabet. An unrestricted
history-dependent policy can select an accepting
or rejecting control action according to whether
each observed input prefix belongs to
$L_{\mathrm{nr}}$. Thus some finite controlled
automaton $\mathcal A_{\mathrm{nr}}$ accepts
$L_{\mathrm{nr}}$ under such a policy.
In contrast, every language induced by a
tape-augmented policy on a finite controlled
automaton is recursive, since its controller
computes an action and halts on every finite
observation history. Hence
\[
\mathcal L_{\mathrm{Tape}}(\mathcal A_{\mathrm{nr}})
\subsetneq
\mathcal L_H(\mathcal A_{\mathrm{nr}}),
\]
which establishes
\[
\mathfrak C_{\mathrm{Tape}}
\prec
\mathfrak C_H.
\]

It remains to characterize the global language
classes.

For memoryless and finite-memory policies,
the product of the finite controlled automaton
with the finite-state controller is a finite
automaton. Conversely, every regular language
can be recognized by a finite controlled
automaton with a singleton control set.
Therefore
\[
\mathcal C_0
=
\mathcal C_F
=
\mathrm{REG}.
\]

By Corollary~\ref{3.30},
\[
\mathcal C_{\mathrm{Stack}}
=
\mathrm{CFL}.
\]

For tape-augmented policies, the combination
of a finite controlled automaton with its
deterministic Turing-transducer controller
can be simulated by a Turing machine.
Because the controller terminates on every
finite observation history, membership in
the induced language can be decided by
exploring the finitely branching execution
tree for the given finite input. Hence
\[
\mathcal C_{\mathrm{Tape}}
\subseteq
\mathrm{REC}.
\]
Conversely, every recursive language has
a decider. A tape-augmented controller can
apply that decider to each observed input
prefix and select an accepting or rejecting
control action on a suitable finite controlled
automaton. Thus
\[
\mathrm{REC}
\subseteq
\mathcal C_{\mathrm{Tape}},
\]
and therefore
\[
\mathcal C_{\mathrm{Tape}}
=
\mathrm{REC}.
\]

By Definition~\ref{3.19}, tape-augmented policies
coincide with total computable
history-dependent policies. Consequently,
for every finite controlled automaton
$\mathcal A$,
\[
\mathcal L_{\mathrm{Tape}}(\mathcal A)
=
\mathcal L_H^{\mathrm{comp}}(\mathcal A).
\]
Taking the union over all finite controlled
automata gives
\[
\mathcal C_{\mathrm{Tape}}
=
\mathcal C_H^{\mathrm{comp}}
=
\mathrm{REC}.
\]

Finally, let $L\subseteq\Sigma^*$ be arbitrary,
where $\Sigma$ is finite. Consider a finite
controlled automaton with two control actions,
one leading to an accepting state and the
other to a rejecting state. An unrestricted
history-dependent policy can select the
accepting action after observing a nonempty
word $w$ exactly when $w\in L$.
The initial-state acceptance condition can
be chosen according to whether
$\varepsilon\in L$. Hence every language
over an alphabet is realizable under
an unrestricted history-dependent policy,
so
\[
\mathcal C_H=\mathrm{ALL}.
\]

The global strict inclusions now follow from
the classical proper inclusions
\[
\mathrm{REG}
\subsetneq
\mathrm{CFL}
\subsetneq
\mathrm{REC}
\subsetneq
\mathrm{ALL}.
\]
This establishes both the strict pointwise
hierarchy and the global characterizations.
\end{proof}

\begin{remark}
\label{remark:pointwise-global}
Theorem~\ref{3.31} distinguishes two complementary
notions of policy expressiveness. The pointwise hierarchy
compares the language families realizable on a fixed
finite controlled automaton, whereas the global hierarchy
allows the underlying automaton to vary.

These distinctions need not coincide. In particular,
\[
\mathfrak C_0 \prec \mathfrak C_F,
\qquad
\mathcal C_0=\mathcal C_F=\mathrm{REG}.
\]
Thus, finite controller memory can strictly increase
expressiveness on a particular controlled automaton
without enlarging the corresponding global language class.

Two further aspects of the global characterization
are worth noting. First, although stack-augmented
policies are implemented by deterministic pushdown
transducers, the underlying controlled automaton
may be nondeterministic, and acceptance is existential
over its executions. Consequently, the resulting
global language class is
\[
\mathcal C_{\mathrm{Stack}}=\mathrm{CFL},
\]
rather than being restricted to deterministic
context-free languages.

Second, tape-augmented policies are required to
produce a control action after finitely many
computation steps for every finite observation
history. This totality requirement yields
\[
\mathcal C_{\mathrm{Tape}}
=
\mathcal C_H^{\mathrm{comp}}
=
\mathrm{REC},
\]
rather than the class $\mathrm{RE}$ of recursively
enumerable languages. By contrast, unrestricted
history-dependent policies need not be computable,
which accounts for
\[
\mathcal C_H=\mathrm{ALL}.
\]
\end{remark}
\subsection{Closure Properties of Controlled Language Classes}
Having characterized policy expressiveness at both the
automaton-specific and global levels, we now examine the algebraic
behavior of the induced language classes under standard language
operations. Correspondingly, we distinguish two notions of closure.
Uniform closure retains the automaton-specific viewpoint: it requires
a single effective construction of a resulting controlled automaton
from the given component automata, uniformly over the policies
realizing the operand languages. Standard closure adopts the global
language-class viewpoint: it asks only whether applying the operation
to languages in the global class yields another language in that same
global class, without requiring a uniform construction from the
particular automata realizing the operands.
\begin{definition} \label{uniform}
Let $\Pi$ be a policy class and let $\mathrm{Op}$
be a $k$-ary operation on languages over $\Sigma$.
\begin{enumerate}
\item[(a)] \emph{Uniform closure.} \label{Uniform closure}
The policy class $\Pi$ is said to be uniformly closed under
$\operatorname{Op}$ if there exists a computable plant constructor
\[
\Phi_{\operatorname{Op}}:
\mathfrak A_{\mathrm{fin}}(\Sigma)^k
\longrightarrow
\mathfrak A_{\mathrm{fin}}(\Sigma)
\]
such that, for every
\[
\mathcal A_1,\ldots,\mathcal A_k
\in\mathfrak A_{\mathrm{fin}}(\Sigma)
\]
and every
\[
L_i\in L_\Pi(\mathcal A_i),\qquad 1\leq i\leq k,
\]
we have
\[
\operatorname{Op}(L_1,\ldots,L_k)
\in
L_\Pi\!\left(
\Phi_{\operatorname{Op}}
(\mathcal A_1,\ldots,\mathcal A_k)
\right).
\]
Thus the construction of the resulting plant depends only on the
component plants and not on the particular policies realizing the
languages $L_i$.

\item[(b)] \emph{Standard closure.}
The global language class
\[
\mathcal C_\Pi
=
\bigcup_{\mathcal A\in\mathfrak A_{\mathrm{fin}}}
L_\Pi(\mathcal A)
\]
is closed under $\operatorname{Op}$ if
\[
L_1,\ldots,L_k\in\mathcal C_\Pi
\quad\Longrightarrow\quad
\operatorname{Op}(L_1,\ldots,L_k)\in\mathcal C_\Pi.
\]
Equivalently, whenever
$L_i=L_{\pi_i}(\mathcal A_i)$ with $\pi_i\in\Pi$, there exist a finite
controlled automaton $\mathcal A'$ and a policy $\pi'\in\Pi$ such that
\[
L_{\pi'}(\mathcal A')
=
\operatorname{Op}
\bigl(
L_{\pi_1}(\mathcal A_1),\ldots,
L_{\pi_k}(\mathcal A_k)
\bigr).
\]
\end{enumerate}
\end{definition}
\begin{remark}
The distinction between uniform and standard closure parallels the
distinction between the automaton-specific and global notions of
expressiveness introduced above. Uniform closure is a model-level
property: it preserves information about the underlying controlled
automata and requires a single effective construction that works
uniformly for all admissible policies on them. Standard closure is a
language-level property: it depends only on the global class
\(\mathcal C_\Pi\) and permits the resulting language to be realized
by any finite controlled automaton under a policy of type \(\Pi\).
Thus, uniform closure implies standard closure, whereas the converse
need not hold. Unless explicitly stated otherwise, the closure results
below refer to standard closure.
\end{remark}

\begin{proposition}
The classes
\[
\mathcal C_0=\mathcal C_F=\mathrm{REG}
\]
are closed under union, intersection, complement, concatenation, and
Kleene star.
\end{proposition}

\begin{proof}
This follows immediately from Proposition~\ref{3.31} and the classical closure
properties of the regular languages.
\end{proof}

\begin{proposition}
The class
\[
\mathcal C_{\mathrm{Stack}}=\mathrm{CFL}
\]
is closed under union, concatenation, and Kleene star, and under
intersection with regular languages. It is not closed, in general,
under intersection or complement.
\end{proposition}

\begin{proof}
By Lemma~\ref{3.29},
\[
\mathcal C_{\mathrm{Stack}}=\mathrm{CFL}.
\]
The result therefore follows from the corresponding classical closure
properties of the context-free languages.
\end{proof}

\begin{proposition}
The classes
\[
\mathcal C_{\mathrm{Tape}}
=
\mathcal C_H^{\mathrm{comp}}
=
\mathrm{REC}
\]
are closed under union, intersection, complement, difference,
concatenation, and Kleene star.
\end{proposition}

\begin{proof}
By Proposition~\ref{3.31},
\[
\mathcal C_{\mathrm{Tape}}
=
\mathcal C_H^{\mathrm{comp}}
=
\mathrm{REC}.
\]
The recursive languages are closed under the stated operations.
Therefore, the corresponding global controlled-language classes inherit
these closure properties.
\end{proof}

\begin{remark}
The preceding propositions are consequences of the global
characterizations established in Lemma~\ref{3.29} and Proposition~\ref{3.31}.
In particular, the determinism of the transducer realizing a
stack-augmented policy does not reduce the global class to
$\mathrm{DCFL}$, because the underlying controlled automaton may itself
be nondeterministic, and acceptance is existential over executions.
Consequently,
\[
\mathcal C_{\mathrm{Stack}}=\mathrm{CFL},
\]
rather than $\mathrm{DCFL}$. Likewise, the totality requirement on
tape-augmented policies yields
\[
\mathcal C_{\mathrm{Tape}}=\mathrm{REC},
\]
rather than $\mathrm{RE}$.
\end{remark}

The preceding results concern standard closure of the global language
classes. We now consider whether some of these operations can also be
realized uniformly by transformations of the underlying controlled
automata. At the regular level, union and intersection admit the usual
disjoint-union and synchronous-product constructions.
\begin{proposition} \label{3.40}
For each
\[
\Pi\in\{0,F\},
\]
the policy class \(\Pi\) is uniformly closed under union and
intersection in the sense of Definition~\ref{Uniform closure}.
\end{proposition}

\begin{proof}
Fix $\Pi\in\{0,F\}$, and let
\[
\mathcal A_i
=
(Q_i,\Sigma,C_i,\delta_i,Q_{0i}),
\qquad i=1,2,
\]
be finite controlled automata over the same alphabet $\Sigma$, with
accepting sets $F_i\subseteq Q_i$. Without loss of generality, take
$Q_1$ and $Q_2$ to be disjoint. Let
\[
L_i=L_{\pi_i}(\mathcal A_i),
\qquad \pi_i\in\Pi.
\]

For union, define the disjoint-union controlled automaton
\[
\mathcal A_\cup
=
(Q_1\uplus Q_2,\Sigma,
 C_1\uplus C_2,\delta_\cup,
 Q_{01}\uplus Q_{02}),
\]
with accepting set
\[
F_\cup=F_1\uplus F_2.
\]
The transition relation consists of the tagged transitions of the two
component automata:
\[
(q,a,\iota_i(c),q')\in\delta_\cup
\quad\Longleftrightarrow\quad
(q,a,c,q')\in\delta_i,
\]
for $q,q'\in Q_i$, where
$\iota_i:C_i\to C_1\uplus C_2$ is the canonical injection.

A policy $\pi_\cup$ is obtained by applying $\pi_i$ whenever the
current state belongs to the $i$th component and tagging its control
action by $\iota_i$. If $\Pi=0$, this policy is memoryless. If
$\Pi=F$, the finite-memory realizations of $\pi_1$ and $\pi_2$ can be
combined by a finite disjoint union of their memory sets, with the
current plant component determining which memory component is used.
Hence $\pi_\cup\in\Pi$, and
\[
L_{\pi_\cup}(\mathcal A_\cup)
=
L_1\cup L_2.
\]
The construction of $\mathcal A_\cup$ depends only on
$\mathcal A_1$ and $\mathcal A_2$.

For intersection, define the synchronous-product controlled automaton
\[
\mathcal A_\cap
=
(Q_1\times Q_2,\Sigma,
 C_1\times C_2,\delta_\cap,
 Q_{01}\times Q_{02}),
\]
with accepting set
\[
F_\cap=F_1\times F_2,
\]
where
\[
\bigl((q_1,q_2),a,(c_1,c_2),(q'_1,q'_2)\bigr)
\in\delta_\cap
\]
if and only if
\[
(q_1,a,c_1,q'_1)\in\delta_1
\quad\text{and}\quad
(q_2,a,c_2,q'_2)\in\delta_2.
\]

For memoryless policies, define
\[
\pi_\cap((q_1,q_2),a)
=
\bigl(
\pi_1(q_1,a),
\pi_2(q_2,a)
\bigr).
\]
For finite-memory policies, if $\pi_i$ is realized by
\[
(M_i,m_{0i},\delta_{M_i},\gamma_{M_i}),
\]
use the product memory set
\[
M_1\times M_2
\]
with initial memory state $(m_{01},m_{02})$, update function
\[
\delta_M
\bigl(
(m_1,m_2),((q_1,q_2),a)
\bigr)
=
\bigl(
\delta_{M_1}(m_1,(q_1,a)),
\delta_{M_2}(m_2,(q_2,a))
\bigr),
\]
and output function
\[
\gamma_M(m_1,m_2)
=
\bigl(
\gamma_{M_1}(m_1),
\gamma_{M_2}(m_2)
\bigr).
\]
Thus $\pi_\cap\in\Pi$, and
\[
L_{\pi_\cap}(\mathcal A_\cap)
=
L_1\cap L_2.
\]

Both plant constructions are computable and depend only on the
component automata, not on the policies chosen on them. Hence the
memoryless and finite-memory policy classes are uniformly closed under
union and intersection.
\end{proof}
\begin{proposition} \label{3.41}
Let $\mathfrak A_{\mathrm{cdet}}(\Sigma)$ denote the class of finite
controlled automata
\[
\mathcal A=(Q,\Sigma,C,\delta,\{q_0\})
\]
with a single initial state, a designated accepting set $F\subseteq Q$,
and a total deterministic transition function
\[
\delta:Q\times\Sigma\times C\longrightarrow Q.
\]
For each
\[
\Pi\in\{0,F\},
\]
the policy class $\Pi$ is uniformly closed under complement on
$\mathfrak A_{\mathrm{cdet}}(\Sigma)$.

More precisely, there exists a computable plant constructor
\[
\Phi_{\mathrm{c}}:
\mathfrak A_{\mathrm{cdet}}(\Sigma)
\longrightarrow
\mathfrak A_{\mathrm{cdet}}(\Sigma)
\]
defined by
\[
\Phi_{\mathrm{c}}
\bigl(
(Q,\Sigma,C,\delta,\{q_0\}),F
\bigr)
=
\bigl(
(Q,\Sigma,C,\delta,\{q_0\}),Q\setminus F
\bigr),
\]
such that, for every policy $\pi\in\Pi$,
\[
L_{\pi}\!\left(\Phi_{\mathrm{c}}(\mathcal A,F)\right)
=
\Sigma^*\setminus L_{\pi}(\mathcal A).
\]
\end{proposition}

\begin{proof}
Fix
\[
\mathcal A=(Q,\Sigma,C,\delta,\{q_0\})
\in\mathfrak A_{\mathrm{cdet}}(\Sigma)
\]
with accepting set $F\subseteq Q$, and let $\pi\in\Pi$, where
$\Pi\in\{0,F\}$.

Because $\mathcal A$ has a single initial state and $\delta$ is total
and deterministic, every word
\[
w=a_0a_1\cdots a_{n-1}\in\Sigma^*
\]
induces exactly one execution under $\pi$,
\[
q_0
\xrightarrow{a_0,c_0}
q_1
\xrightarrow{a_1,c_1}
\cdots
\xrightarrow{a_{n-1},c_{n-1}}
q_n,
\]
where each control action $c_i$ is uniquely determined by $\pi$.
The automaton
\[
\Phi_{\mathrm{c}}(\mathcal A,F)
\]
has exactly the same states, transitions, initial state, and policy;
only its accepting set is changed from $F$ to $Q\setminus F$.

Hence
\[
w\in L_{\pi}(\mathcal A)
\quad\Longleftrightarrow\quad
q_n\in F
\]
and
\[
w\in
L_{\pi}\!\left(\Phi_{\mathrm{c}}(\mathcal A,F)\right)
\quad\Longleftrightarrow\quad
q_n\in Q\setminus F.
\]
Therefore
\[
w\in
L_{\pi}\!\left(\Phi_{\mathrm{c}}(\mathcal A,F)\right)
\quad\Longleftrightarrow\quad
w\notin L_{\pi}(\mathcal A),
\]
and thus
\[
L_{\pi}\!\left(\Phi_{\mathrm{c}}(\mathcal A,F)\right)
=
\Sigma^*\setminus L_{\pi}(\mathcal A).
\]

Since $\Phi_{\mathrm{c}}$ depends only on the plant and its accepting
set, and not on the particular policy $\pi$, the complement is uniformly
realizable on $\mathfrak A_{\mathrm{cdet}}(\Sigma)$ for both
memoryless and finite-memory policies.
\end{proof}

\begin{remark}
The assumptions of a single initial state and total determinism in
Proposition~\ref{3.41} are essential for this direct complementation
construction. If the controlled automaton is nondeterministic, or has
several possible initial states, acceptance is existential over
executions, and simply replacing $F$ by $Q\setminus F$ does not in
general complement the accepted language. Similarly, if the transition
function is partial, a word for which no complete execution exists is
rejected both before and after the accepting set is complemented.

A finite deterministic controlled automaton with a single initial
state can, when $\Sigma$ and $C$ are finite, be completed in the usual
way by adjoining a nonaccepting sink state and directing every
previously undefined transition to that state. Thus the totality
assumption is not restrictive when complement is considered within this
deterministic setting.
\end{remark}

\begin{remark}
Propositions~\ref{3.40} and \ref{3.41} illustrate the distinction between standard and
uniform closure introduced in Definition~\ref{Uniform closure}. At the global level,
\[
\mathcal C_0=\mathcal C_F=\mathrm{REG}
\]
is closed under all classical regular-language operations. Proposition~\ref{3.40}
shows that union and intersection admit uniform plant-level
constructions for arbitrary finite controlled automata, whereas
Proposition~\ref{3.41} gives a uniform complementation construction under the
additional assumption of complete determinism with a single initial
state.

For the stack- and tape-augmented policy classes, the closure results
established above are statements about the global classes
\[
\mathcal C_{\mathrm{Stack}}=\mathrm{CFL}
\qquad\text{and}\qquad
\mathcal C_{\mathrm{Tape}}=\mathrm{REC}.
\]
No corresponding uniform plant-level closure claim is made here.
\end{remark}
\subsection{Controlled Automata on Infinite Words}

We extend the finite-word semantics of controlled automata to infinite
activity words using B\"uchi acceptance. Each activity word is treated as a fixed external input; no distribution
over activity words is assumed.
Policies retain the activity-first observation semantics of the underlying
controlled automata: the control action is selected after observing the
current state and activity. Acceptance is existential over genuine infinite
executions; a finite execution cannot satisfy the B\"uchi condition.

\begin{definition}[Controlled B\"uchi automaton and accepted $\omega$-language]
A \emph{controlled B\"uchi automaton} is a tuple
\[
\mathcal A_\omega=(Q,\Sigma,C,\delta,Q_0,F),
\]
where $(Q,\Sigma,C,\delta,Q_0)$ is a finite controlled automaton and
$F\subseteq Q$ is its set of accepting states. An infinite execution under
a policy $\pi$ is
\[
\rho:\quad q_0\xrightarrow{a_0,c_0}q_1
\xrightarrow{a_1,c_1}q_2\xrightarrow{a_2,c_2}\cdots,
\]
where $q_0\in Q_0$, $(q_i,a_i,c_i,q_{i+1})\in\delta$, and
$c_i=\pi((q_0,a_0)\cdots(q_i,a_i))$ for all $i\geq0$.
Write
\[
\operatorname{Inf}(\rho)=
\{q\in Q:q=q_i\text{ for infinitely many }i\}.
\]
The execution is \emph{B\"uchi accepting} if
$\operatorname{Inf}(\rho)\cap F\neq\varnothing$. The accepted
$\omega$-language is
\[
L^\omega_\pi(\mathcal A_\omega)
=
\{w\in\Sigma^\omega
\mid
\exists\text{ a B\"uchi-accepting execution
under $\pi$ labeled by }w\}.
\]
Acceptance is existential over the nondeterministic executions of the
plant. 
\end{definition}

For $\Pi\in\{0,F,\mathrm{Stack},\mathrm{Tape},H\}$, define
\[
\mathcal L^\omega_\Pi(\mathcal A_\omega)
=\{L^\omega_\pi(\mathcal A_\omega):\pi\in\Pi(\mathcal A_\omega)\},
\]
where $\Pi(\mathcal A_\omega)$ denotes policies of class $\Pi$ on the
underlying controlled automaton. The notation
$\mathcal L_H^{\omega,\mathrm{comp}}(\mathcal A_\omega)$ restricts
$\pi$ to total computable history-dependent policies.

\begin{definition}[B\"uchi-preserving equivalence]
Let $\mathcal A_\omega=(Q,\Sigma,C,\delta,Q_0,F)$ and
$\mathcal A'_\omega=(Q',\Sigma,C,\delta',Q'_0,F')$ have underlying
controlled automata related by a bisimulation $\gamma$ in the sense of
Definition~3.5. They are \emph{B\"uchi-equivalent under $\gamma$} if
\[
(q,q')\in\gamma\quad\Longrightarrow\quad
(q\in F\Longleftrightarrow q'\in F').
\]
\end{definition}

\begin{proposition}\label{3.49}
Let $\mathcal A_\omega$ and $\mathcal A'_\omega$ be B\"uchi-equivalent
under a bisimulation $\gamma$. If $\pi$ and $\pi'$ are equivalent policies
on their underlying controlled automata, as in
Definition~\ref{equivpolicy}, then
\[
L^\omega_\pi(\mathcal A_\omega)
=L^\omega_{\pi'}(\mathcal A'_\omega).
\]
Moreover, for any $\Pi\in\{0,F,\mathrm{Stack},\mathrm{Tape},H\}$, if
every policy of class $\Pi$ on either underlying controlled automaton has
an equivalent policy of the same class on the other, then
\[
\mathcal L^\omega_\Pi(\mathcal A_\omega)
=\mathcal L^\omega_\Pi(\mathcal A'_\omega).
\]
\end{proposition}
\begin{proof}
Given an accepting execution of $\mathcal A_\omega$ under $\pi$,
the initial-state and transition clauses of $\gamma$ yield an infinite
matching execution of $\mathcal A'_\omega$. Equivalent policies select
the same controls along corresponding observation histories. At each
position, the related states either both belong to their respective
accepting sets, or both do not. Hence the matching execution visits $F'$
infinitely often exactly when the original visits $F$ infinitely often.
This proves one language inclusion; symmetry gives equality.
For the second assertion, transport each policy realizing a member of
$\mathcal L^\omega_\Pi(\mathcal A_\omega)$ to an equivalent policy of
class $\Pi$ on $\mathcal A'_\omega$, and apply the first assertion.
The converse transport condition gives the reverse inclusion.
\end{proof}
\begin{definition}[Global $\omega$-language classes]  \label{3.54}
In analogy with Definition~\ref{Language families and global classes finite words} 
for finite words, consider the class 
$\mathfrak A^\omega_{\mathrm{fin}}(\Sigma)$
of finite controlled B\"uchi automata
$\mathcal A_\omega=(Q,\Sigma,C,\delta,Q_0,F)$,
where the finite state set $Q$ and finite nonempty
control-action set $C$ may vary.

For each policy class $\Pi$, define
\[
\mathfrak C^\omega_\Pi
=
\left\{
\mathcal L^\omega_\Pi(\mathcal A_\omega)
:
\mathcal A_\omega\in
\mathfrak A^\omega_{\mathrm{fin}}(\Sigma)
\right\},
\]
and
\[
\mathcal C^\omega_\Pi
=
\bigcup_{\mathcal A_\omega\in
\mathfrak A^\omega_{\mathrm{fin}}(\Sigma)}
\mathcal L^\omega_\Pi(\mathcal A_\omega).
\]
For any two policy classes $\Pi_1$ and $\Pi_2$, write
\[
\mathfrak C_{\Pi_1}^{\omega}
\preceq
\mathfrak C_{\Pi_2}^{\omega}
\]
if, for every
$A_\omega\in\mathfrak A_{\mathrm{fin}}^\omega(\Sigma)$,
\[
\mathcal L_{\Pi_1}^{\omega}(A_\omega)
\subseteq
\mathcal L_{\Pi_2}^{\omega}(A_\omega).
\]
Thus, $\preceq$ denotes pointwise inclusion
of the corresponding language classes.
\end{definition}

\begin{proposition}[Policy expressiveness on infinite words]
\label{3.55}
Fix a finite nonempty observable alphabet $\Sigma$.
For every finite controlled B\"uchi automaton
$\mathcal A_\omega\in\mathfrak A^\omega_{\mathrm{fin}}(\Sigma)$,
\[
\mathcal L^\omega_0(\mathcal A_\omega)
\subseteq
\mathcal L^\omega_F(\mathcal A_\omega)
\subseteq
\mathcal L^\omega_{\mathrm{Stack}}(\mathcal A_\omega)
\subseteq
\mathcal L^\omega_{\mathrm{Tape}}(\mathcal A_\omega)
\subseteq
\mathcal L^\omega_H(\mathcal A_\omega),
\]
and
\[
\mathcal L^\omega_{\mathrm{Tape}}(\mathcal A_\omega)
=
\mathcal L_H^{\omega,\mathrm{comp}}(\mathcal A_\omega).
\]

Consequently,
\[
\mathfrak C^\omega_0
\preceq
\mathfrak C^\omega_F
\preceq
\mathfrak C^\omega_{\mathrm{Stack}}
\preceq
\mathfrak C^\omega_{\mathrm{Tape}}
=
\mathfrak C_H^{\omega,\mathrm{comp}}
\preceq
\mathfrak C^\omega_H.
\]

At the global level,
\[
\omega\text{-}\mathrm{REG}
=
\mathcal C^\omega_0
=
\mathcal C^\omega_F
\subseteq
\mathcal C^\omega_{\mathrm{Stack}}
\subseteq
\mathcal C^\omega_{\mathrm{Tape}}
=
\mathcal C_H^{\omega,\mathrm{comp}}
\subseteq
\mathcal C^\omega_H,
\]
where $\omega\text{-}\mathrm{REG}$ denotes the class of
$\omega$-regular languages over $\Sigma$, and
$\mathcal C_H^{\omega,\mathrm{comp}}$
is the global language class induced by total
computable history-dependent policies.
No strictness is asserted for the displayed inclusions.
\end{proposition}

\begin{proof}
The pointwise inclusions follow from the inclusions
among the corresponding policy classes.
By Definition~\ref{3.20}, tape-augmented policies
are precisely the total computable history-dependent
policies. Hence, for every
$\mathcal A_\omega\in
\mathfrak A^\omega_{\mathrm{fin}}(\Sigma)$,
\[
\mathcal L^\omega_{\mathrm{Tape}}(\mathcal A_\omega)
=
\mathcal L_H^{\omega,\mathrm{comp}}(\mathcal A_\omega).
\]
Taking the corresponding collections of
automaton-specific language families gives
\[
\mathfrak C^\omega_{\mathrm{Tape}}
=
\mathfrak C_H^{\omega,\mathrm{comp}}.
\]
The remaining relations between these collections
follow from the pointwise inclusions.

For a fixed memoryless policy, the induced
transition relation defines a finite
nondeterministic B\"uchi automaton.
For a finite-memory policy, the product of
the controlled B\"uchi automaton with the
finite-state controller likewise yields a
finite nondeterministic B\"uchi automaton,
with the accepting condition inherited
from the plant.
Therefore,
\[
\mathcal C^\omega_0
\subseteq
\omega\text{-}\mathrm{REG},
\qquad
\mathcal C^\omega_F
\subseteq
\omega\text{-}\mathrm{REG}.
\]

Conversely, every finite nondeterministic
B\"uchi automaton over $\Sigma$ can be
represented as a finite controlled B\"uchi
automaton with a singleton control-action
set. Its unique memoryless policy induces
the original $\omega$-language. Thus,
\[
\omega\text{-}\mathrm{REG}
\subseteq
\mathcal C^\omega_0.
\]
Since
$\mathcal C^\omega_0\subseteq\mathcal C^\omega_F$,
we obtain
\[
\mathcal C^\omega_0
=
\mathcal C^\omega_F
=
\omega\text{-}\mathrm{REG}.
\]

Finally, taking unions of the pointwise
language-family inclusions over all automata
in $\mathfrak A^\omega_{\mathrm{fin}}(\Sigma)$
yields the remaining global inclusions.
The pointwise equality for tape-augmented
and computable history-dependent policies
similarly gives
\[
\mathcal C^\omega_{\mathrm{Tape}}
=
\mathcal C_H^{\omega,\mathrm{comp}}.
\]
This establishes all the assertions.
\end{proof}
For closure properties, we use the distinction made in the finite-word
setting. \emph{Standard closure} concerns the global language class;
\emph{uniform closure} requires an effective construction of a target
controlled automaton from the source automata alone, such that for every
choice of source policies of the specified class, some target policy of
that class realizes the desired language operation. In particular, the
target plant may not depend on the chosen source policies. This is the
infinite-word counterpart of the closure notions used earlier in Section~3.

\begin{theorem}[Uniform union and intersection]
For $\Pi\in\{0,F\}$, finite controlled B\"uchi automata are uniformly
closed under finite union and finite intersection with respect to
policies of class $\Pi$.
\end{theorem}
\begin{proof}
It suffices to treat two automata
$\mathcal A_{\omega,i}=(Q_i,\Sigma,C_i,\delta_i,Q_{0i},F_i)$,
$i=1,2$, with policies $\pi_i\in\Pi$ and accepted languages $L_i$.

\emph{Union.} Take disjoint copies of the two plants, with state set
$Q_1\uplus Q_2$, control set $C_1\uplus C_2$, initial states
$Q_{01}\uplus Q_{02}$, accepting states $F_1\uplus F_2$, and only the
tagged component transitions. On a history within component $i$, the
target policy applies $\pi_i$ and tags its control output with the
canonical injection $\iota_i:C_i\to C_1\uplus C_2$. The current state
identifies the component, so memorylessness is preserved. For
finite-memory policies, a controller with memory $M_1\times M_2$
updates only the memory component associated with the observed plant
component. Outputs and updates on other observation histories can be
defined arbitrarily to ensure totality. Every infinite execution stays
in its initial component, and therefore
\[
L^\omega_{\pi_\cup}(\mathcal A_{\omega,\cup})=L_1\cup L_2.
\]

\emph{Intersection.} Use the plant
\[
\mathcal A_{\omega,\cap}=
\bigl(Q_1\times Q_2\times\{1,2\},\Sigma,C_1\times C_2,
\delta_\cap,Q_{01}\times Q_{02}\times\{1\},F_\cap\bigr).
\]
The last state component records the next acceptance obligation. Put
\[
\operatorname{next}(q_1,q_2,r)=
\begin{cases}
2,&r=1\text{ and }q_1\in F_1,\\
1,&r=2\text{ and }q_2\in F_2,\\
r,&\text{otherwise},
\end{cases}
\qquad
F_\cap=Q_1\times F_2\times\{2\}.
\]
The transition relation contains exactly
\[
((q_1,q_2,r),a,(c_1,c_2),(q'_1,q'_2,r'))
\]
for which $(q_i,a,c_i,q'_i)\in\delta_i$ for $i=1,2$ and
$r'=\operatorname{next}(q_1,q_2,r)$. Each product execution projects
to an execution of both source plants on the same activity word.
Infinitely many visits to $F_\cap$ force infinitely many visits to
both $F_1$ and $F_2$, since successive visits to $F_\cap$ require an
intervening discharge of the $F_1$ obligation. Conversely, if both
projected executions visit their accepting sets infinitely often,
the obligation component alternates indefinitely, and $F_\cap$ is
visited infinitely often.

The target policy outputs the pair of source controls, using the
projected observation histories. For memoryless policies this is
\[
\pi_\cap((q_1,q_2,r),a)
=(\pi_1(q_1,a),\pi_2(q_2,a)).
\]
For finite-memory policies, the target controller uses memory
$M_1\times M_2$ and updates and outputs componentwise. Both
constructions are total and preserve the policy class. Hence
\[
L^\omega_{\pi_\cap}(\mathcal A_{\omega,\cap})=L_1\cap L_2.
\]
The union and intersection plants depend only on the source plants,
not on the chosen policies. Iteration gives finite union and
intersection.
\end{proof}

\begin{remark}
Since $\mathcal C^\omega_0=\mathcal C^\omega_F
=\omega\text{-}\mathrm{REG}$, these global classes are also
closed under complement and difference by the classical Boolean
closure of $\omega$-regular languages. This standard closure does not
establish uniform complementation for controlled B\"uchi automata;
in particular, simply replacing $F$ by $Q\setminus F$ does not in
general complement a B\"uchi language. No additional closure claim is
made here for stack-augmented, tape-augmented, or unrestricted
history-dependent policies.
\end{remark}
\section{Probabilistic Models}

We now extend controlled activity networks by resolving the choice
among enabled instantaneous activities probabilistically. As in
Section~3, control actions determine the effects of completed timed
activities. The probabilistic extension changes only the selection
mechanism for instantaneous activities; the underlying network
structure, enabling conditions, and gate-based marking updates remain
as defined in Section~3.

\subsection{Model Structure}
\begin{definition}\label{def:controlled-probabilistic-activity-network}
A \emph{controlled probabilistic activity network} is a pair
\[
L=(K,IP),
\]
where
\[
K=(P,IA,TA,CA,\Gamma,IG,OG,IR,IOR,TOR)
\]
is a controlled activity network as in Definition~\ref{def3.1}, and
\[
IP:\mathbb{N}^{n}\times IA\longrightarrow[0,1],
\qquad n=|P|,
\]
is an \emph{instantaneous-activity weight function}.

For a marking $\mu$, let
\[
IA_{\mathrm{en}}(\mu)
=
\{a\in IA\mid a\text{ is enabled in }\mu\}.
\]
We require
\[
IP(\mu,a)=0
\qquad
\text{whenever }a\notin IA_{\mathrm{en}}(\mu),
\]
and, for every unstable marking $\mu$,
\[
\sum_{a\in IA_{\mathrm{en}}(\mu)}IP(\mu,a)>0.
\]
\end{definition}

\subsection{Model Behavior}

The behavior of a controlled probabilistic activity network follows
the marking-update rules of Section~3. In a stable marking, the
selection of an enabled timed activity remains nondeterministic.
When a timed activity completes, a control action is selected
according to the applicable policy, and the input- and output-gate
functions determine the resulting marking.

In an unstable marking, instantaneous activities have priority over
timed activities. Unlike the nondeterministic model of Section~3,
the next instantaneous activity is selected probabilistically.

More precisely, let $L=(K,IP)$ be a controlled probabilistic activity
network, and let $\mu$ be an unstable marking. For each
$a\in IA_{\mathrm{en}}(\mu)$, define
\[
p_{\mu}(a)
=
\frac{IP(\mu,a)}
{\displaystyle
 \sum_{b\in IA_{\mathrm{en}}(\mu)}IP(\mu,b)}.
\]
The denominator is positive by Definition~
\ref{def:controlled-probabilistic-activity-network}; hence
\[
\sum_{a\in IA_{\mathrm{en}}(\mu)}p_{\mu}(a)=1.
\]
The activity $a$ is selected for completion with probability
$p_{\mu}(a)$. Once selected, its completion changes the marking
through the input- and output-gate functions defined in Section~3.

Thus, the probabilistic extension resolves the choice among enabled
instantaneous activities while retaining the nondeterministic
selection of timed activities and the policy-dependent choice of
control actions.

\subsection{Semantic Models}

To describe the resulting behavior formally, we introduce
controlled probabilistic automata. As in Definition~\ref{def3.4}, transitions
carry an observable activity and a control action; here, however,
their successor states are specified by probability distributions.
\begin{definition}\label{def:controlled-probabilistic-automaton}
A \emph{controlled probabilistic automaton} is a tuple
\[
\mathcal U=(Q,\Sigma,C,P,\nu_0),
\]
where:
\begin{itemize}
\item $Q$ is a finite or countable set of states;
\item $\Sigma$ is the observable activity alphabet;
\item $C$ is the set of control actions;
\item
\[
P:Q\times\Sigma\times C\times Q\longrightarrow[0,1]
\]
is the probabilistic transition function, satisfying
\[
\sum_{q'\in Q}P(q,a,c,q')\in\{0,1\}
\qquad
\text{for every }q\in Q,\ a\in\Sigma,\ c\in C;
\]
\item $\nu_0:Q\longrightarrow[0,1]$ is the initial-state
probability distribution, satisfying
\[
\sum_{q\in Q}\nu_0(q)=1.
\]
\end{itemize}

An observable activity $a\in\Sigma$ is said to be
\emph{enabled} in state $q\in Q$ if there exist
$c\in C$ and $q'\in Q$ such that
\[
P(q,a,c,q')>0.
\]
Similarly, an activity--control-action pair $(a,c)$ is said to be
\emph{enabled} in state $q$ if there exists $q'\in Q$ such that
\[
P(q,a,c,q')>0.
\]

We write
\[
q\xrightarrow[\alpha]{a,c}q'
\]
when
\[
P(q,a,c,q')=\alpha>0.
\]
Thus, after observable activity $a$ and control action $c$,
the successor state is $q'$ with probability $\alpha$.
\end{definition}
At the level of the CPA transition kernel, a transition from state $q_i$
is indexed by an observable activity $a_i$ and a control action $c_i$,
and the successor state $q_{i+1}$ is sampled according to
\[
P(q_i,a_i,c_i,\cdot).
\]
Thus, the kernel specifies the successor-state distribution conditional
on the realized activity and the control action associated with that activity.
The transition kernel itself does not specify how the activity or control
action is selected; the corresponding selection mechanisms are specified
separately.

\noindent\textbf{Example (Routing arrivals to parallel queues).}
Consider a service system with $m$ parallel queues and
state $q=(n_1,\ldots,n_m)$, where $n_j$ denotes the
number of customers in queue $j$. The arrival of a
customer is represented by the observable activity
$a_{\mathrm{arr}}$. For each queue $j$, let $c_j$ denote
the control action that routes the arriving customer to
queue $j$. In a simple deterministic model, if activity
$a_{\mathrm{arr}}$ occurs and control $c_j$ is associated
with that activity, the resulting state is
$q'=q+e_j$, where $e_j$ is the $j$-th unit vector.
Consequently,
\[
P(q,a_{\mathrm{arr}},c_j,q+e_j)=1.
\]
Thus, the CPA transition kernel specifies the successor-state distribution
conditional on the observable activity and its associated control action.
The mechanisms by which activities and control actions are selected are
specified separately.
\subsection{Bisimulation}

We next define a probabilistic counterpart of the bisimulation
notion introduced for controlled automata in Section~3.
The relation must preserve both the initial probability
distribution and the probability assigned to each equivalence
class of successor states.

\begin{definition}\label{def:probabilistic-bisimulation}
Let
\[
\mathcal U_1=(Q_1,\Sigma,C,\mathsf P_1,\nu_{0,1})
\quad\text{and}\quad
\mathcal U_2=(Q_2,\Sigma,C,\mathsf P_2,\nu_{0,2})
\]
be controlled probabilistic automata with the same observable
activity alphabet $\Sigma$ and control-action set $C$.
Assume, without loss of generality, that $Q_1$ and $Q_2$
are disjoint.

A \emph{probabilistic bisimulation} between $\mathcal U_1$
and $\mathcal U_2$ is an equivalence relation
$\sim$ on $Q_1\cup Q_2$ such that, for every equivalence
class $B$ of $\sim$, the following conditions hold:
\begin{enumerate}
\item The initial distributions agree on $B$:
\[
\sum_{q\in B\cap Q_1}\nu_{0,1}(q)
=
\sum_{q'\in B\cap Q_2}\nu_{0,2}(q').
\]

\item Whenever $q_1\in Q_1$ and $q_2\in Q_2$ satisfy
$q_1\sim q_2$, then, for every $a\in\Sigma$,
$c\in C$, and equivalence class $B$,
\[
\sum_{r\in B\cap Q_1}
\mathsf P_1(q_1,a,c,r)
=
\sum_{r'\in B\cap Q_2}
\mathsf P_2(q_2,a,c,r').
\]
\end{enumerate}

The automata $\mathcal U_1$ and $\mathcal U_2$ are
\emph{probabilistically bisimilar}, or \emph{equivalent},
if such a relation exists. They are \emph{isomorphic}
if there is a bijection $f:Q_1\to Q_2$ that preserves
the initial distributions and all transition probabilities:
\[
\nu_{0,1}(q)=\nu_{0,2}(f(q))
\qquad(q\in Q_1),
\]
and
\[
\mathsf P_1(q,a,c,r)
=
\mathsf P_2(f(q),a,c,f(r))
\]
for all $q,r\in Q_1$, $a\in\Sigma$, and $c\in C$.
\end{definition}

\begin{proposition}\label{prop:probabilistic-bisimilarity-equivalence}
Probabilistic bisimilarity is an equivalence relation on
controlled probabilistic automata with a common observable
activity alphabet and control-action set.
\end{proposition}

\begin{proof}
Reflexivity follows from the identity relation on the state
space of an automaton. Symmetry follows because the conditions
in Definition~\ref{def:probabilistic-bisimulation} are
symmetric in the two automata.

For transitivity, suppose that $\mathcal U_1$ is
probabilistically bisimilar to $\mathcal U_2$ and that
$\mathcal U_2$ is probabilistically bisimilar to
$\mathcal U_3$. Consider the disjoint union of their
state spaces and the equivalence relation generated by
the two bisimulations. The probability-matching
conditions are preserved under unions of equivalence
classes: summing the matching equalities first between
$\mathcal U_1$ and $\mathcal U_2$, and then between
$\mathcal U_2$ and $\mathcal U_3$, yields the corresponding
equalities between $\mathcal U_1$ and $\mathcal U_3$.
The same argument applies to the initial distributions.
The induced relation therefore witnesses probabilistic
bisimilarity of $\mathcal U_1$ and $\mathcal U_3$.
\end{proof}

\subsection{Operational Semantics}

We now formalize the behavior of a controlled probabilistic
activity network by associating with it a controlled probabilistic
automaton. As in the nondeterministic semantics of
Definition~\ref{def:can-realization}, the automaton records transitions between stable
markings induced by timed-activity completions. Here, the
instantaneous activities that follow a timed-activity completion
determine a probability distribution over the next stable marking
and the possibility of instantaneous divergence.

For a marking $\mu$, let $IA_{\mathrm{en}}(\mu)$ denote the set
of instantaneous activities enabled in $\mu$. If $\mu$ is
unstable, the probability that $b\in IA_{\mathrm{en}}(\mu)$
is selected is
\[
p_{\mu}(b)
=
\frac{IP(\mu,b)}
{\displaystyle\sum_{d\in IA_{\mathrm{en}}(\mu)}IP(\mu,d)}.
\]
Write $\mu\xrightarrow{b}\eta$ when completion of the
instantaneous activity $b$ changes marking $\mu$ to $\eta$
according to the gate functions of Section~3.

For markings $\mu$ and $\eta$, define
\[
R(\mu,\eta)
=
\sum_{\substack{b\in IA_{\mathrm{en}}(\mu)\\
                 \mu\xrightarrow{b}\eta}}
p_{\mu}(b)
\]
when $\mu$ is unstable. For a stable marking $\mu$, set
$R(\mu,\eta)=0$ for all $\eta$. Thus, $R$ describes one
probabilistic instantaneous-activity completion.

Let $S$ be the set of stable markings and, for $\eta\in S$,
let $H(\mu,\eta)$ be the probability that, starting at
$\mu$, a finite (possibly empty) sequence of instantaneous
activities first reaches the stable marking $\eta$.
Equivalently,
\[
H(\mu,\eta)
=
\begin{cases}
1, & \mu=\eta\in S,\\
0, & \mu\in S,\ \mu\ne\eta,\\[2mm]
\displaystyle\sum_{\zeta}R(\mu,\zeta)H(\zeta,\eta),
   & \mu\notin S.
\end{cases}
\]
These equations may admit multiple nonnegative solutions in the presence of instantaneous divergence; \(H\) is defined as their least nonnegative solution.
The probability of instantaneous divergence from $\mu$ is
\[
D(\mu)
=
1-\sum_{\eta\in S}H(\mu,\eta).
\]
In particular, $D(\mu)$ is the probability that the
instantaneous-activity sequence never reaches a stable marking.
The sum defining $H$ may be taken over all markings; only
markings reachable by instantaneous activities contribute.
\begin{definition}[Almost-surely well-behaved controlled probabilistic
activity network]
\label{def:almost-surely-well-behaved-cpan}
Let $L=(K,IP)$ be a controlled probabilistic activity network
with initial marking $\mu_0$. The initialized network $(L,\mu_0)$
is \emph{almost surely well behaved} if
\[
D(\mu)=0
\]
for $\mu=\mu_0$ and for every marking $\mu$ that can arise
immediately after a feasible timed-activity completion from a
reachable stable marking. Equivalently, from every such marking,
the ensuing sequence of instantaneous activities reaches a stable
marking with probability one.
\end{definition}
\begin{proposition}[Almost-sure well-behavedness]
\label{prop:almost-sure-well-behavedness}
Let $L=(K,IP)$ be a controlled probabilistic activity network
with initial marking $\mu_0$.

\begin{enumerate}
\item If the underlying initialized controlled activity network
$(K,\mu_0)$ is well behaved in the sense of
Definition~\ref{def:well-behaved-can}, then $(L,\mu_0)$ is
almost surely well behaved. The converse need not hold.

\item For the class of computable controlled probabilistic
activity networks, the problem of determining whether an
initialized network $(L,\mu_0)$ is almost surely well behaved
is undecidable.
\end{enumerate}
\end{proposition}

\begin{proof}
For the first statement, if $(K,\mu_0)$ is well behaved, then
every sequence of instantaneous-activity completions starting
from an operationally relevant marking reaches a stable marking
after finitely many steps. Hence instantaneous divergence is
impossible, and therefore
\[
D(\mu)=0
\]
for every such marking. Thus $(L,\mu_0)$ is almost surely well
behaved.

The converse need not hold because an infinite sequence of
instantaneous activities may be possible in the underlying
controlled activity network while having probability zero under
the instantaneous-activity selection probabilities.

For the second statement, consider the reduction used in the
proof of Theorem~\ref{thm:well-behavedness-undecidable}.
Given a Turing machine $M$ and an input $x$, construct the
instantaneous-activity simulation so that at each nonhalting
simulation marking exactly one instantaneous activity is enabled.
Assign this activity any positive instantaneous weight. Its
selection probability is then one. The instantaneous computation
reaches a stable marking if and only if $M$ halts on $x$.

Consequently,
\[
D(\mu_0)=0
\]
if $M$ halts on $x$, whereas
\[
D(\mu_0)=1
\]
if $M$ does not halt on $x$. Hence the constructed controlled
probabilistic activity network is almost surely well behaved if
and only if $M$ halts on $x$. A decision procedure for
almost-sure well-behavedness would therefore decide the halting
problem. Thus almost-sure well-behavedness is undecidable.
\end{proof}
\begin{definition}[Realization of a CPA] \label{def:probabilistic-operational-semantics}
Let $L=(K,IP)$ be a controlled probabilistic activity
network as in
Definition~\ref{def:controlled-probabilistic-activity-network},
with initial marking $\mu_0$. The pair $(L,\mu_0)$
\emph{realizes} the controlled probabilistic automaton
\[
\mathcal U_L=(Q,\Sigma,C,\mathsf P,\nu_0),
\]
defined as follows:
\begin{itemize}

\item $Q$ consists of the stable markings reachable from
$\mu_0$ through feasible activity completions, together
with a distinguished divergence state $\Delta$ whenever
instantaneous divergence has positive probability from
$\mu_0$ or from a marking reached after a feasible
timed-activity completion.

\item $\Sigma=TA$ is the observable activity alphabet,
and $C=CA$ is the set of control actions.

\item For a stable marking $\mu\in Q$, an
activity--control-action pair $(a,c)\in\Sigma\times C$
is said to be \emph{feasible at $\mu$} if the timed
activity $a$ is enabled at $\mu$ and completion of $a$
under $c$ is defined by the network's gate semantics.

For a pair $(a,c)$ feasible at $\mu$, let
$\operatorname{post}(\mu,a,c)$ denote the marking obtained
immediately after completion of $a$ under control
action $c$, before any subsequent instantaneous
activity completes. For each stable marking
$\eta\in Q\setminus\{\Delta\}$, set
\[
\mathsf P(\mu,a,c,\eta)
=
H\bigl(\operatorname{post}(\mu,a,c),\eta\bigr).
\]
If $\Delta\in Q$, set
\[
\mathsf P(\mu,a,c,\Delta)
=
D\bigl(\operatorname{post}(\mu,a,c)\bigr).
\]
Thus, for every pair $(a,c)$ feasible at $\mu$,
\[
\sum_{q'\in Q}\mathsf P(\mu,a,c,q')=1.
\]

If $(a,c)$ is not feasible at $\mu$, set
\[
\mathsf P(\mu,a,c,q')=0
\qquad
\text{for every }q'\in Q.
\]

\item The initial-state distribution is given, for
each stable marking $\eta\in Q\setminus\{\Delta\}$, by
\[
\nu_0(\eta)=H(\mu_0,\eta).
\]
If $\Delta\in Q$, set
\[
\nu_0(\Delta)=D(\mu_0).
\]
Consequently,
\[
\sum_{q\in Q}\nu_0(q)=1.
\]

\item If $\Delta\in Q$, then
\[
\mathsf P(\Delta,a,c,q')=0
\]
for every $a\in\Sigma$, $c\in C$, and $q'\in Q$.
Thus, $\Delta$ has no outgoing observable transitions.

\end{itemize}
\end{definition}
A notion of equivalence for controlled probabilistic
activity networks follows from the equivalence of their
realized automata.

\begin{definition}\label{def:probabilistic-network-equivalence}
Two controlled probabilistic activity networks with
initial markings $(L_1,\mu_{0,1})$ and
$(L_2,\mu_{0,2})$ are \emph{equivalent} if the controlled
probabilistic automata they realize are probabilistically
bisimilar in the sense of
Definition~\ref{def:probabilistic-bisimulation}.
\end{definition}

\begin{proposition}\label{prop:probabilistic-network-equivalence}
Equivalence of controlled probabilistic activity
networks with initial markings is an equivalence relation.
\end{proposition}

\begin{proof}
By Definition~
\ref{def:probabilistic-network-equivalence},
two networks with initial markings are equivalent
precisely when their realized controlled probabilistic
automata are probabilistically bisimilar. The result
therefore follows from
Proposition~\ref{prop:probabilistic-bisimilarity-equivalence}.
\end{proof}
\subsection{Policy Types}
We now extend the deterministic control policies of Section~3 to
probabilistic policies. The probabilistic behavior of a controlled
probabilistic activity network arises from the selection of enabled
instantaneous activities and does not, by itself, require the controller
to randomize. A probabilistic policy introduces a separate source of
randomness by specifying a probability distribution over control actions
after each observation.

As in Section~3, an observation consists of the current state and the
realized observable activity. Thus, after observing an activity in the
current state, the controller selects a control action, possibly at random.
The policies defined here therefore retain the activity-conditioned
information structure of Section~3 while allowing randomized control
selection. Deterministic policies are recovered as the special case in
which the selected control action has probability one.

For a nonempty finite set $C$, let $\operatorname{Dist}(C)$ denote the
set of probability distributions over $C$.The policy semantics adopted here treats control randomization
activity-wise: after an activity is realized, a control is
sampled for that activity alone.  This viewpoint will later
motivate the conditional-independence convention used for
provisional activity-specific control assignments in the
contingent-action semantics.
\begin{definition}[Probabilistic policy]
\label{def:probabilistic-policy}
Let
\[
\mathcal U=(Q,\Sigma,C,\mathsf P,\nu_0)
\]
be a controlled probabilistic automaton, and let
$O=Q\times\Sigma$ be its observation alphabet.

A \emph{history-dependent probabilistic policy} is a function
\[
\pi:O^*\longrightarrow\operatorname{Dist}(C).
\]

Consider a finite execution prefix with observation history
\[
h_i=(q_0,a_0)(q_1,a_1)\cdots(q_i,a_i),
\]
where $q_j\in Q$ and $a_j\in\Sigma$ for $0\leq j\leq i$.
After observing $(q_i,a_i)$, the controller selects a
control action $c_i$ according to the distribution
$\pi(h_i)$. Thus, for every $c\in C$,
\[
\Pr(c_i=c\mid h_i)=\pi(h_i)(c).
\]
For every such feasible execution history, the policy
must assign positive probability only to control actions
for which $(a_i,c)$ is enabled at $q_i$; that is,
\[
\pi(h_i)(c)>0
\quad\Longrightarrow\quad
(a_i,c)\text{ is enabled at }q_i.
\]
Equivalently, by Definition~\ref{def:controlled-probabilistic-automaton},
\[
\pi(h_i)(c)>0
\quad\Longrightarrow\quad
\sum_{q'\in Q}\mathsf P(q_i,a_i,c,q')=1.
\]

The policy is \emph{memoryless} if there exists a function
\[
\widehat{\pi}:Q\times\Sigma
\longrightarrow\operatorname{Dist}(C)
\]
such that, for every nonempty observation history $h_i$
ending in $(q_i,a_i)$,
\[
\pi(h_i)=\widehat{\pi}(q_i,a_i).
\]

A probabilistic policy is \emph{deterministic} if
$\pi(h)$ is a point-mass distribution for every
nonempty history $h\in O^*$. In this case, it can be
identified with a deterministic policy
$\pi_{\mathrm{det}}:O^*\to C$ as defined in
Section~3.

The value $\pi(\varepsilon)$ is independently fixed
and is not used to select control actions along
executions.
\end{definition}
In the following definitions, let
\[
\mathcal U=(Q,\Sigma,C,\mathsf P,\nu_0)
\]
be a controlled probabilistic automaton, let $O=Q\times\Sigma$,
and assume that $O$ and $C$ are finite and $C\neq\varnothing$.

A controller receives observations
$o_i=(q_i,a_i)\in O$ sequentially. After processing $o_i$,
it produces a probability distribution over $C$, from which
the control action $c_i$ is sampled. Along a feasible execution,
the admissibility requirement of
Definition~\ref{def:probabilistic-policy} ensures that
$(a_i,c_i)$ is enabled at $q_i$ whenever $c_i$ is selected
with positive probability. The successor state is then sampled
according to the probability distribution
\[
\mathsf P(q_i,a_i,c_i,\cdot).
\]

The controller retains its resulting configuration for the next
observation. Each controller below must produce an action for every
nonempty finite observation history, whether feasible
in $\mathcal A$ or not. The induced policy
$\pi:O^*\to C$ assigns an arbitrary fixed value to
$\pi(\varepsilon)$, which is not used along executions.

We next define probabilistic controllers with different
memory and computational resources. These controllers
share the observation and action-selection semantics
described above, but differ in how they store and
process observation histories.

\begin{definition}[Finite-memory probabilistic policy]
\label{def:finite-memory-probabilistic-policy}
A \emph{finite-memory probabilistic controller} for
$\mathcal U$ is a tuple
\[
\mathcal M=(M,m_0,\delta_M,\gamma_M),
\]
where $M$ is a finite set of memory states,
$m_0\in M$ is the initial memory state, and
\[
\delta_M:M\times O\longrightarrow M,
\qquad
\gamma_M:O\times M\longrightarrow
\operatorname{Dist}(C)
\]
are total functions.

On receiving $o_i=(q_i,a_i)$ in memory state
$m_i$, the controller produces the action
distribution
\[
\rho_i=\gamma_M(o_i,m_i)
\]
and updates its memory to
\[
m_{i+1}=\delta_M(m_i,o_i).
\]
The control action $c_i$ is sampled according
to $\rho_i$; thus, for every $c\in C$,
\[
\Pr(c_i=c\mid o_0\cdots o_i)
=
\rho_i(c).
\]
For each nonempty observation history
$h=o_0\cdots o_i$, the induced policy is
\[
\pi(h)=\rho_i.
\]
A history-dependent probabilistic policy is
called \emph{finite-memory} if a finite-memory probabilistic controller induces it.
\end{definition}

For the next two definitions, an
\emph{observation transition} consumes one
observation $o\in O$, whereas an
\emph{internal transition}, or
$\varepsilon$-transition, consumes no
observation. After one observation transition,
the controller executes internal transitions
until none is enabled and then produces its
action distribution.

The transition functions below are
deterministic and partial. Whenever an
observation is presented at an observation
boundary, its observation transition must
be defined. Every ensuing internal computation
must terminate after finitely many steps at a
configuration where the output distribution
is defined. These requirements apply to
every finite observation history.

\begin{definition}[Stack-augmented probabilistic policy]
\label{def:stack-augmented-probabilistic-policy}
A \emph{stack-augmented probabilistic controller}
for $\mathcal U$ is a tuple
\[
\mathcal S=
(M,\Gamma,\bot,m_0,s_0,
 \delta_S,\delta_S^\varepsilon,\gamma_S),
\]
where $M$ is a finite set of controller states,
$\Gamma$ is a finite stack alphabet containing
a bottom-of-stack symbol $\bot$,
$m_0\in M$ is the initial controller state,
and
\[
s_0\in(\Gamma\setminus\{\bot\})^*\bot
\]
is the specified finite initial stack.

Its partial deterministic transition functions
are
\[
\delta_S:
M\times\Gamma\times O
\rightharpoonup M\times\Gamma^*,
\]
\[
\delta_S^\varepsilon:
M\times\Gamma
\rightharpoonup M\times\Gamma^*.
\]
A transition replaces the top stack symbol
by a finite word and changes the controller
state; it may neither inspect nor modify
the remainder of the stack. The bottom
symbol is preserved, so the stack is
never empty. Initially, the configuration
is $(m_0,s_0)$.

The partial terminal-output function is
\[
\gamma_S:
O\times M\times\Gamma
\rightharpoonup\operatorname{Dist}(C).
\]
It is evaluated after all internal transitions
have terminated and depends only on the
current observation, terminal controller
state, and top stack symbol. The function
$\gamma_S$ performs no additional stack
inspection or modification.

On receiving $o_i$, the controller applies
$\delta_S$ to its current state and top
stack symbol, then repeatedly applies
$\delta_S^\varepsilon$ until it is undefined.
If the resulting terminal configuration
is $(m,s)$, the controller produces
\[
\rho_i=
\gamma_S(o_i,m,\operatorname{top}(s)).
\]
The control action $c_i$ is sampled
according to $\rho_i$, so that
\[
\Pr(c_i=c\mid o_0\cdots o_i)
=
\rho_i(c),
\qquad c\in C.
\]

The terminal configuration persists until
the next observation. For every nonempty
finite observation history, the specified
observation transitions must exist, the
internal computation must terminate,
and the terminal-output distribution
must be defined.

For each nonempty history $h=o_0\cdots o_i$,
the induced policy is
\[
\pi(h)=\rho_i.
\]
A history-dependent probabilistic policy
is called \emph{stack-augmented} if it is
induced by a stack-augmented probabilistic
controller.
\end{definition}

\begin{definition}[Tape-augmented probabilistic policy]
\label{def:tape-augmented-probabilistic-policy}
A \emph{tape-augmented probabilistic controller}
for $\mathcal U$ is a tuple
\[
\mathcal T=
(M,\Gamma,\sqcup,m_0,\tau_0,j_0,
 \delta_T,\delta_T^\varepsilon,\gamma_T),
\]
where $M$ is a finite set of controller
states, $\Gamma$ is a finite tape alphabet
containing a blank symbol $\sqcup$,
$m_0\in M$ is the initial controller state,
$\tau_0$ specifies the initial tape contents,
and $j_0$ is the initial head position.

The tape is an ordinary unbounded
Turing-machine tape, initially blank
except for a finite specified region.
Its partial deterministic transition
functions are
\[
\delta_T:
M\times\Gamma\times O
\rightharpoonup
M\times\Gamma\times\{L,R,S\},
\]
\[
\delta_T^\varepsilon:
M\times\Gamma
\rightharpoonup
M\times\Gamma\times\{L,R,S\}.
\]
Each transition changes the controller
state, overwrites the currently scanned
symbol, and moves the tape head one cell
left, one cell right, or not at all.
An observation transition also consumes
the current observation.

The partial terminal-output function
\[
\gamma_T:
O\times M\times\Gamma
\rightharpoonup\operatorname{Dist}(C)
\]
is a finite table evaluated using the
current observation, terminal controller
state, and scanned tape symbol.

On receiving $o_i$, the controller
applies $\delta_T$ and then repeatedly
applies $\delta_T^\varepsilon$ until
it is undefined. At the resulting
terminal configuration $(m,\tau,j)$,
it produces
\[
\rho_i=\gamma_T(o_i,m,\tau(j)).
\]
The control action $c_i$ is sampled
according to $\rho_i$; hence
\[
\Pr(c_i=c\mid o_0\cdots o_i)
=
\rho_i(c),
\qquad c\in C.
\]

The resulting controller state, tape
contents, and head position persist
until the next observation. For every
nonempty finite observation history,
the specified observation transitions
must exist, the internal computation
must terminate, and the terminal-output
distribution must be defined.

For each nonempty history $h=o_0\cdots o_i$,
the induced policy is
\[
\pi(h)=\rho_i.
\]
A history-dependent probabilistic policy
is called \emph{tape-augmented} if it is
induced by a tape-augmented probabilistic
controller.
\end{definition}
\begin{proposition}[Conditional probability weight of a finite run]
\label{prop:probability-finite-run}
Let
\[
\mathcal U=(Q,\Sigma,C,\mathsf P,\nu_0)
\]
be a controlled probabilistic automaton, and let
\[
\pi:(Q\times\Sigma)^*
    \longrightarrow\operatorname{Dist}(C)
\]
be a history-dependent probabilistic policy.

Consider a finite run of length $n\geq 1$,
\[
\rho=(q_0,a_0)(q_1,a_1)\cdots
     (q_{n-1},a_{n-1})q_n,
\]
where $q_i\in Q$ for $0\leq i\leq n$
and $a_i\in\Sigma$ for $0\leq i<n$.
For each $0\leq i<n$, let
\[
h_i=(q_0,a_0)(q_1,a_1)\cdots(q_i,a_i)
\]
be the corresponding nonempty observation history.

The conditional probability weight of $\rho$ under $\pi$,
given the specified sequence of observable activities
$a_0,\ldots,a_{n-1}$, is
\[
W_\pi(\rho)
=
\nu_0(q_0)
\prod_{i=0}^{n-1}
\left(
\sum_{c\in C}
\pi(h_i)(c)\,
\mathsf P(q_i,a_i,c,q_{i+1})
\right).
\]
For a run of length zero, consisting of a single state $q_0$,
set
\[
W_\pi(q_0)=\nu_0(q_0).
\]

The quantity $W_\pi(\rho)$ is conditional on the specified
observable activities and does not assign probability to their
selection or occurrence.
\end{proposition}

\begin{proof}
For each $0\leq i<n$, after the observation history $h_i$,
the controller selects a control action $c\in C$ with probability
$\pi(h_i)(c)$. Conditional on the realized activity $a_i$ and
the selected control action $c$, the successor state is
$q_{i+1}$ with probability
\[
\mathsf P(q_i,a_i,c,q_{i+1}).
\]
Marginalizing over the possible control actions therefore gives
the conditional one-step probability weight
\[
\sum_{c\in C}
\pi(h_i)(c)\,
\mathsf P(q_i,a_i,c,q_{i+1}).
\]
Multiplying these one-step weights along the finite run and
including the initial-state probability $\nu_0(q_0)$ yields
the stated expression.
\end{proof}

\subsection{Probabilistic Languages on Finite Words}

We now consider language recognition by controlled probabilistic
automata. Throughout this subsection, assume that $Q$ and $\Sigma$
are finite and that $F\subseteq Q$ is a designated set of
accepting states.

For a fixed activity word, the acceptance probability is obtained
by summing the conditional probability weights of all runs labeled
by that word that end in an accepting state. The activity word is
treated as given; no probability is assigned here to its selection
or occurrence.

\begin{definition}[Probabilistic language on finite words]
\label{def:probabilistic-language-finite}
Let
\[
\mathcal U=(Q,\Sigma,C,\mathsf P,\nu_0)
\]
be a controlled probabilistic automaton with accepting set
$F\subseteq Q$, and let $\pi$ be a probabilistic policy on
$\mathcal U$.

For a word
\[
w=a_0a_1\cdots a_{n-1}\in\Sigma^*,
\]
let $\operatorname{AR}(w)$ denote the set of accepting runs
labeled by $w$:
\[
\operatorname{AR}(w)
=
\left\{
(q_0,a_0)\cdots(q_{n-1},a_{n-1})q_n
\;\middle|\;
q_0,\ldots,q_n\in Q,\;
q_n\in F
\right\}.
\]

Define the acceptance probability of the fixed word $w$ under
$\pi$ by
\[
\operatorname{Acc}_{\pi,\mathcal U}(w)
=
\sum_{\rho\in\operatorname{AR}(w)}
W_\pi(\rho),
\]
where $W_\pi(\rho)$ is the conditional probability weight
specified in
Proposition~\ref{prop:probability-finite-run}.

For the empty word, set
\[
\operatorname{Acc}_{\pi,\mathcal U}(\varepsilon)
=
\sum_{q\in F}\nu_0(q).
\]

For a threshold $\theta\in(0,1]$, the probabilistic language
accepted by $\mathcal U$ under $\pi$ at threshold $\theta$ is
\[
L_\pi^\theta(\mathcal U)
=
\left\{
w\in\Sigma^*
\;\middle|\;
\operatorname{Acc}_{\pi,\mathcal U}(w)\geq\theta
\right\}.
\]

The \emph{positive-acceptance language} is
\[
L_\pi^{>0}(\mathcal U)
=
\left\{
w\in\Sigma^*
\;\middle|\;
\operatorname{Acc}_{\pi,\mathcal U}(w)>0
\right\},
\]
and the \emph{almost-sure-acceptance language} is
\[
L_\pi^{=1}(\mathcal U)
=
\left\{
w\in\Sigma^*
\;\middle|\;
\operatorname{Acc}_{\pi,\mathcal U}(w)=1
\right\}
=
L_\pi^1(\mathcal U).
\]
\end{definition}
\begin{proposition}[Language families induced by policy classes]
\label{prop:probabilistic-policy-language-families}
Let
\[
\mathcal U=(Q,\Sigma,C,\mathsf P,\nu_0)
\]
be a controlled probabilistic automaton with finite
state set $Q$ and accepting set $F\subseteq Q$.
Fix $\theta\in(0,1]$.

For
\[
X\in\{0,F,\mathrm{Stack},\mathrm{Tape},H\},
\]
let $\Pi_X^{\mathrm D}(\mathcal U)$ and
$\Pi_X^{\mathrm P}(\mathcal U)$ denote,
respectively, the deterministic and probabilistic
policies of type $X$ on $\mathcal U$.
Define
\[
\mathcal L_X^\theta(\mathcal U)
=
\bigl\{
L_\pi^\theta(\mathcal U)
\mid
\pi\in\Pi_X^{\mathrm D}(\mathcal U)
\bigr\},
\]
and
\[
\mathcal L_{X,\mathrm P}^\theta(\mathcal U)
=
\bigl\{
L_\pi^\theta(\mathcal U)
\mid
\pi\in\Pi_X^{\mathrm P}(\mathcal U)
\bigr\}.
\]

Then
\[
\mathcal L_0^\theta(\mathcal U)
\subseteq
\mathcal L_F^\theta(\mathcal U)
\subseteq
\mathcal L_{\mathrm{Stack}}^\theta(\mathcal U)
\subseteq
\mathcal L_{\mathrm{Tape}}^\theta(\mathcal U)
\subseteq
\mathcal L_H^\theta(\mathcal U),
\]
and
\[
\mathcal L_{0,\mathrm P}^\theta(\mathcal U)
\subseteq
\mathcal L_{F,\mathrm P}^\theta(\mathcal U)
\subseteq
\mathcal L_{\mathrm{Stack},\mathrm P}^\theta(\mathcal U)
\subseteq
\mathcal L_{\mathrm{Tape},\mathrm P}^\theta(\mathcal U)
\subseteq
\mathcal L_{H,\mathrm P}^\theta(\mathcal U).
\]

Moreover, identifying a deterministic action $c$
with the Dirac distribution $\delta_c$ gives
\[
\mathcal L_X^\theta(\mathcal U)
\subseteq
\mathcal L_{X,\mathrm P}^\theta(\mathcal U)
\qquad
\text{for every }X.
\]

The same inclusions hold when the threshold
language $L_\pi^\theta(\mathcal U)$ is replaced
throughout by the positive-acceptance language
$L_\pi^{>0}(\mathcal U)$.
\end{proposition}

\begin{proof}
Each successive deterministic policy class
contains the preceding class, since its controller
can simulate a controller with less memory or
computational capability. The same argument
applies to the probabilistic policy classes.
Taking the images of these policy sets under
the map
\[
\pi\longmapsto L_\pi^\theta(\mathcal U)
\]
gives the first two chains of inclusions.

A deterministic policy can be represented as
a probabilistic policy by assigning probability
one to its selected control action at every
history. This representation preserves the
probability weight of every run and therefore
the acceptance probability of every word.
Consequently,
\[
\mathcal L_X^\theta(\mathcal U)
\subseteq
\mathcal L_{X,\mathrm P}^\theta(\mathcal U).
\]
The argument does not depend on the chosen
positive threshold and applies equally to
positive-acceptance semantics.
\end{proof}

\begin{definition}[Acceptance-preserving probabilistic bisimulation]
\label{def:acceptance-preserving-prob-bisimulation}
Let
\[
\mathcal U=(Q,\Sigma,C,\mathsf P,\nu_0)
\quad\text{and}\quad
\mathcal U'=(Q',\Sigma,C,\mathsf P',\nu_0')
\]
be controlled probabilistic automata with
accepting sets $F\subseteq Q$ and
$F'\subseteq Q'$, respectively.

The automata with accepting sets
$(\mathcal U,F)$ and $(\mathcal U',F')$
are \emph{acceptance-preserving probabilistically
bisimilar} if there exists a probabilistic
bisimulation $\sim$ between $\mathcal U$
and $\mathcal U'$ in the sense of
Definition~\ref{def:probabilistic-bisimulation} such that
\[
q\sim q'
\quad\Longrightarrow\quad
\bigl(q\in F\iff q'\in F'\bigr)
\]
for every $q\in Q$ and $q'\in Q'$.

Thus, each equivalence class of $\sim$
is either entirely accepting or entirely
nonaccepting, when restricted to the
states of the two automata.
\end{definition}

For the following proposition, a pair of policies
$\pi$ on $\mathcal U$ and $\pi'$ on $\mathcal U'$
is called \emph{$\sim$-compatible} if, whenever
\[
h=(q_0,a_0)\cdots(q_i,a_i)
\quad\text{and}\quad
h'=(q_0',a_0)\cdots(q_i',a_i)
\]
have the same activity sequence and satisfy
$q_j\sim q_j'$ for all $0\leq j\leq i$,
their control-action distributions agree:
\[
\pi(h)=\pi'(h').
\]
A deterministic policy is identified here with
its Dirac-valued probabilistic representation.

\begin{proposition}[Preservation of probabilistic languages]
\label{prop:probabilistic-language-preservation}
Let $(\mathcal U,F)$ and $(\mathcal U',F')$
be acceptance-preserving probabilistically
bisimilar under $\sim$, as in
Definition~\ref{def:acceptance-preserving-prob-bisimulation}.
If $\pi$ and $\pi'$ are $\sim$-compatible
probabilistic policies, then
\[
\operatorname{Acc}_{\pi,\mathcal U}(w)
=
\operatorname{Acc}_{\pi',\mathcal U'}(w)
\qquad
\text{for every }w\in\Sigma^*.
\]
Consequently, for every $\theta\in(0,1]$,
\[
L_\pi^\theta(\mathcal U)
=
L_{\pi'}^\theta(\mathcal U'),
\]
and
\[
L_\pi^{>0}(\mathcal U)
=
L_{\pi'}^{>0}(\mathcal U').
\]

Furthermore, fix
\[
X\in\{0,F,\mathrm{Stack},\mathrm{Tape},H\}.
\]
If every policy in
$\Pi_X^{\mathrm P}(\mathcal U)$ has a
$\sim$-compatible counterpart in
$\Pi_X^{\mathrm P}(\mathcal U')$,
and conversely, then
\[
\mathcal L_{X,\mathrm P}^\theta(\mathcal U)
=
\mathcal L_{X,\mathrm P}^\theta(\mathcal U').
\]
The analogous statement holds for
$\mathcal L_X^\theta$ if the counterpart
condition holds within the deterministic
policy class $\Pi_X^{\mathrm D}$.
All these family-equality statements also
hold under positive-acceptance semantics.
\end{proposition}

\begin{proof}
Fix $w=a_0\cdots a_{n-1}$.
Probabilistic bisimulation preserves the
initial probability mass of each
$\sim$-equivalence class and, for each
activity $a$ and control action $c$,
the transition probability into every
such class.

Because $\pi$ and $\pi'$ are
$\sim$-compatible, they assign the same
distribution over control actions after
corresponding observation histories.
Therefore, the induced transition
probabilities between equivalence classes
are equal at each step of $w$.

By induction on the length of $w$, the
two automata assign equal total probability
to corresponding equivalence classes
after processing $w$. Since acceptance
is constant on each equivalence class,
summing the probabilities of the accepting
classes gives
\[
\operatorname{Acc}_{\pi,\mathcal U}(w)
=
\operatorname{Acc}_{\pi',\mathcal U'}(w).
\]
For $w=\varepsilon$, the same equality
follows directly from preservation of
initial probability mass and acceptance.

Equality of the threshold and
positive-acceptance languages follows
from their definitions.

Finally, under the stated two-way
counterpart condition, each language
generated by a policy of type $X$ on
a policy generates one automaton
of the same type on the other.
Applying the established language equality
in both directions proves equality of
the corresponding language families.
\end{proof}
\subsection{Language Families and Global Classes for
Controlled Probabilistic Automata on Finite Words}

We now introduce for controlled probabilistic automata
the same two levels of language expressiveness used
in Section~3. First, we collect the automaton-specific
language families induced by each policy class.
Second, by taking the union of these families over
all finite controlled probabilistic automata, we
obtain global classes whose elements are individual
finite-word languages.

The former notion captures the expressive power
of different policy classes on a fixed plant,
whereas the latter characterizes their overall
language-theoretic expressive power. We distinguish
deterministic from probabilistic policies and
consider both threshold and positive-acceptance
semantics.

We follow the notation conventions introduced in
Section~\ref{Policy Types Nondeterministic}. 
The classes of finite controlled
probabilistic automata and finite controlled
probabilistic B\"uchi automata over $\Sigma$ are
denoted by $\mathfrak U_{\mathrm{fin}}(\Sigma)$
and $\mathfrak U^\omega_{\mathrm{fin}}(\Sigma)$,
respectively.

For probabilistic language semantics, the superscript
$\alpha$ specifies the acceptance criterion, while
the additional subscript $P$ distinguishes
probabilistic policies from deterministic policies.
\begin{definition}[Language families and global classes]
\label{def:global-cpa-language-classes}
Consider the class $\mathfrak U_{\mathrm{fin}}(\Sigma)$
of finite controlled probabilistic automata
\[
\mathcal U=(Q,\Sigma,C,P,\nu_0)
\]
equipped with a designated accepting set $F\subseteq Q$.

For each threshold $\theta\in(0,1]$ and each
policy class
\[
X\in\{0,F,\mathrm{Stack},\mathrm{Tape},H\},
\]
define the collections of automaton-specific
language families
\[
\mathcal C_X^\theta
=
\left\{
\mathcal L_X^\theta(\mathcal U)
\;\middle|\;
\mathcal U\in
\mathfrak U_{\mathrm{fin}}(\Sigma)
\right\},
\]
and
\[
\mathcal C_{X,\mathrm P}^\theta
=
\left\{
\mathcal L_{X,\mathrm P}^\theta(\mathcal U)
\;\middle|\;
\mathcal U\in
\mathfrak U_{\mathrm{fin}}(\Sigma)
\right\}.
\]

The corresponding global language classes are
\[
\mathbf C_X^\theta
=
\bigcup_{\mathcal U\in
\mathfrak U_{\mathrm{fin}}(\Sigma)}
\mathcal L_X^\theta(\mathcal U)
=
\bigcup\mathcal C_X^\theta,
\]
and
\[
\mathbf C_{X,\mathrm P}^\theta
=
\bigcup_{\mathcal U\in
\mathfrak U_{\mathrm{fin}}(\Sigma)}
\mathcal L_{X,\mathrm P}^\theta(\mathcal U)
=
\bigcup\mathcal C_{X,\mathrm P}^\theta.
\]

Thus, $\mathcal C_X^\theta$ and
$\mathcal C_{X,\mathrm P}^\theta$ are
collections of automaton-specific language
families, whereas $\mathbf C_X^\theta$
and $\mathbf C_{X,\mathrm P}^\theta$
are classes of individual languages
over $\Sigma$.

For two policy classes $X$ and $Y$, write
\[
\mathcal C_X^\theta
\preceq
\mathcal C_Y^\theta
\]
if, for every
$\mathcal U\in
\mathfrak U_{\mathrm{fin}}(\Sigma)$,
\[
\mathcal L_X^\theta(\mathcal U)
\subseteq
\mathcal L_Y^\theta(\mathcal U).
\]

Write
\[
\mathcal C_X^\theta
\prec
\mathcal C_Y^\theta
\]
if
\[
\mathcal C_X^\theta
\preceq
\mathcal C_Y^\theta
\]
and there exists an automaton
$\mathcal U\in
\mathfrak U_{\mathrm{fin}}(\Sigma)$
such that
\[
\mathcal L_X^\theta(\mathcal U)
\subsetneq
\mathcal L_Y^\theta(\mathcal U).
\]

The relations $\preceq$ and $\prec$
are defined analogously for the
probabilistic-policy collections
$\mathcal C_{X,\mathrm P}^\theta$.

Finally, replacing $L_\pi^\theta$
throughout by $L_\pi^{>0}$ defines
the corresponding collections and
global classes under positive-acceptance
semantics.
\end{definition}

\begin{proposition}[Inclusions among language families
and global classes]
\label{prop:global-cpa-language-inclusions}

For every $\theta\in(0,1]$,
\[
\mathcal C_0^\theta
\preceq
\mathcal C_F^\theta
\preceq
\mathcal C_{\mathrm{Stack}}^\theta
\preceq
\mathcal C_{\mathrm{Tape}}^\theta
\preceq
\mathcal C_H^\theta,
\]
and
\[
\mathcal C_{0,\mathrm P}^\theta
\preceq
\mathcal C_{F,\mathrm P}^\theta
\preceq
\mathcal C_{\mathrm{Stack},\mathrm P}^\theta
\preceq
\mathcal C_{\mathrm{Tape},\mathrm P}^\theta
\preceq
\mathcal C_{H,\mathrm P}^\theta.
\]

Consequently,
\[
\mathbf C_0^\theta
\subseteq
\mathbf C_F^\theta
\subseteq
\mathbf C_{\mathrm{Stack}}^\theta
\subseteq
\mathbf C_{\mathrm{Tape}}^\theta
\subseteq
\mathbf C_H^\theta,
\]
and
\[
\mathbf C_{0,\mathrm P}^\theta
\subseteq
\mathbf C_{F,\mathrm P}^\theta
\subseteq
\mathbf C_{\mathrm{Stack},\mathrm P}^\theta
\subseteq
\mathbf C_{\mathrm{Tape},\mathrm P}^\theta
\subseteq
\mathbf C_{H,\mathrm P}^\theta.
\]

Moreover, for every policy class $X$,
\[
\mathcal L_X^\theta(\mathcal U)
\subseteq
\mathcal L_{X,\mathrm P}^\theta(\mathcal U)
\]
for every
$\mathcal U\in
\mathfrak U_{\mathrm{fin}}(\Sigma)$,
and hence
\[
\mathbf C_X^\theta
\subseteq
\mathbf C_{X,\mathrm P}^\theta.
\]

All these inclusions also hold under
positive-acceptance semantics.
\end{proposition}

\begin{proof}
Proposition~
\ref{prop:probabilistic-policy-language-families}
establishes the corresponding inclusions
between policy-induced language families
for each fixed controlled probabilistic
automaton.

By the definition of $\preceq$, these
automaton-specific inclusions yield the
two chains of relations between the
collections $\mathcal C_X^\theta$ and
$\mathcal C_{X,\mathrm P}^\theta$.

Taking unions over all automata in
$\mathfrak U_{\mathrm{fin}}(\Sigma)$
gives the stated inclusions between
the global language classes.

Every deterministic policy is represented
by a probabilistic policy that assigns
probability one to its selected control
action. This representation preserves
the acceptance probability of every
word, proving the remaining inclusions.

The same reasoning applies to
positive-acceptance semantics.
\end{proof}

\begin{proposition}[Strict separation of memoryless and
finite-memory policies on a fixed automaton]
\label{prop:strict-memoryless-finite-probabilistic}
Let $\Sigma=\{a\}$ and suppose that $C$
contains two distinct actions. For every
$\theta\in(0,1]$, there exists a finite
controlled probabilistic automaton
$\mathcal U$ such that
\[
\mathcal L_{0,\mathrm P}^{\theta}(\mathcal U)
\subsetneq
\mathcal L_{F,\mathrm P}^{\theta}(\mathcal U).
\]
Consequently, in the automaton-specific
ordering of
Definition~\ref{def:global-cpa-language-classes},
\[
\mathcal C_{0,\mathrm P}^{\theta}
\prec
\mathcal C_{F,\mathrm P}^{\theta}.
\]
The corresponding strict relation also
holds between the deterministic-policy
collections:
\[
\mathcal C_0^\theta
\prec
\mathcal C_F^\theta.
\]
\end{proposition}

\begin{proof}
Choose distinct control actions
$c_{\mathrm{go}},c_{\mathrm{return}}\in C$.
Construct a controlled probabilistic
automaton with
\[
Q=\{q_0,q_1\},\qquad
\nu_0=\delta_{q_0},\qquad
F=\{q_0\}.
\]
For each $q\in Q$, define
\[
\mathsf P(q,a,c_{\mathrm{go}},q_1)=1,
\qquad
\mathsf P(q,a,c_{\mathrm{return}},q_0)=1.
\]
Any additional actions, if present,
have the same effect as
$c_{\mathrm{go}}$. Thus every
action is enabled, and all plant
transitions are deterministic.

Consider a deterministic finite-memory
policy $\pi_3$ whose memory records the
number of observations modulo $3$.
It selects $c_{\mathrm{return}}$
on every third observation and
$c_{\mathrm{go}}$ on all other
observations. Its memory update is
independent of the observed plant state.

After processing $a^n$, the plant is
in $q_0$ exactly when $3$ divides $n$,
including $n=0$. Hence
\[
\operatorname{Acc}_{\pi_3,\mathcal U}(a^n)
=
\begin{cases}
1,&3\mid n,\\
0,&3\nmid n,
\end{cases}
\]
and therefore
\[
L_{\pi_3}^{\theta}(\mathcal U)
=(a^3)^*
\qquad
\text{for every }\theta\in(0,1].
\]

We show that no memoryless probabilistic
policy recognizes $(a^3)^*$ on this
automaton at threshold $\theta$.

Fix such a policy $\pi$. Because the
alphabet is unary, its action
distributions depend only on the
current plant state. Let $r_j$ be
the total probability of selecting
actions that lead to $q_0$ when
the current state is $q_j$, for
$j\in\{0,1\}$. The induced
state-transition matrix is
\[
T=
\begin{pmatrix}
r_0&1-r_0\\
r_1&1-r_1
\end{pmatrix}.
\]
Write
\[
p_n=\operatorname{Acc}_{\pi,\mathcal U}(a^n).
\]
Since $p_0=1$, the transition matrix
gives
\[
p_{n+1}
=r_1+(r_0-r_1)p_n.
\]
Put $\lambda=r_0-r_1\in[-1,1]$.

If $\lambda=1$, then necessarily
$r_0=1$ and $r_1=0$, so $p_n=1$
for all $n$. The resulting threshold
language is $\Sigma^*$, not $(a^3)^*$.

If $\lambda=-1$, then necessarily
$r_0=0$ and $r_1=1$, so
$p_n=1$ for even $n$ and $p_n=0$
for odd $n$. The resulting threshold
language is $(a^2)^*$, not $(a^3)^*$.

Finally, suppose $|\lambda|<1$.
The recurrence has the unique
fixed point
\[
p_*=\frac{r_1}{1-\lambda},
\]
and
\[
p_n-p_*=(1-p_*)\lambda^n.
\]
If $\theta\neq p_*$, the threshold
comparison $p_n\geq\theta$ is
eventually constant. If
$\theta=p_*$, its truth value is
eventually constant on each parity
class (and is constant for all
$n\geq1$ when $\lambda=0$).
In either case, the resulting
threshold language cannot be
$(a^3)^*$, whose membership is
not eventually constant on either
parity class.

Thus
\[
(a^3)^*
\notin
\mathcal L_{0,\mathrm P}^{\theta}(\mathcal U),
\]
whereas
\[
(a^3)^*
\in
\mathcal L_{F}^{\theta}(\mathcal U)
\subseteq
\mathcal L_{F,\mathrm P}^{\theta}(\mathcal U).
\]
Together with the inclusions already
established in
Proposition~
\ref{prop:global-cpa-language-inclusions},
this proves both strict
automaton-specific relations.
\end{proof}
\subsection{Controlled Probabilistic Automata on Infinite Words}
\label{subsec:probabilistic-omega-languages}

We extend the finite-word semantics of controlled probabilistic automata
to infinite activity words using the B\"uchi condition introduced for controlled automata
on infinite words in Section~3. As in the nondeterministic setting, an activity word is
a fixed external input; no distribution over activity words is assumed.
Policies retain the activity-first observation semantics: a control
action (or a distribution over actions) is selected after observing the
current state and activity. Unlike existential acceptance, probabilistic
acceptance measures the set of genuine infinite B\"uchi-accepting runs.
Since terminated runs are not artificially extended, this measure may
have total mass less than one.

\begin{definition}[Controlled probabilistic B\"uchi automaton]
% Former Definition 4.23
A \emph{controlled probabilistic B\"uchi automaton} (CPBA) is
\[
  \mathcal U_\omega=(Q,\Sigma,C,P,\nu_0,F),
\]
where $(Q,\Sigma,C,P,\nu_0)$ is a controlled probabilistic automaton
as in Definition~\ref{def:controlled-probabilistic-automaton}, with finite state set $Q$ and finite nonempty
control set $C$, and $F\subseteq Q$. For an infinite state sequence
$\mathbf q=q_0q_1\cdots$, write
$\operatorname{Inf}(\mathbf q)=
\{q\in Q:q=q_i\text{ for infinitely many }i\}$.
It is \emph{B\"uchi accepting} if
$\operatorname{Inf}(\mathbf q)\cap F\neq\varnothing$.
\end{definition}

Fix $w=a_0a_1\cdots\in\Sigma^\omega$ and an activity-conditioned
probabilistic policy
$\pi:(Q\times\Sigma)^*\to\operatorname{Dist}(C)$.
For a joint prefix $\xi=q_0c_0q_1\cdots c_{n-1}q_n$, put
\[
\begin{gathered}
 h_i=(q_0,a_0)\cdots(q_i,a_i),\\
 p_{\pi,w}(\xi)=\nu_0(q_0)
 \prod_{i=0}^{n-1}\pi(h_i)(c_i)
 P(q_i,a_i,c_i,q_{i+1}).
\end{gathered}
\]

The empty product gives $p_{\pi,w}(q_0)=\nu_0(q_0)$.
A joint prefix is \emph{feasible} if its weight is positive.
Let $\Omega^{sa}_{\pi,w}$, where $sa$ denotes
\emph{state--action}, be the set of infinite joint paths
$q_0c_0q_1c_1\cdots$ all of whose finite prefixes are
feasible. For a finite joint prefix $\xi$, let $[\xi]_{sa}$
denote the cylinder of genuine infinite state--action
paths extending $\xi$. The finite-prefix weights
determine a subprobability measure $\mathbb P^{sa}_{\pi,w}$
on this space, characterized by
\[
\mathbb P^{sa}_{\pi,w}([\xi]_{sa})
=
\lim_{m\to\infty}
\sum_{\substack{\eta\text{ of length }m\\
                 \eta\text{ extends }\xi}}
p_{\pi,w}(\eta).
\]
Here length counts activity transitions. The sums are nonincreasing
and bounded, since $P$ is a (possibly subprobability) transition
kernel for each state, activity, and control. Equivalently, one may
adjoin a cemetery state for terminated paths, construct the usual
probability measure on the completed path space, and restrict it to
paths never visiting that state. This restriction yields the displayed
cylinder masses and requires no conditioning on survival.

Let $\operatorname{pr}_Q$ denote the projection that
removes control actions from infinite state--action
paths, and let
$\Omega_{\pi,w}=\operatorname{pr}_Q(\Omega^{sa}_{\pi,w})$.
Define the induced subprobability measure on state
paths by
\[
\mathbb P_{\pi,w}
=\mathbb P^{sa}_{\pi,w}\circ
\operatorname{pr}_Q^{-1}.
\]
Let $\Omega_{\pi,w}$ be the state-coordinate projection of
$\Omega^{sa}_{\pi,w}$ and write
\[
 \mathbb P_{\pi,w}
 =\mathbb P^{sa}_{\pi,w}\circ\operatorname{pr}_Q^{-1}.
\]
The mass $\mathbb P_{\pi,w}(\Omega_{\pi,w})$ may be less than one:
a positive-weight finite execution need not continue along $w$.

\begin{definition}[Probabilistic B\"uchi acceptance]
% Former Definition 4.28
For a CPBA $\mathcal U_\omega$, a policy $\pi$, and a fixed
word $w$, let $\operatorname{Buchi}(F)$ be the event of genuine
infinite state paths visiting $F$ infinitely often. Define
\[
\begin{aligned}
 \operatorname{Buchi}(F)
  &=\bigcap_{N\ge0}\bigcup_{n\ge N}
     \{\mathbf q\in\Omega_{\pi,w}:q_n\in F\},\\
 \operatorname{Acc}^{\omega}_{\pi,\mathcal U_\omega}(w)
  &=\mathbb P_{\pi,w}(\operatorname{Buchi}(F)).
\end{aligned}
\]
For $\theta\in(0,1]$, the threshold, positive-acceptance,
and almost-sure-acceptance languages are, respectively,
\[
\begin{aligned}
 L^{\omega,\theta}_\pi(\mathcal U_\omega)
  &=\{w\in\Sigma^\omega:
       \operatorname{Acc}^{\omega}_{\pi,\mathcal U_\omega}(w)
       \ge\theta\},\\
 L^{\omega,>0}_\pi(\mathcal U_\omega)
  &=\{w\in\Sigma^\omega:
       \operatorname{Acc}^{\omega}_{\pi,\mathcal U_\omega}(w)>0\},\\
 L^{\omega,1}_\pi(\mathcal U_\omega)
  &=\{w\in\Sigma^\omega:
       \operatorname{Acc}^{\omega}_{\pi,\mathcal U_\omega}(w)=1\}.
\end{aligned}
\]
\end{definition}

The B\"uchi event is measurable by its countable union--intersection
representation. Unlike nondeterministic acceptance, the existence of an
accepting infinite run does not imply positive acceptance probability. Almost-sure acceptance requires both infinite
continuation with probability one and B\"uchi acceptance with
probability one; no normalization conditional on continuation is used.

\begin{proposition}[Policy-class inclusions]
% Consolidates former Proposition 4.30 and Corollary 4.34
Let $\alpha\in(0,1]\cup\{>0\}$ denote an acceptance criterion.
For $X\in\{0,F,\mathrm{Stack},\mathrm{Tape},H\}$, let
$\Pi_X^D(\mathcal U_\omega)$ and
$\Pi_X^P(\mathcal U_\omega)$ denote the deterministic and
probabilistic policies of type $X$. Define
\[
\begin{aligned}
 \mathcal L_X^{\omega,\alpha}(\mathcal U_\omega)
 &=\{L_\pi^{\omega,\alpha}(\mathcal U_\omega):
       \pi\in\Pi_X^D(\mathcal U_\omega)\},\\
 \mathcal L_{X,P}^{\omega,\alpha}(\mathcal U_\omega)
 &=\{L_\pi^{\omega,\alpha}(\mathcal U_\omega):
       \pi\in\Pi_X^P(\mathcal U_\omega)\}.
\end{aligned}
\]
For each fixed CPBA, both families are monotone along
$0\subseteq F\subseteq\mathrm{Stack}\subseteq\mathrm{Tape}
\subseteq H$, in the sense that
\[
\begin{gathered}
 \mathcal L_0^{\omega,\alpha}
 \subseteq\mathcal L_F^{\omega,\alpha}
 \subseteq\mathcal L_{\mathrm{Stack}}^{\omega,\alpha}
 \subseteq\mathcal L_{\mathrm{Tape}}^{\omega,\alpha}
 \subseteq\mathcal L_H^{\omega,\alpha},\\
 \mathcal L_{0,P}^{\omega,\alpha}
 \subseteq\mathcal L_{F,P}^{\omega,\alpha}
 \subseteq\mathcal L_{\mathrm{Stack},P}^{\omega,\alpha}
 \subseteq\mathcal L_{\mathrm{Tape},P}^{\omega,\alpha}
 \subseteq\mathcal L_{H,P}^{\omega,\alpha},
\end{gathered}
\]
where the common argument $\mathcal U_\omega$ is suppressed.
Moreover,
$\mathcal L_X^{\omega,\alpha}(\mathcal U_\omega)
\subseteq\mathcal L_{X,P}^{\omega,\alpha}(\mathcal U_\omega)$.
\end{proposition}

\begin{proof}
A stronger controller class simulates every policy of a weaker
class. Taking the images of the nested policy sets under
$\pi\mapsto L_\pi^{\omega,\alpha}(\mathcal U_\omega)$ proves the
chains. A deterministic policy is represented probabilistically by
Dirac distributions, preserving all prefix weights and hence
acceptance probabilities.
\end{proof}

\begin{definition}[Acceptance-preserving probabilistic bisimulation]
% Former Definition 4.31
Two CPBAs
$\mathcal U_\omega=(Q,\Sigma,C,P,\nu_0,F)$ and
$\mathcal U'_\omega=(Q',\Sigma,C,P',\nu'_0,F')$
are \emph{acceptance-preserving probabilistically bisimilar}
if their underlying CPAs are probabilistically bisimilar under
an equivalence relation $\sim$ as in Definition~4.3 and
\[
 q\sim q'\quad\Longrightarrow\quad
 (q\in F\iff q'\in F').
\]
\end{definition}

\begin{proposition}[Preservation of probabilistic
$\omega$-languages]
Let $\mathcal U_\omega$ and $\mathcal U'_\omega$ be
acceptance-preserving probabilistically bisimilar
CPBAs under $\sim$. Let $\pi$ and $\pi'$ be
activity-conditioned probabilistic policies
(including deterministic policies as special cases).
Suppose that, for every pair of observation histories
with identical activity sequences and pairwise
$\sim$-related states, the policies assign identical
control distributions.

Then, for every $w\in\Sigma^\omega$,
\[
\operatorname{Acc}^{\omega}_{\pi,\mathcal U_\omega}(w)
=
\operatorname{Acc}^{\omega}_{\pi',\mathcal U'_\omega}(w).
\]
Consequently,
\[
L_\pi^{\omega,\alpha}(\mathcal U_\omega)
=
L_{\pi'}^{\omega,\alpha}(\mathcal U'_\omega)
\]
for every $\alpha\in(0,1]\cup\{>0\}$.
Here the equality applies to both deterministic and
probabilistic policies; the notation
$L_\pi^{\omega,\alpha}$ denotes the language induced
by an individual policy $\pi$.
\end{proposition}

\begin{proof}
Fix $w\in\Sigma^\omega$. By probabilistic
bisimulation, related states have identical
transition probabilities into corresponding
equivalence classes for each activity and control.
The compatibility assumption ensures that
corresponding observation histories induce
identical distributions over controls.

Consequently, the induced probability measures
agree on cylinders determined by finite sequences
of equivalence classes. Passing to genuine infinite
paths preserves this equality, including any loss
of probability mass due to termination.

Since the B\"uchi acceptance condition is invariant
under $\sim$, the B\"uchi events have equal
probabilities. Hence the acceptance probabilities
coincide for every $w$, and the asserted language
equalities follow.
\end{proof}
As in the finite-word setting, the corresponding
collections of automaton-specific language families
and global language classes are defined by
\[
\mathfrak C_X^{\omega,\alpha}
=
\left\{
\mathcal L_X^{\omega,\alpha}(\mathcal U_\omega)
:
\mathcal U_\omega\in
\mathfrak U_{\mathrm{fin}}^\omega(\Sigma)
\right\},
\qquad
\mathcal C_X^{\omega,\alpha}
=
\bigcup_{\mathcal U_\omega\in
\mathfrak U_{\mathrm{fin}}^\omega(\Sigma)}
\mathcal L_X^{\omega,\alpha}(\mathcal U_\omega).
\]
The corresponding probabilistic-policy collections
$\mathfrak C_{X,P}^{\omega,\alpha}$ and global classes
$\mathbf C_{X,P}^{\omega,\alpha}$ are defined analogously,
whereas the subscript $P$
in $\mathcal L_{X,P}^{\omega,\alpha}$ identifies
families generated by probabilistic policies of type X.

As in the finite-word and nondeterministic
$\omega$-word settings, the relationships between
policy classes established above induce
corresponding pointwise preorder relations
$\preccurlyeq$ between automaton-indexed
collections of language families and ordinary
inclusions between their global language classes.
These inclusions follow from containment of the
corresponding policy classes and do not, by
themselves, establish strictness, which would
require separating languages.
Likewise, no exact characterization of the
stack- or tape-augmented probabilistic
$\omega$-language classes is claimed, since
such a characterization would require matching
upper and lower bounds under the adopted
probabilistic acceptance semantics.
\begin{remark}[Activity-conditioned versus contingent-action policies]
\label{rem:cpa-policy-architectures}
The fixed-word acceptance semantics above uses
\emph{activity-conditioned} policies: the activity is observed
before its control is selected. In the activity-generation and
DTMDP formulations below, a \emph{contingent-action} policy
instead assigns controls to enabled activities before the next
activity is generated. Under the randomized contingent-action
convention, these provisional assignments are conditionally
independent given the controller's information, inducing a
product-form distribution over contingent-action maps. The
activity-generation law may nevertheless depend on the complete
realized map. The two policy architectures must therefore not
be identified without an explicit correspondence argument.
\end{remark}
\subsection{Reward Models and PAC Optimality}
\label{subsec:cpa-reward-models}

We now equip controlled probabilistic automata with rewards for
evaluating and comparing policies over discrete activity steps. The
reward model developed here concerns the stochastic experiment in
which activities are generated according to an explicit
activity-generation distribution. This is distinct from the
fixed-word semantics of the preceding subsection, in which an
activity word is supplied externally. The present formulation is the
one used below for optimization and for the connection with
discrete-time Markov decision processes (DTMDPs).

The reward model does not assign physical holding times to states or
integrate reward rates over continuous time. Those quantities belong
to the stochastic reward model introduced for controlled stochastic
automata in Section~5.
\begin{definition}[Reward structure for a controlled probabilistic automaton]
\label{def:cpa-reward-structure}
Let
\[
    \mathcal{U}=(Q,\Sigma,C,P,\nu_0)
\]
be a controlled probabilistic automaton. A \emph{discrete-step reward
structure} for $\mathcal{U}$ is a bounded function
\[
    r:Q\times\Sigma\times C\times Q\longrightarrow\mathbb{R},
\]
where $r(q,a,c,q')$ specifies the reward associated with a realized
transition from state $q$ to state $q'$ under activity $a$ and control
action $c$. The reward function is evaluated only on admissible
transitions, that is, on tuples $(q,a,c,q')$ for which
\[
    P(q,a,c,q')>0.
\]
\end{definition}
Only values on transitions of positive probability affect execution
rewards. If the intended model permits only nonnegative rewards, the
codomains may be restricted to $\mathbb R_{\geq0}$.

For $q\in Q$ and $a\in\Sigma$, write
\[
    C(q,a)
    =
    \left\{
        c\in C:
        \exists q'\in Q,\;
        P(q,a,c,q')>0
    \right\},
\]
and
\[
    \Sigma_{\mathrm{en}}(q)
    =
    \left\{
        a\in\Sigma:
        C(q,a)\neq\varnothing
    \right\}.
\]
Put
\[
    Q_{\mathrm{live}}
    =
    \left\{
        q\in Q:
        \Sigma_{\mathrm{en}}(q)\neq\varnothing
    \right\},
    \qquad
    Q_{\mathrm{stop}}
    =
    Q\setminus Q_{\mathrm{live}}.
\]
For each $q\in Q_{\mathrm{live}}$, define the set of admissible
\emph{contingent actions}
\begin{equation}
\label{eq:cpa-contingent-actions}
    B(q)=
    \left\{
       b:\Sigma_{\mathrm{en}}(q)\to C
       \;\middle|\;
       b(a)\in C(q,a)
       \text{ for every }a\in\Sigma_{\mathrm{en}}(q)
    \right\},
\end{equation}
and put
\[
    B=\bigcup_{q\in Q_{\mathrm{live}}}B(q).
\]
Thus $b\in B(q)$ provisionally assigns an admissible control $b(a)$
to every enabled activity $a$ at $q$.

For any finite set $X$, let $\operatorname{Dist}(X)$ denote the
set of probability distributions on $X$.  Recall that the
activity-conditioned probabilistic policies of
Definition~\ref{def:probabilistic-policy} select a control distribution after the realized
activity has been observed.  Thus, at each decision step, policy randomization concerns only
the control associated with the realized activity; the
activity-conditioned semantics does not require a joint
distribution over controls for the other competing enabled
activities.

The contingent-action semantics moves this control decision
before the realized activity is known.  Accordingly, at each
decision epoch and for each enabled activity
$a\in\Sigma_{\rm en}(q)$, the controller provisionally specifies
an activity-specific control distribution
\[
    \kappa_a(\cdot\mid h)\in\operatorname{Dist}(C(q,a)),
\]
where $h$ denotes the information history available to the
controller.  To retain the activity-wise character of the
probabilistic-policy semantics of Definition~\ref{def:probabilistic-policy}, we assume that
these provisional activity-specific control assignments are
conditionally independent given $h$.  Hence, if $B$ denotes the
random contingent-action map, then
\begin{equation}
    \Pr(B=b\mid h)
    =
    \prod_{a\in\Sigma_{\rm en}(q)}
    \kappa_a\bigl(b(a)\mid h\bigr),
    \qquad b\in B(q).
    \label{4.3}
\end{equation}

Thus, $\kappa_a(\cdot\mid h)$ specifies the provisional control
distribution associated with activity $a$, while ~(\ref{4.3}) gives the
product coupling of these activity-specific distributions over
complete contingent-action maps.  
The conditional-independence assumption applies only to the
controller's provisional activity-specific control assignments; it
imposes no independence assumption on the stochastic mechanism that
subsequently determines which activity is generated or completed next.  A
deterministic contingent-action policy is obtained when every
$\kappa_a(\cdot\mid h)$ is a Dirac distribution.

Equation~(\ref{4.3}) specifies how the controller selects a complete
contingent-action map, but it does not determine which enabled
activity is subsequently realized.  A separate probabilistic
mechanism is therefore required to select the activity after the
contingent action has been chosen.  Since the provisional
assignments in $b$ may influence the likelihood of the competing
activities, we allow this activity-generation distribution to
depend on both the current state $q$ and the complete contingent
action $b$.  When activity generation is independent of the
provisional control assignments, this reduces to the special case
$\beta(q,a,b)=\beta(q,a)$.  This leads to the following
definition.
\begin{definition}[Activity-generation distribution]
\label{def:cpa-activity-generation}
An \emph{activity-generation distribution} for $\mathcal U$ is a
function
\[
    \beta:Q_{\mathrm{live}}\times\Sigma\times B
       \longrightarrow[0,1]
\]
such that, for every $q\in Q_{\mathrm{live}}$ and $b\in B(q)$,
\[
    \beta(q,a,b)=0
    \qquad
    \text{for }a\notin\Sigma_{\mathrm{en}}(q),
\]
and
\[
    \sum_{a\in\Sigma_{\mathrm{en}}(q)}\beta(q,a,b)=1.
\]
The value $\beta(q,a,b)$ is the conditional probability that activity
$a$ is generated when the current state is $q$ and the contingent
action $b$ has been selected. No activity-generation distribution is
required at stopped states. 
\end{definition}
For decision epoch $n$, let the random variables
$Q_n$, $B_n$, $A_n$, $C_n$, and $Q_{n+1}$ denote, respectively,
the current state, the contingent-action map selected by the controller,
the subsequently generated activity, the control actually executed,
and the resulting successor state, with
\[
    C_n = B_n(A_n).
\]
With this notation, the stochastic experiment induced by a
contingent-action policy and $\beta$ realizes the architecture
described in Remark~\ref{rem:cpa-policy-architectures} and follows the decision order
\begin{equation}
Q_n \longrightarrow B_n \longrightarrow A_n
\longrightarrow C_n=B_n(A_n) \longrightarrow Q_{n+1}.
\label{eq:cpa-beta-decision-order}
\end{equation}
Given the available history $H_n$, whose current state is
$Q_n=q$, the activity-specific policy kernels induce, by the
conditional-independence convention above, the product distribution
\[
    \Pr(B_n=b\mid H_n)
    =
    \prod_{a\in\Sigma_{\rm en}(q)}
    \kappa_{n,a}\bigl(b(a)\mid H_n\bigr),
    \qquad b\in B(q).
\]
Conditional on $Q_n=q$ and $B_n=b$, the activity $A_n$ is
generated according to $\beta(q,\cdot,b)$.  If $A_n=a$, the
control actually executed is then
\[
    C_n=B_n(A_n)=b(a),
\]
and the successor state is sampled according to
\[
    Q_{n+1}\sim P(q,a,b(a),\cdot).
\]
The conditional-independence assumption concerns only the
controller's provisional activity-specific control assignments
given its available information.  It does not require the
activity-generation law to be independent of those assignments:
$\beta(q,a,b)$ may depend on the complete contingent action $b$.
Thus, although the provisional assignments are conditionally
independent given the controller's information, their realized
values may jointly influence which activity is subsequently
generated.

For a realization of the $\beta$-experiment, define the one-step reward
at decision epoch $i$ by
\begin{equation}
\label{eq:cpa-beta-one-step-reward}
    G_i
    =r(Q_i,A_i,B_i(A_i),Q_{i+1}).
\end{equation}
At a stopped state, no CPA transition occurs and no further CPA reward
is accumulated.

Having specified the stochastic experiment induced by a
contingent-action policy and the corresponding one-step reward,
we can now define cumulative reward criteria and compare
admissible policies with respect to their expected performance.
We consider finite-horizon and discounted expected rewards.
\begin{definition}[Expected finite-horizon and discounted reward]
\label{def:cpa-beta-reward-objectives}
Let $\Pi$ be a specified nonempty class of admissible
contingent-action policies satisfying the product-form
conditional-independence convention in~(\ref{4.3}), let $q\in Q$,
and let $n\in\mathbb{N}$.

For $\pi\in\Pi$, define the finite-horizon
accumulated reward
\[
    R_n^{\pi,\beta}
    =\sum_{i=0}^{n-1}G_i,
\]
where rewards after the process has reached a stopped state are taken
to be zero. Its expected value is
\[
    J_n(\pi;q,\beta)
    =\mathbb E_q^{\pi,\beta}
       \left[R_n^{\pi,\beta}\right].
\]
The optimal finite-horizon value over $\Pi$ is
\[
    V_n^\Pi(q,\beta)
    =\sup_{\pi\in\Pi}J_n(\pi;q,\beta).
\]
A policy $\pi^*\in\Pi$ is finite-horizon expected-reward optimal if
\[
    J_n(\pi^*;q,\beta)=V_n^\Pi(q,\beta),
\]
and, for $\varepsilon>0$, is $\varepsilon$-optimal if
\[
    J_n(\pi^*;q,\beta)
    \geq V_n^\Pi(q,\beta)-\varepsilon.
\]

For a discount factor $\gamma\in(0,1)$, define
\[
    R_\gamma^{\pi,\beta}
    =\sum_{i=0}^{\infty}\gamma^iG_i,
    \qquad
    J_\gamma(\pi;q,\beta)
    =\mathbb E_q^{\pi,\beta}
       \left[R_\gamma^{\pi,\beta}\right],
\]
and
\[
    V_\gamma^\Pi(q,\beta)
    =\sup_{\pi\in\Pi}J_\gamma(\pi;q,\beta).
\]
Since $g$ is bounded and $\gamma<1$, the discounted reward is bounded
and integrable. Undiscounted infinite-horizon reward may also be considered under
suitable conditions ensuring that the corresponding expectation is
well-defined.
\end{definition}
The preceding reward criteria are defined directly on realizations
of the $\beta$-experiment.  It is also useful to associate with
each state--contingent-action pair $(q,b)$ its conditional expected
one-step reward.  Besides summarizing the reward of one decision
epoch, this quantity will provide the link to the rewarded DTMDP
representation developed below.
\begin{proposition}[One-step expected reward under activity generation]
\label{prop:cpa-beta-one-step-expected-reward}
For $q\in Q_{\mathrm{live}}$ and $b\in B(q)$, the conditional expected
one-step reward of the CPA experiment induced by $\beta$ is
\begin{equation}
\label{eq:cpa-beta-expected-reward}
\begin{aligned}
    \bar r_\beta(q,b)
    &:=\mathbb E\!\left[
       r(q,A,b(A),Q')\mid Q=q,B=b
       \right]\\
    &=\sum_{a\in\Sigma_{\mathrm{en}}(q)}
       \beta(q,a,b)
       \sum_{q'\in Q}
       P(q,a,b(a),q')\,
       r(q,a,b(a),q').
\end{aligned}
\end{equation}
For $q\in Q_{\mathrm{stop}}$, set
\[
    \bar r_\beta(q,\mathsf{stop})=0.
\]
\end{proposition}

\begin{proof}
Conditional on $(Q,B)=(q,b)$, activity $a$ is generated with
probability $\beta(q,a,b)$. Conditional further on $A=a$, the
executed control is $b(a)$ and the successor state $q'$ is sampled
according to $P(q,a,b(a),q')$. Summing the realized reward
$r(q,a,b(a),q')$ over these two conditional distributions gives
\eqref{eq:cpa-beta-expected-reward}.
\end{proof}

\begin{remark}[Fixed-word reward semantics]
\label{rem:cpa-fixed-word-reward}
The reward objectives above concern the stochastic experiment induced
by the contingent-action policy and $\beta$. They do not replace the
fixed-word semantics used for language and acceptance questions. If an
activity word $w=a_0a_1\ldots$ is externally prescribed, rewards may
instead be accumulated along the corresponding fixed-word executions.
Such conditional objectives are meaningful when performance relative
to a specified input sequence is itself of interest, but they are
distinct from the unconditional reward-optimization problem induced
by $\beta$ and from the DTMDP representation below.
\end{remark}

\begin{remark}[Reward optimization and PAC optimality]
\label{rem:cpa-reward-pac}
Expected-reward optimality is always relative to the specified policy
class $\Pi$, reward objective, activity-generation distribution
$\beta$, and information available to the controller. For finite state
and action spaces of manageable size, exact policy synthesis may be
possible using dynamic-programming methods. For larger systems,
approximate dynamic programming, simulation, and reinforcement-learning
methods may be used to search for policies within a chosen
approximation or learning class.

A PAC-style statement requires additional statistical assumptions and
a specified learning procedure. For example, for an objective value
$J$ and policy class $\Pi$, a learned policy $\widehat\pi$ may be called
$(\varepsilon,\delta)$-PAC optimal when, with probability at least
$1-\delta$ over the data or sampling process used by the learning
algorithm,
\[
    J(\widehat\pi)
    \geq
    \sup_{\pi\in\Pi}J(\pi)-\varepsilon.
\]
This is a performance criterion rather than an automatic consequence
of the CPA semantics. Establishing such a guarantee requires explicit
assumptions on sampling, observability, model knowledge, approximation,
and the learning algorithm. The activity-generation law $\beta$ is
part of the stochastic experiment with respect to which the objective
and any statistical guarantee are defined.
\end{remark}

The preceding formulation makes the activity-generation mechanism,
the contingent-action policy convention, and the reward structure
explicit before introducing DTMDPs. We now show that the same
stochastic experiment admits a standard DTMDP representation and that
the expected one-step reward defined in
\eqref{eq:cpa-beta-expected-reward}
coincides with the corresponding DTMDP reward.
\subsection{Discrete-Time Markov Decision Process (DTMDP)}
\label{subsec:cpa-dtmdp}

A controlled probabilistic automaton specifies the successor-state
distribution conditional on an activity and a chosen control, whereas
the activity-generation distribution of
Definition~\ref{def:cpa-activity-generation} specifies how the next
activity is generated after a contingent action has been selected.
We now marginalize over that generated activity to obtain a DTMDP
representation of the same stochastic experiment. The representation
does not describe the CPA's fixed-word semantics.

\begin{definition}[Discrete-time Markov decision process]
\label{def:dtmdp}
A \emph{discrete-time Markov decision process} (DTMDP) is a tuple
\[
    \mathcal D=(S,B,\mu_0,P),
\]
where:
\begin{itemize}
    \item $S$ is a finite or countable state space;
    \item $B$ is a set of actions;
    \item $\mu_0\in\operatorname{Dist}(S)$ is the initial-state distribution;
    \item
    \[
        P:S\times B\rightharpoonup\operatorname{Dist}(S)
    \]
    is a partial transition-probability kernel.
\end{itemize}
Here $\rightharpoonup$ denotes a partial function. For each $s\in S$,
define
\[
    B(s)=\{b\in B:P(s,b,\cdot)\text{ is defined}\},
\]
and require $B(s)\neq\varnothing$. Equivalently,
\[
    \operatorname{dom}(P)
    =\{(s,b)\in S\times B:b\in B(s)\}.
\]
If $b\in B(s)$ is selected, then
\[
    P(s,b,s')
    =\Pr(S_{n+1}=s'\mid S_n=s,B_n=b),
\]
with $P(s,b,s')\geq0$ and
\[
    \sum_{s'\in S}P(s,b,s')=1.
\]
A stationary deterministic policy selects $\varphi(s)\in B(s)$,
whereas a stationary randomized policy assigns
\[
    \varphi(\cdot\mid s)\in\operatorname{Dist}(B(s)).
\]
\end{definition}
A reward structure for a DTMDP
$\mathcal{D}=(S,B,\mu_0,P)$ is a bounded function
\[
    r^{\mathcal D}:S\times B\times S\longrightarrow\mathbb{R},
\]
where $r^{\mathcal D}(s,b,s')$ is the reward associated with a
transition from state $s$ to state $s'$ under action $b$.

Recall $Q_{\mathrm{live}}$, $Q_{\mathrm{stop}}$, $B(q)$, $B$, $\Sigma_{\mathrm{en}}(q)$, and
$\beta$ from Subsection~\ref{subsec:cpa-reward-models}.
\begin{definition}[Contingent-action DTMDP representation of a CPA]
\label{def:cpa-contingent-dtmdp}
Let $\mathcal U=(Q,\Sigma,C,P,\nu_0)$ be a CPA equipped with an
activity-generation distribution $\beta$. Introduce a dummy action
$\mathsf{stop}$ and put
\[
    B_{\mathrm{stop}}=\{\mathsf{stop}\},
    \qquad
    B_\beta=B\cup B_{\mathrm{stop}}.
\]
For each $q\in Q$, define
\[
    B_\beta(q)=
    \begin{cases}
       B(q), & q\in Q_{\mathrm{live}},\\[1mm]
       \{\mathsf{stop}\}, & q\in Q_{\mathrm{stop}}.
    \end{cases}
\]
The \emph{contingent-action DTMDP representation} of
$(\mathcal U,\beta)$ is
\[
    \mathcal D_\beta=(Q,B_\beta,\nu_0,P_\beta),
\]
where
\[
P_\beta:\operatorname{dom}(P_\beta)
\longrightarrow \operatorname{Dist}(Q),
\]
with
\[
\operatorname{dom}(P_\beta)
=\{(q,b)\in Q\times B_\beta:b\in B_\beta(q)\}.
\]
For $q\in Q_{\mathrm{live}}$, $b\in B(q)$, and $q'\in Q$, define
\begin{equation}
\label{eq:dtmdp-contingent-kernel}
    P_\beta(q,b,q')
    =\sum_{a\in\Sigma_{\mathrm{en}}(q)}
       \beta(q,a,b)P(q,a,b(a),q').
\end{equation}
For $q\in Q_{\mathrm{stop}}$, define
\[
    P_\beta(q,\mathsf{stop},\cdot)=\delta_q.
\]
The dummy self-loop is a bookkeeping device only and does not represent
a CPA transition. Since $\Sigma$ and $C$ are finite, every
$B_\beta(q)$ is finite and nonempty.
\end{definition}

The construction is well defined. Indeed, $B_\beta(q)$ is nonempty for
every $q\in Q$. Moreover, for $q\in Q_{\mathrm{live}}$ and $b\in B(q)$,
\[
\begin{aligned}
\sum_{q'\in Q} P_\beta(q,b,q')
&=
\sum_{a\in\Sigma_{\mathrm{en}}(q)}
\beta(q,a,b)
\sum_{q'\in Q} P(q,a,b(a),q')\\
&=
\sum_{a\in\Sigma_{\mathrm{en}}(q)}
\beta(q,a,b)
=1.
\end{aligned}
\]
Thus $P_\beta(q,b,\cdot)$ is a probability distribution. At a stopped
state, $P_\beta(q,\mathsf{stop},\cdot)=\delta_q$ is also a probability
distribution.

\begin{remark}[Decision order]
\label{rem:cpa-contingent-decision-order}
The contingent-action representation uses exactly the decision order
in \eqref{eq:cpa-beta-decision-order}. Equation
\eqref{eq:dtmdp-contingent-kernel} marginalizes over the generated
activity and therefore gives the state-transition law of the DTMDP
representation. It does not describe the fixed-word experiment in
which the next activity is externally prescribed.
\end{remark}

A stationary deterministic contingent-action policy selects
\[
    \varphi(q)\in B(q)
    \qquad(q\in Q_{\mathrm{live}}),
\]
and corresponds directly to the stationary deterministic DTMDP policy
selecting the same action.

For randomized policies, recall that the CPA contingent-action
convention assumes conditional independence of the provisional
activity-specific control assignments. Thus a stationary randomized
CPA policy is specified by distributions
\[
    \kappa_a(\cdot\mid q)
    \in\operatorname{Dist}(C(q,a)),
    \qquad
    a\in\Sigma_{\mathrm{en}}(q),
\]
and induces the product-form DTMDP policy
\begin{equation}
\label{eq:dtmdp-product-policy}
    \varphi(b\mid q)
    =
    \prod_{a\in\Sigma_{\mathrm{en}}(q)}
       \kappa_a\bigl(b(a)\mid q\bigr),
    \qquad b\in B(q).
\end{equation}
More generally, the same correspondence applies to history-dependent
policies by conditioning the activity-specific kernels on the
available history.

An unrestricted randomized DTMDP policy may assign an arbitrary
distribution in $\operatorname{Dist}(B(q))$ and is therefore, in
general, more general than the randomized contingent-action CPA policy
class adopted here. Direct randomized-policy correspondence holds for
the product-form subclass defined by
\eqref{eq:dtmdp-product-policy}.

\begin{proposition}[Relationship between DTMDPs and CPAs]
\label{prop:dtmdp-cpa-relationship}
Every DTMDP admits a single-activity CPA representation that, under
corresponding policies, activity-generation distribution, and matching
initial conditions, preserves the state-process law.

Conversely, the contingent-action DTMDP representation of a CPA
equipped with an activity-generation distribution $\beta$ preserves
the state-process law induced by $\beta$ under corresponding
deterministic contingent-action policies and, for randomized policies,
under the product-form correspondence described above. The
representation marginalizes over the generated activity and therefore
does not, in general, preserve activity-labelled histories.
\end{proposition}

\begin{proof}
For the first direction, let
\[
    \mathcal D=(S,B,\mu_0,P_D)
\]
be a DTMDP. Introduce a single dummy activity $a_0$ and define a CPA
with
\[
    Q=S,\qquad \Sigma=\{a_0\},\qquad C=B,\qquad \nu_0=\mu_0,
\]
with
\[
    C(s,a_0)=B(s),
    \qquad
    P(s,a_0,b,s')=P_D(s,b,s')
\]
for every admissible $(s,b)$, and set
\[
    \beta(s,a_0,b)=1.
\]
Thus the contingent-action set at $s$ is canonically identified with
$B(s)$ and the resulting DTMDP representation has kernel $P_D$.
Because only one activity is enabled, every randomized distribution
over the corresponding contingent-action set is trivially of product
form.

Conversely, Equation~\eqref{eq:dtmdp-contingent-kernel} is exactly the
conditional successor-state distribution of the CPA experiment induced
by $\beta$ after selecting $b$. Deterministic corresponding policies
therefore induce the same state-process law. The same conclusion holds
for randomized policies whenever the DTMDP distribution over $B(q)$ is
the product distribution induced by the CPA's activity-specific policy
kernels.
\end{proof}
\begin{proposition}[Reward correspondence]
\label{prop:dtmdp-reward-correspondence}
Let
\[
    \mathcal A=(Q,\Sigma,C,P,\nu_0)
\]
be a CPA equipped with the bounded reward structure $r$ of
Definition~\ref{def:cpa-reward-structure}, and let $\beta$ be an
activity-generation distribution. Let
\[
    \mathcal D_\beta=(Q,B,\nu_0,P_\beta)
\]
be the contingent-action DTMDP induced by $(\mathcal A,\beta)$,
where the admissible action set at $q$ is $B(q)$ and, for
$q\in Q_{\mathrm{live}}$, $b\in B(q)$, and $q'\in Q$,
\begin{equation}
\label{eq:dtmdp-contingent-kernel-reward}
    P_\beta(q,b,q')
    =
    \sum_{a\in\Sigma_{\mathrm{en}}(q)}
    \beta(q,a,b)\,
    P(q,a,b(a),q').
\end{equation}
At a stopped state, the DTMDP has the prescribed dummy self-loop.

Equip $\mathcal D_\beta$ with the induced bounded
transition-reward function $r_\beta^{\mathcal D}$, defined, whenever
$P_\beta(q,b,q')>0$, by
\begin{equation}
\label{eq:dtmdp-induced-transition-reward}
    r_\beta^{\mathcal D}(q,b,q')
    =
    \frac{
        \displaystyle
        \sum_{a\in\Sigma_{\mathrm{en}}(q)}
        \beta(q,a,b)\,
        P(q,a,b(a),q')\,
        r(q,a,b(a),q')
    }{
        \displaystyle
        P_\beta(q,b,q')
    }.
\end{equation}
When $P_\beta(q,b,q')=0$, its value is immaterial and may be set
to zero. At stopped states, the dummy self-loop is assigned reward
zero.

Then, for every $q\in Q_{\mathrm{live}}$ and $b\in B(q)$,
\begin{equation}
\label{eq:dtmdp-cpa-reward-correspondence}
    \sum_{q'\in Q}
    P_\beta(q,b,q')\,
    r_\beta^{\mathcal D}(q,b,q')
    =
    \bar r_\beta(q,b),
\end{equation}
where
\[
    \bar r_\beta(q,b)
    =
    \sum_{a\in\Sigma_{\mathrm{en}}(q)}
    \beta(q,a,b)
    \sum_{q'\in Q}
    P(q,a,b(a),q')\,
    r(q,a,b(a),q')
\]
is the conditional expected one-step CPA reward.

Under matching initial conditions and corresponding deterministic
policies, the CPA experiment induced by $\beta$ and
$\mathcal D_\beta$ have identical expected finite-horizon and
discounted returns. The same conclusion holds for corresponding
randomized policies whenever the distribution over $B(q)$ in
$\mathcal D_\beta$ is the product distribution induced by the
CPA's activity-specific policy kernels.
\end{proposition}

\begin{proof}
Fix $q\in Q_{\mathrm{live}}$ and $b\in B(q)$. Under the CPA
experiment induced by $\beta$, the conditional probability that
activity $a$ is generated is
\[
    \beta(q,a,b).
\]
Conditional on the generated activity $a$, the executed control is
$b(a)$ and the successor state has distribution
\[
    P(q,a,b(a),\cdot).
\]
Hence the joint conditional probability that activity $a$ is
generated and the successor state is $q'$ is
\[
    \beta(q,a,b)\,
    P(q,a,b(a),q').
\]
Marginalizing over the generated activity gives
\[
    P_\beta(q,b,q')
    =
    \sum_{a\in\Sigma_{\mathrm{en}}(q)}
    \beta(q,a,b)\,
    P(q,a,b(a),q'),
\]
which is precisely the transition kernel of $\mathcal D_\beta$.

For every $q'$ such that $P_\beta(q,b,q')>0$,
\eqref{eq:dtmdp-induced-transition-reward} is the conditional
expected CPA reward given the current state $q$, the selected
contingent action $b$, and the successor state $q'$. Therefore,
\[
\begin{split}
&\sum_{q'\in Q}
P_\beta(q,b,q')\,
r_\beta^{\mathcal D}(q,b,q')
\\
&\quad =
\sum_{q'\in Q}
\sum_{a\in\Sigma_{\mathrm{en}}(q)}
\beta(q,a,b)\,
P(q,a,b(a),q')\,
r(q,a,b(a),q')
\\
&\quad =
\sum_{a\in\Sigma_{\mathrm{en}}(q)}
\beta(q,a,b)
\sum_{q'\in Q}
P(q,a,b(a),q')\,
r(q,a,b(a),q')
\\
&\quad =
\bar r_\beta(q,b).
\end{split}
\]
Thus the CPA experiment and $\mathcal D_\beta$ have the same
conditional expected one-step reward given the current state and
selected contingent action.

Under corresponding deterministic policies, the two models select
corresponding contingent actions and, by the definition of
$P_\beta$, have the same conditional successor-state law.
Consequently, they induce the same state--contingent-action
distributions at successive decision epochs. Together with
\eqref{eq:dtmdp-cpa-reward-correspondence}, this gives equality of
the expected one-step rewards at every decision epoch and hence,
by summation, equality of expected finite-horizon returns.

The same argument applies to corresponding randomized policies
whenever the DTMDP distribution over $B(q)$ is the product
distribution induced by the CPA's activity-specific policy kernels.
Finally, boundedness of the reward function and a discount factor
$\gamma<1$ justify interchange of expectation and the absolutely
convergent discounted sum, yielding equality of expected discounted
returns.
\end{proof}
\begin{remark}[Scope]
\label{rem:dtmdp-representation-scope}
The DTMDP representation depends on the additional activity-generation
law $\beta$ and describes the stochastic experiment in which a
contingent action is selected before the next activity is generated.
It does not describe the fixed-word experiment. Since the
representation marginalizes over the generated activity,
activity-labelled histories and objectives depending on such histories
require separate treatment.

The DTMDP itself permits arbitrary randomized distributions over
$B(q)$, whereas the randomized contingent-action CPA convention adopted
here corresponds to the product-form subclass induced by conditionally
independent activity-specific assignments. Thus results involving
randomized DTMDP policies transfer directly to the CPA only when this
policy-class restriction is respected. No such qualification is needed
for deterministic policies.

Since $\Sigma$ and $C$ are finite, each $B(q)$ is finite. Hence, when
$Q$ is finite, standard finite-state DTMDP results may be applied
provided that the relevant policy class, reward objective, and other
hypotheses are matched.
\end{remark}

\begin{remark}[Optimization and learning from the DTMDP representation]
\label{rem:dtmdp-optimal-control}
Once $\beta$ is fixed, $\mathcal D_\beta$ is an ordinary DTMDP.
Consequently, established methods for DTMDP optimization and policy
synthesis may be applied whenever their hypotheses and the relevant
policy and reward classes are matched. For finite state and action
spaces, infinite-horizon expected discounted-reward objectives admit
stationary deterministic optimal policies, while finite-horizon reward
and reachability objectives admit deterministic Markov policies that
may depend on the remaining horizon
\cite{puterman1994markov,Howard71}. Standard computational methods
include value iteration, policy iteration, and related
dynamic-programming techniques.

Because these standard optimal-policy results admit deterministic
optimal policies, the product-form restriction on randomized
contingent-action CPA policies does not affect the corresponding
optimal value for these objectives under the stated hypotheses.
For objectives or policy restrictions for which randomization is
essential, however, an unrestricted randomized DTMDP policy need not
correspond to an admissible randomized CPA policy; the product-form
restriction must then be imposed explicitly.

As the state space becomes too large for exact dynamic programming,
the same representation provides a basis for approximate dynamic
programming and reinforcement learning. Such methods may use value- or
action-value-function approximation, sampling and simulation, or direct
policy approximation. When the transition model induced by the CPA
and $\beta$ is available, it can also support model-based reinforcement
learning, planning, and simulation. For larger structured systems, the
underlying CPA or activity-network structure may provide information
for state aggregation, feature construction, guided exploration, and
structure-aware policy synthesis.

These observations suggest a scalability progression from exact dynamic
programming to approximate and learning-based methods as explicit
state-space analysis becomes infeasible. They do not imply automatic
optimality, convergence, or PAC guarantees; such guarantees depend on
the objective, approximation architecture, learning algorithm,
information model, sampling assumptions, and other regularity
conditions. The PAC criterion in
Remark~\ref{rem:cpa-reward-pac} therefore becomes applicable only after
these ingredients are specified.

The representation describes the stochastic experiment induced by
$\beta$, rather than the fixed-word semantics. Accordingly, transfer
of optimization or learning results back to the CPA requires
preservation of the relevant objective and information structure and,
for randomized policies, compatibility with the product-form
contingent-action convention. For countable state spaces or
undiscounted objectives, additional boundedness, transience,
integrability, or related regularity conditions may be required
\cite{puterman1994markov,feinberg2002handbook,HernandezLerma1996}.
Computational-complexity conclusions likewise depend on the precise
objective and its encoding \cite{Mundhenk2000}.
\end{remark}
\section{Stochastic Models}
\label{sec:stochastic-models}

Section~4 assigns probabilities to the outcomes of observable activity
completions but does not specify their durations. We now equip timed
activities with stochastic activity-time requirements and
control-dependent processing rates. Concurrently enabled timed
activities may race to complete.

At a stable marking, the policy provisionally assigns an admissible
control action to each enabled timed activity, using the history that
would result from observing that activity's completion. These
contingent assignments determine the activities' processing rates
during the ensuing race. When an activity completes, its provisionally
assigned action is executed, and the resulting transition is resolved
as specified in Section~4. Thus, the action executed is associated
with the activity that actually completes, although its provisional
assignment has already influenced that activity's processing rate.

Instantaneous activities retain their priority and probabilistic
resolution as specified in Section~4.

The underlying controlled probabilistic automaton (CPA) continues to
describe transitions conditional on an observed activity and an
admissible control action. A timed stochastic model additionally
requires a specification of the joint evolution of activity-time
requirements. In particular, marginal activity-time distributions
alone need not determine the law of the next completion.

\subsection{Model Structure}

\begin{definition}[Controlled stochastic activity network]
A \emph{controlled stochastic activity network} (CSAN) is a tuple
\[
 \mathcal N
 =(P,IA,TA,CA,\Gamma,IG,OG,IR,IOR,TOR,IP,F,\rho,\Pi),
\]
where
\[
 (P,IA,TA,CA,\Gamma,IG,OG,IR,IOR,TOR,IP)
\]
is a controlled probabilistic activity network as in
Definition~4.1 and, writing $n=|P|$,
\begin{itemize}
\item
$F=\{F(\cdot\mid a):a\in TA\}$
is a family of probability distribution functions on
$[0,\infty)$ specifying the \emph{activity-time requirements}
of timed activities;

\item
$\rho:\mathbb N^n\times TA\times CA
       \to\mathbb R_{\geq0}$
is the \emph{control-dependent processing-rate function};

\item
$\Pi:\mathbb N^n\times TA
       \to\{\mathrm{true},\mathrm{false}\}$
is the \emph{reactivation predicate}.
\end{itemize}
These timing parameters are relevant only to enabled timed
activities in stable markings. The processing rate
$\rho(\mu,a,c)$ is evaluated for an enabled timed activity $a$
at marking $\mu$ and its provisionally assigned admissible
control action $c$. An enabled activity with processing rate
zero accumulates no processing while that rate remains zero.
\end{definition}

The CSAN has thirteen components: the ten components of its
underlying controlled probabilistic activity network and the three
timing components $F$, $\rho$, and $\Pi$. The activity-time
requirement distribution $F(\cdot\mid a)$ depends only on the
activity. The processing rate $\rho(\mu,a,c)$ depends on the
marking, activity, and provisionally assigned control action,
whereas the reactivation predicate $\Pi(\mu,a)$ depends on the
marking and activity. Consequently, control can influence
stochastic timing both directly, through processing rates, and
indirectly, through changes to the marking and enabled activities.

The tuple specifies the activity-time distributions, processing rates,
and reactivation rules.
A fully specified CSAN additionally fixes the following operational
conventions:
\begin{enumerate}
\item activity-time requirements for newly activated timed activities are 
sampled independently according to their activity-specific distributions specified by \(F\);

\item a rule for resolving simultaneous timed completions when they can
occur with positive probability;

\item the convention for assigning and maintaining provisional control
actions during a race, including the information available to the policy
when those assignments are made.
\end{enumerate}
These conventions are part of the stochastic timing specification,
not additional components of the CSAN tuple.
The stochastic requirements of newly activated timed activities are sampled independently. Their subsequent completion 
dynamics may nevertheless be coupled through the control-dependent processing-rate function.
The probabilistic selection of instantaneous activities
is governed separately by the underlying controlled probabilistic
activity network.

\subsection{Model Behavior}

The marking changes at activity completions according to the
input- and output-gate rules of Sections~3--4. Instantaneous
activities have priority: whenever the marking is unstable,
enabled instantaneous activities are selected according to the
normalized weights of Definition~\ref{def:controlled-probabilistic-activity-network}. 
Their completions consume
no modeled physical time.

If instantaneous resolution terminates after finitely many steps,
timed processing may resume from the resulting stable marking.
An infinite instantaneous-completion sequence is
\emph{instantaneous divergence}: it does not produce a successor
stable marking or an observable timed completion. Its probability is 
accounted for by the operational semantics of Section 4 through the 
distinguished divergence state \(\Delta\), rather than by an ordinary 
transition to a successor stable marking.

Consider a stable marking $\mu$ at physical time $t$, after
instantaneous resolution has terminated. A timed activity $a$
is \emph{activated} at $t$ if it is enabled in $\mu$ and at
least one of the following holds:
\begin{enumerate}
\item
$a$ was inactive immediately before $t$;

\item
$a$ completed at $t$ and remains enabled;

\item
$\Pi(\mu,a)=\mathrm{true}$.
\end{enumerate}
Upon activation, $a$ receives a fresh nonnegative
activity-time requirement $\tau_a$ whose conditional marginal
distribution is $F(\cdot\mid a)$, with joint sampling governed
by the specified timing conventions.

While $a$ remains enabled without reactivation, it retains its
assigned requirement and its accumulated processing. Disabling
an activity makes it inactive and discards its outstanding
requirement. Subsequent enabling constitutes a new activation.
If an activity remains enabled but its reactivation predicate
becomes true at an activity-completion epoch, it discards its
previous requirement and accumulated processing and receives a
fresh requirement. Otherwise, it continues with its retained
requirement and accumulated processing. Changes in processing
rate alone do not constitute reactivation.

Let $h_t$ denote the history available to the controller at
the stable marking $\mu$ at time $t$. For each enabled timed
activity $a$, let $h_t\cdot a$ denote the corresponding
\emph{provisional completion history}: the history obtained by
appending the observation that $a$ is the next timed activity
to complete, but before applying its controlled marking update.
The provisional history contains no information about an
unobserved future outcome.

A deterministic policy $\pi$ provisionally assigns to each
enabled timed activity $a$ an admissible control action
\[
c_a=\pi(h_t\cdot a).
\]
Equivalently, the assignments form a contingent-action rule
\[
b_t:a\longmapsto c_a.
\]
The rule $b_t$ is fixed for the ensuing timed race. Its
assignment $b_t(a)$ determines the processing rate of activity
$a$ during that race. If $a$ completes next, $b_t(a)$ is the
control action executed for that completion. Assignments to
activities that do not complete are not executed.

Let $u_a$ be the most recent activation time of an active timed
activity $a$. Its accumulated processing at physical time $t$ is
\[
    W_a(t)
    =
    \int_{u_a}^{t}
    \rho(\mu_s,a,b_s(a))\,ds,
\]
where the integrand is evaluated only at times $s$ at which $a$ is
enabled, $\mu_s$ is the marking at physical time $s$, and $b_s$ is the
contingent-action rule in effect at physical time $s$. The integral is
reset to zero upon reactivation.

An activity becomes ready to complete when its accumulated
processing first reaches its assigned requirement. In
particular, a zero requirement may produce a completion at its
activation time; such a completion is still a timed-activity
completion and is not an instantaneous activity of $IA$.

The specified tie-resolution rule determines the order in which
simultaneously ready timed activities are handled.

Suppose that timed activity $a$ completes next. The controller
observes its identity and executes the action
\[
c=b_t(a)
\]
that was provisionally assigned to $a$ for the current race.
The applicable input- and output-gate rules and probabilistic
choices determine the resulting marking. Any ensuing
instantaneous activities are resolved before timed processing
resumes.

If instantaneous resolution terminates, the resulting stable
marking and updated history determine a new contingent-action
rule for the next timed race. Requirements and accumulated
processing of activities that remain enabled are retained or
reset according to the activation and reactivation rules above.
A new provisional assignment may change an activity's
processing rate without reactivating that activity.

Thus, for a completion that leads to a successor stable
marking, the decision order is
\[
\begin{aligned}
&\text{stable marking and current history}\\
&\quad\longrightarrow
  \text{provisional action assignment for each enabled activity}\\
&\quad\longrightarrow
  \text{timed race under the assigned processing rates}\\
&\quad\longrightarrow
  \text{completion and observation of activity }a\\
&\quad\longrightarrow
  \text{execution of the assigned action }b_t(a)\\
&\quad\longrightarrow
  \text{probabilistic marking update}\\
&\quad\longrightarrow
  \text{instantaneous resolution}\\
&\quad\longrightarrow
  \text{successor stable marking}.
\end{aligned}
\]
The final step is absent if instantaneous resolution diverges.
In particular, the action executed after observing $a$ is
selected through its provisional completion history before
the race, rather than chosen independently after the race
has been resolved.

A timed execution may also cease to produce observable
completions because no further timed activity can complete.
Such cessation, instantaneous divergence, and an infinite
sequence of genuine observable completions are distinct
behaviors. The timing specification must also address the
possibility of infinitely many timed completions in a bounded
physical-time interval (\emph{explosion}). Where a stochastic
process is considered beyond every finite physical time, we
assume non-explosion or explicitly specify its treatment.
\subsection{Semantic Model}

A controlled stochastic automaton abstracts behavior at stable states
while retaining the observable activity, contingent control action,
probabilistic successor-state kernel, and stochastic timing mechanism.
A contingent action assigns an admissible control action to each enabled
activity before a stochastic race. The assigned action determines that
activity's processing rate and is executed if the activity completes next.

\begin{definition}[Controlled stochastic automaton]
\label{def:csa}
A \emph{controlled stochastic automaton} (CSA) is a tuple
\[
    \mathcal V=(Q,\Sigma,C,P,\nu_0,F,\rho,\Pi),
\]
where
\[
    (Q,\Sigma,C,P,\nu_0)
\]
is a controlled probabilistic automaton (CPA) as in
Definition~\ref{def:controlled-probabilistic-automaton}, and
\begin{itemize}
\item
\[
    F=\{F(\cdot\mid a):a\in\Sigma\}
\]
specifies nonnegative activity-time requirement distributions;

\item
\[
    \rho:Q\times\Sigma\times C
    \longrightarrow\mathbb R_{\geq0}
\]
specifies control-dependent processing rates;

\item
\[
    \Pi:Q\times\Sigma
    \longrightarrow\{\mathrm{true},\mathrm{false}\}
\]
specifies reactivation.
\end{itemize}

The timing parameters are relevant only when the activity is enabled at
the current state. The processing rate $\rho(q,a,c)$ is relevant only
when $(a,c)$ is enabled at $q$, that is, when
\[
    P(q,a,c,q')>0
\]
for some $q'\in Q$. For an enabled pair,
$P(q,a,c,\cdot)$ is a probability distribution; for a disabled pair it
is the zero.
\end{definition}

The CSA has eight tuple components: the five components of its
underlying CPA and the three timing components $F$, $\rho$, and $\Pi$.
In particular, $\nu_0$ is an initial probability distribution, not an
initial-state set. The observable activity alphabet is denoted by
$\Sigma$, consistently with Definition~\ref{def:controlled-probabilistic-automaton}.

The activation, accumulated-processing, retention, and reactivation
rules of the CSA are those of the CSAN, with state $q$ in place of
marking $\mu$. Activity-time requirements of newly activated activities
are sampled independently according to their activity-specific
distributions specified by $F$. As for a CSAN, a \emph{fully specified
CSA} additionally fixes a rule for resolving simultaneous completions
when they can occur with positive probability and a convention for
forming and maintaining provisional control assignments during
stochastic races. These conventions determine the stochastic race;
they do not modify the underlying CPA transition kernel.

For $q\in Q$ and $a\in\Sigma$, define
\[
    C(q,a)
    =
    \left\{
        c\in C:
        \exists q'\in Q,\;
        P(q,a,c,q')>0
    \right\},
\]
and let
\[
    \Sigma(q)
    =
    \{a\in\Sigma:C(q,a)\neq\varnothing\}
\]
denote the enabled activity set at $q$. Thus, an activity $a$ is enabled
at $q$ precisely when $C(q,a)\neq\varnothing$.

A contingent action at $q$ is a map
\[
    b:\Sigma(q)\to C
\]
such that
\[
    b(a)\in C(q,a)
    \qquad
    \text{for every }a\in\Sigma(q).
\]
The set of all contingent actions at $q$ is therefore
\[
    B(q)
    =
    \left\{
        b:\Sigma(q)\to C
        \;\middle|\;
        b(a)\in C(q,a)
        \text{ for every }a\in\Sigma(q)
    \right\}.
\]
At the beginning of a stochastic race, a policy determines a contingent
action $b$ using the information available to the controller under the
specified observation convention. The assignments are fixed for that
race. During the race, activity $a$ processes at rate
\[
    \rho(q,a,b(a)).
\]
If $a$ completes next, the controller observes its identity and executes
the previously assigned action $b(a)$. Conditional on the completed
activity and executed action, the successor state is sampled according
to
\[
    P(q,a,b(a),\cdot).
\]
Assignments to activities that do not complete are not executed. After
a completion and the resulting transition to a successor state, a new
contingent action may be determined. A change of processing rate alone
does not reactivate an unfinished activity.

The information used to determine a contingent action may include the
current observable history and any timing information already available
to the controller, but it may not depend on an as-yet-unobserved
completion time or probabilistic transition outcome. If completion times
or other timing information become observable, they may be used in
subsequent decisions after they have been observed. The policy classes
and observation conventions are developed below.

The CPA kernel determines the successor-state distribution
\emph{conditional on} $(q,a,c)$; it does not determine the distribution
of the next completed activity or its completion time. Those
distributions depend on the current timing configuration, the sampled
activity-time requirements, their evolution under the applicable
processing rates, and the contingent action governing those rates.
Although activity-time requirements are sampled independently, their
subsequent completion dynamics may be coupled through the
control-dependent processing-rate function. For general activity-time 
distributions, the conditional joint distribution of the time and identity of 
the next activity completion may depend on residual timing information 
in addition to the current state \(q\) and contingent action \(b\).

Thus, the underlying CPA specifies the probabilistic state transition
associated with a completed activity and executed control action,
whereas the timing components $F$, $\rho$, and $\Pi$, together with the
operational timing conventions, determine the stochastic evolution
between successive activity completions.
\subsection{Bisimulation}
We use a structural, timing-preserving bisimulation that matches the
probabilistic branching of the underlying CPAs and the timing parameters
of corresponding observable activities and admissible control actions.
\begin{definition}[Timing-preserving bisimulation]
Let
\[
 \mathcal V_i
 =(Q_i,\Sigma,C,P_i,\nu_{0,i},F_i,\rho_i,\Pi_i),
 \qquad i\in\{1,2\},
\]
be CSAs with the same observable activity alphabet and
control-action set. A relation $\sim$ on the disjoint union
$Q_1\uplus Q_2$ is a
\emph{timing-preserving bisimulation} if it is an equivalence
relation satisfying the following conditions:
\begin{enumerate}
\item
Its restriction to the underlying CPAs is a probabilistic
bisimulation in the sense of Definition~4.3, including the
enabled-pair, transition-probability, and
initial-distribution conditions.

\item Whenever $q\in Q_1$, $q'\in Q_2$,
and $q\sim q'$, then, for every activity
$a\in\Sigma$ enabled at both $q$ and $q'$,
\[
F_1(\cdot\mid a)=F_2(\cdot\mid a),
\qquad
\Pi_1(q,a)=\Pi_2(q',a),
\]
and, for every admissible control action
$c\in C(q,a)=C(q',a)$,
\[
\rho_1(q,a,c)=\rho_2(q',a,c).
\]
\end{enumerate}
Two CSAs are \emph{structurally timing-equivalent} if such
a bisimulation exists.
\end{definition}

The enabled-pair condition is essential: the CPA transition
kernel and the processing-rate semantics are relevant only
for admissible state--activity--action combinations.
Equality of the activity-time requirement distributions
and reactivation predicates is required only for enabled
activities, while equality of processing rates is required
for every admissible action associated with each such
activity.

For related states, equality of the admissible control sets
$C(q,a)$ identifies the same admissible contingent-action maps.
Corresponding maps assign the same control action to each enabled
activity, and the processing-rate condition ensures that these
assignments induce matching activity-specific processing rates.
\begin{proposition}
Structural timing-equivalence of CSAs is an equivalence
relation.
\end{proposition}

\begin{proof}
Reflexivity follows from the identity relation, and
symmetry follows by reversing a timing-preserving
bisimulation.

For transitivity, suppose that $\mathcal V_1$ is
structurally timing-equivalent to $\mathcal V_2$ and that
$\mathcal V_2$ is structurally timing-equivalent to
$\mathcal V_3$. The corresponding CPA bisimulation
conditions compose as in Proposition~4.4. For states
matched through an intermediate state of $\mathcal V_2$,
the enabled-pair conditions identify the same observable
activities and admissible control actions. Equality of
$F$, $\rho$, and $\Pi$ is transitive on these corresponding
domains. The initial-distribution conditions also
compose. The induced equivalence relation therefore
satisfies the CPA and timing-matching conditions between
$\mathcal V_1$ and $\mathcal V_3$.
\end{proof}
\subsection{Operational Semantics}
We now relate the abstract CSA model to the operational behavior of a
CSAN. The underlying controlled probabilistic activity network determines
the CPA component as in Section~4, while the CSAN timing components
determine the corresponding activity-time distributions, processing
rates, and reactivation rules. The following realization relation formalizes this correspondence.
\begin{definition}[Realization of a CSA]
Let $\mathcal N$ be a CSAN with initial marking $\mu_0$,
and let $\mathcal L$ be its underlying controlled
probabilistic activity network obtained by omitting
$F$, $\rho$, and $\Pi$. Suppose that
$(\mathcal L,\mu_0)$ realizes the CPA
\[
    \mathcal U=(Q,\Sigma,C,P,\nu_0)
\]
under the operational semantics of Section~4.

The CSAN $(\mathcal N,\mu_0)$ \emph{realizes} the CSA
\[
    \mathcal V
    =(Q,\Sigma,C,P,\nu_0,
      F_{\mathcal V},\rho_{\mathcal V},\Pi_{\mathcal V})
\]
if, for every reachable stable marking $q\in Q$ and
every activity $a\in\Sigma$ enabled there,
\[
\begin{aligned}
    F_{\mathcal V}(\cdot\mid a)
        &=F_{\mathcal N}(\cdot\mid a),\\
    \Pi_{\mathcal V}(q,a)
        &=\Pi_{\mathcal N}(q,a),
\end{aligned}
\]
and, for every admissible control action
$c\in C(q,a)$,
\[
    \rho_{\mathcal V}(q,a,c)
        =\rho_{\mathcal N}(q,a,c).
\]

For a realization of a \emph{fully specified} CSA, the
sampling, reactivation, provisional-assignment, and
simultaneous-completion conventions are inherited from
the fully specified CSAN under the state--stable-marking
correspondence.
\end{definition}

Realization preserves the contingent-action and timing
semantics. At each corresponding state and stable marking,
the same admissible contingent actions are available, and
corresponding assignments induce the same activity-specific
processing rates. When an activity completes, its assigned
control action is executed, and the common CPA transition kernel determines the successor state.

Thus, realization preserves both the probabilistic
successor-state behavior and the stochastic timing
mechanism of the CSAN at the CSA level.

\begin{definition}[Structural timing-equivalence of CSANs]
Two CSANs with specified initial markings are
\emph{structurally timing-equivalent} if the CSAs
they realize are structurally timing-equivalent
under the preceding timing-preserving bisimulation.
\end{definition}

Structural timing-equivalence does not, by itself,
imply equality of the resulting physical-time
process laws. Such a conclusion additionally requires
compatible initial timing configurations, conditional
joint sampling laws, clock-retention and reactivation
conventions, provisional-assignment conventions, and
simultaneous-completion rules, together with
corresponding policies and an appropriate treatment
of explosion.

\begin{proposition}
Structural timing-equivalence of CSANs is an
equivalence relation.
\end{proposition}

\begin{proof}
The result follows from the preceding proposition
on structural timing-equivalence of CSAs and the
definition of realization.
\end{proof}
\subsection{Policy Types}

Let
\[
    \mathcal V=(Q,\Sigma,C,P,\nu_0,F,\rho,\Pi)
\]
be a fully specified CSA. At the beginning of each stochastic
race, the controller provisionally assigns an admissible control
action to every enabled activity. These assignments determine the
activities' processing rates during that race. If activity $a$
completes next, its assigned action is executed before the
corresponding successor state is sampled. Assignments to activities
that do not complete are not executed.

For each $q\in Q$, let $\Sigma(q)$ and $B(q)$ be the enabled-activity
set and the admissible contingent-action set defined in
Section~5.3. Thus, a contingent action is an activity-indexed map,
not an individual control action. The set $B(q)$ is canonically
identified with the Cartesian product
\[
    \prod_{a\in\Sigma(q)} C(q,a).
\]
When $\Sigma(q)=\varnothing$, the empty map is the unique member
of $B(q)$.

Let
\[
    0\leq T_0\leq T_1\leq T_2\leq\cdots
\]
denote successive observable activity-completion times, whenever
they exist. Equality between successive completion times is
permitted. In particular, zero activity-time requirements or
simultaneous activity completions may yield distinct observable
completion events with the same physical time stamp; the specified
simultaneous-completion rule determines their resolution.

Write $Q_n$ for the state at the beginning of the race preceding
the completion at $T_n$, $B_n$ for the contingent action assigned
to that race, $A_n$ for the activity that completes, and $Q_{n+1}$
for the successor state, when it exists. The action executed at
that completion is
\[
    C_n=B_n(A_n).
\]
The decision and transition order is therefore
\[
    Q_n
    \;\longrightarrow\;
    B_n
    \;\longrightarrow\;
    (T_n,A_n)
    \;\longrightarrow\;
    C_n=B_n(A_n)
    \;\longrightarrow\;
    Q_{n+1}.
\]
The arrow from $B_n$ to $(T_n,A_n)$ represents the stochastic race
under the assigned processing rates; it does not assert that
$Q_n$ and $B_n$ alone determine the next-completion law.

The index $n$ counts observable completions, not distinct physical
times. If several activities complete at the same physical time, they
are resolved according to the specified simultaneous-completion rule.
Each resolved completion is treated as a separate decision epoch.
After a completion produces a successor state, the resulting state and
timing configuration determine which unresolved simultaneous
completions, if any, remain applicable. A new contingent action is then
determined for the next race.

For every enabled state--activity--action combination and every
successor state $q'$,
\[
\begin{aligned}
    &\Pr\bigl(
        Q_{n+1}=q'
        \mid Q_n=q,\ B_n=b,\ A_n=a,\ H_n
      \bigr) \\
    &\qquad =P\bigl(q,a,b(a),q'\bigr),
\end{aligned}
\]
whenever the indicated conditional probability is defined and the
completion leads to a successor state. Here $H_n$ denotes the
information available before $B_n$ is selected, as specified
below. The timing mechanism and $B_n$ jointly determine the
next-completion law, whereas $P$ determines the successor-state
distribution conditional on the completed activity and its
assigned action.

Under full observation of states, completion times, completed
activities, and previously executed actions, the controller's
history at the beginning of the $n$-th race may be represented by
\[
\begin{aligned}
    H_n=\bigl(&Q_0,\,
       T_0,A_0,C_0,Q_1,\ldots,\\
       &T_{n-1},A_{n-1},C_{n-1},Q_n\bigr).
\end{aligned}
\]
In particular, $H_n$ does not contain the as-yet-unobserved
completion time $T_n$, the completing activity $A_n$, or the
contingent action $B_n$ that is about to be selected.

For each $a\in\Sigma(Q_n)$, define the
\emph{provisional completion history} $H_n\cdot a$ as the current
history augmented by the hypothetical observation that $a$ is the
next completed activity, before its controlled successor-state
transition is resolved. This notation identifies the activity for
which an action is provisionally assigned; it does not supply an
unknown future completion time, transition outcome, or other
information unavailable when the assignment is made.

For a deterministic contingent-action policy, the provisional
assignments may be expressed as
\[
    B_n(a)=\pi(H_n\cdot a),
    \qquad a\in\Sigma(Q_n),
\]
subject to
\[
    \pi(H_n\cdot a)\in C(Q_n,a).
\]
Equivalently, the policy selects the single map
$B_n\in B(Q_n)$ before the race. The map is fixed throughout that
race, and $B_n(a)$ determines the processing rate of activity $a$
during the race.

The displayed history is a full-observation representation, not
an assumption that every CSA controller observes all of its
entries. More generally, $H_n$ denotes the history supplied by the
specified observation model, and a policy may use only information
available in that history. If the observation model supplies
additional information, such as activity ages, accumulated
processing, outstanding requirements, or clock status, it must be
included in the observation history or in an explicitly augmented
observable decision state.

The stochastic model's full timing configuration and the
controller's observation history are distinct notions. For general
activity-time distributions, future completion laws may depend on
outstanding requirements and their dependence structure even when
those quantities are not observable to the controller.

For any finite set $X$, let $\operatorname{Dist}(X)$ denote the
set of probability distributions on $X$.  As in the
probabilistic semantics of Section~4, we retain an activity-wise
view of control randomization.  The difference is that, in the
stochastic semantics, the control assignments must be made
before the timed race is resolved, since they may affect the
processing rates of the competing activities.

Accordingly, a history-dependent randomized contingent-action
policy specifies, for every reachable race-decision
configuration and every $a\in\Sigma(Q_n)$, a conditional
distribution
\[
    \kappa_{n,a}(\cdot\mid H_n)
    \in \operatorname{Dist}(C(Q_n,a)).
\]
To retain the activity-wise character of the probabilistic-policy
semantics developed in Section~4, we assume that these
provisional activity-specific control assignments are
conditionally independent given $H_n$.  Hence, if $B_n$ denotes
the random contingent-action map selected before the race, then
\[
    \Pr(B_n=b\mid H_n)
    =
    \prod_{a\in\Sigma(Q_n)}
    \kappa_{n,a}\bigl(b(a)\mid H_n\bigr),
    \qquad b\in B(Q_n).
\]
For $\Sigma(Q_n)=\varnothing$, the empty product is one, so the
unique empty contingent-action map has probability one.

The conditional-independence convention concerns only the
controller's randomization of the provisional activity-specific
control assignments. It is distinct from the independence
assumptions governing activity-time requirements in the
stochastic semantics. Once $B_n$ has been sampled, its
components determine the control-dependent processing rates
of the competing activities.

A policy is \emph{deterministic} if each
$\kappa_{n,a}(\cdot\mid H_n)$ is a Dirac distribution. Equivalently,
it selects one admissible contingent-action map as a function of the
information available at the beginning of each race. The policy kernels
are required to be measurable with respect to the specified observation
histories.

A policy is \emph{Markov relative to a specified observable
race-decision state} $S_n$ if, for each enabled activity $a$, its
conditional control distribution depends on the available history only
through $S_n$ and, when permitted, the current physical time at the
beginning of the race. Writing $t_n$ for that current time, a
time-dependent Markov policy has component kernels of the form
\[
    \kappa_{n,a}(c\mid H_n)
    =
    \psi_a(c\mid t_n,S_n),
    \qquad c\in C(Q_n,a),
\]
and the induced distribution on $B(Q_n)$ is the corresponding product
distribution. A \emph{stationary Markov policy} is independent of
absolute time and decision index:
\[
    \kappa_{n,a}(c\mid H_n)
    =
    \psi_a(c\mid S_n).
\]
These descriptions are relative to the chosen observable decision state
and observation model; they do not assert that the unaugmented
physical-state process has the Markov property.

When the controller uses only the current state, a stationary randomized
policy is specified by
\[
    \psi_{q,a}(\cdot)
    \in\operatorname{Dist}\bigl(C(q,a)\bigr),
    \qquad a\in\Sigma(q),
\]
at each reachable state $q$, with induced product distribution
\[
    \varphi(b\mid q)
    =
    \prod_{a\in\Sigma(q)}
    \psi_{q,a}\bigl(b(a)\bigr),
    \qquad b\in B(q).
\]
We call this a \emph{memoryless policy with respect to state
observations}.

A deterministic memoryless contingent-action policy is
equivalently a map
\[
    \varphi:Q\to\bigcup_{q\in Q}B(q),
    \qquad
    \varphi(q)\in B(q)
\]
on reachable states. Its activity-specific assignments can
equivalently be written
\[
    \widehat{\varphi}(q,a)\in C(q,a),
    \qquad a\in\Sigma(q),
\]
where
\[
    \varphi(q)(a)=\widehat{\varphi}(q,a).
\]
This activity-indexed representation is the componentwise form of
a deterministic contingent-action memoryless policy. It also shows
how the action choices of a deterministic activity-conditioned
memoryless policy can be assembled into a contingent-action map
before activity generation.

The term ``memoryless'' describes the controller's dependence on
its observations. It does not imply that the physical-state process
is Markov. For general activity-time distributions, future
completion laws may depend on elapsed processing, outstanding
requirements, or other clock information. A Markov representation
of the controlled dynamics may therefore require an augmented state
\[
    S_n=(Q_n,\Xi_n),
\]
where $\Xi_n$ contains sufficient timing information. The choice of
$\Xi_n$ depends on the specified timing semantics. If $\Xi_n$ is
not observable, a controller cannot automatically use it as its
decision state; an appropriate observable information state may
instead be required.

\paragraph{Example (routing arrivals to parallel queues).}

Consider a system with $m$ parallel queues and state
\[
    q=(n_1,\ldots,n_m),
\]
where $n_j$ is the number of customers in queue $j$. Suppose that
the enabled arrival activity is $a_{\mathrm{arr}}$ and that $c_j$
assigns an arriving customer to queue $j$. A deterministic
memoryless contingent-action policy may assign
\[
    b_q(a_{\mathrm{arr}})=c_j
\]
before the stochastic race. The arrival activity then processes
at rate
\[
    \rho(q,a_{\mathrm{arr}},c_j).
\]
If the arrival completes next, the assigned routing action is
executed. In a simple deterministic successor-state model,
\[
    P(q,a_{\mathrm{arr}},c_j,q+e_j)=1,
\]
where $e_j$ is the $j$-th unit vector.

The provisional routing assignment may therefore affect the
arrival's completion time when its processing rate depends on the
selected action. If arrivals have different observable classes
represented by different activities, the contingent action may
assign different routing actions to those activities. A policy
may additionally use past observations or available timing
information, provided that these belong to its specified
observation history.

The contingent-action policy architecture above is the
continuous-time counterpart of the contingent-action architecture
introduced in Section~4. In both settings, the controller selects
a contingent action $B_n$ before the activity to be executed is
determined. In Section~4, an activity-generation distribution
determines the next activity according to
\[
    Q_n
    \longrightarrow B_n
    \longrightarrow A_n
    \longrightarrow C_n=B_n(A_n)
    \longrightarrow Q_{n+1},
\]
whereas in a CSA the stochastic timing mechanism and the assigned
processing rates determine both the next completion time and the
completing activity:
\[
    Q_n
    \longrightarrow B_n
    \longrightarrow (T_n,A_n)
    \longrightarrow C_n=B_n(A_n)
    \longrightarrow Q_{n+1}.
\]
Thus, the two settings share the same contingent-action decision
order, while differing in their activity-generation mechanisms.
In particular, in a CSA the selected contingent action may
influence which activity completes next and when through the
control-dependent processing rates.

\subsection{Reward Models and PAC-Optimality}
We equip a fully specified CSA
\[
    \mathcal V=(Q,\Sigma,C,P,\nu_0,F,\rho,\Pi)
\]
with a reward structure
\[
    r=(r_{\mathrm{rate}},r_{\mathrm{act}}).
\]

Recall that, for each $q\in Q$, the set of admissible contingent
actions is
\[
    B(q)
    :=
    \left\{
        b:\Sigma(q)\to C
        \;\middle|\;
        b(a)\in C(q,a)
        \text{ for every }a\in\Sigma(q)
    \right\}.
\]
Let
\begin{equation} 
    B:=\bigcup_{q\in Q}B(q),
    \qquad
    \mathsf{Dom}(B)
    :=
    \{(q,b):q\in Q,\ b\in B(q)\}.
    \label{dom(B)}
\end{equation}

The state--contingent-action reward-rate function
\[
    r_{\mathrm{rate}}:\mathsf{Dom}(B)\to\mathbb R
\]
specifies the reward accumulated per unit of physical time while
the system remains in state $q$ under the selected contingent
action $b\in B(q)$, whereas
\[
    r_{\mathrm{act}}:Q\times\Sigma\times C\times Q
    \to\mathbb R
\]
specifies the reward associated with a genuine observable activity
completion and its resulting transition. The completion-reward
function is evaluated only on admissible transitions.

At the beginning of each timed race in state $q$, the
controller selects a contingent action $b\in B(q)$.
For every enabled activity $a\in\Sigma(q)$, the
provisionally assigned action $b(a)$ determines its
processing rate
\[
\rho(q,a,b(a))
\]
during that race. If $a$ completes next, the action
$b(a)$ is executed and the resulting transition
receives its completion reward. Actions assigned to
activities that do not complete are not executed and
receive no completion reward.

Under the present reward structure, the reward rate
during a race is determined by the state occupied
during that race. The contingent action may nevertheless
affect the accumulated reward by changing the timing
and identity of subsequent completions. An explicitly
action-dependent reward-rate function would constitute
a different, more general reward structure.

For simplicity, assume that both reward functions are measurable and
uniformly bounded. Let $X^\pi(u)$ denote the state occupied at physical
time $u$ under an admissible policy $\pi$, wherever the physical-time
state process is defined, and let $B^\pi(u)$ denote the contingent
action in force at time $u$. Thus, during the $n$-th timed race,
$X^\pi(u)=Q_n$ and $B^\pi(u)=B_n$.

Let $T_n$ be the time of the $n$-th observable completion, $Q_n$ the
state preceding its timed race, $B_n\in B(Q_n)$ the contingent action
assigned to that race, $A_n$ the activity that completes, and
$Q_{n+1}$ its successor state. The action executed at completion is
\[
    C_n=B_n(A_n).
\]

For $0\le s\le t$, define the accumulated reward on an execution for
which the following expression is well-defined by
\begin{equation}
\label{eq:csa-accumulated-reward}
R_{s,t}^{\pi}
=
\int_s^t
r_{\mathrm{rate}}\bigl(X^\pi(u),B^\pi(u)\bigr)\,du
+
\sum_{n:\,s<T_n\le t}
r_{\mathrm{act}}
\bigl(Q_n,A_n,B_n(A_n),Q_{n+1}\bigr).
\end{equation}
The sum includes only genuine observable activity
completions that have an associated successor-state
transition. No completion reward is assigned to
fictitious transitions representing termination,
instantaneous divergence, or explosion. A final
genuine completion contributes its completion reward
even if no subsequent observable completion occurs.

The convention $s<T_n\leq t$ excludes completions
at the left endpoint and includes completions at
the right endpoint. If multiple genuine completions
occur at the same physical time, each contributes
its completion reward according to the specified
simultaneous-completion semantics. No positive
physical-time reward accumulates between completions
resolved at the same instant.

The reward functional is evaluated only over the
portion of the execution for which the physical-time
semantics is defined.  If the execution ceases to produce observable completions, its
subsequent physical-time behavior and reward accumulation must be
specified by the model. In particular, persistence in a state with
no further possible completion requires an explicit convention
specifying both its physical-time continuation and any reward
accrued during that continuation.

For the finite-interval functional~(\ref{dom(B)}), assume that
$X^\pi(u)$ and $B^\pi(u)$ are defined throughout $[s,t]$
and that the completion sum is well defined. Non-explosion on
$[s,t]$, together with bounded completion rewards, is sufficient
for the completion sum to be finite. Bounded individual completion
rewards alone do not guarantee a finite sum when infinitely many
completions occur in a bounded physical-time interval.

Let $\mathcal P$ be a specified class of admissible policies, with its
observation model and permitted timing information fixed, and let
$z_0$ denote an initial condition sufficient to initialize the
stochastic timing semantics. We assume that $\mathcal P$ and $z_0$
make the reward objective under consideration well defined.

For real-valued random variables $Y$ and $Z$,
write
\[
Y\leq_{\mathrm{st}} Z
\quad\Longleftrightarrow\quad
\Pr(Y>x)\leq\Pr(Z>x)
\quad\text{for every }x\in\mathbb R.
\]
A policy $\pi^*\in\mathcal P$ is
\emph{stochastically optimal} for the interval
$[s,t]$ and a specified set $\mathcal Z_0$
of initial conditions if
\[
R_{s,t}^{\pi}\mid z_0
\leq_{\mathrm{st}}
R_{s,t}^{\pi^*}\mid z_0
\]
for every $\pi\in\mathcal P$ and every
$z_0\in\mathcal Z_0$.

Such a policy need not exist: reward
distributions associated with different
policies may be incomparable under
stochastic order.

An alternative objective is expected reward.
When the relevant expectations are finite,
define
\[
V_{s,t}^{\pi}(z_0)
=
\mathbb E_{z_0}^{\pi}
\left[R_{s,t}^{\pi}\right],
\qquad
V_{s,t}^{*}(z_0)
=
\sup_{\pi\in\mathcal P}
V_{s,t}^{\pi}(z_0).
\]
A policy $\pi^*\in\mathcal P$ is
\emph{expected-reward optimal} over
$\mathcal Z_0$ if
\[
V_{s,t}^{\pi^*}(z_0)
=
V_{s,t}^{*}(z_0)
\]
for every $z_0\in\mathcal Z_0$.
A policy $\pi\in\mathcal P$ is
\emph{$\varepsilon$-optimal} over
$\mathcal Z_0$ if
\[
V_{s,t}^{\pi}(z_0)
\geq
V_{s,t}^{*}(z_0)-\varepsilon
\]
for every $z_0\in\mathcal Z_0$.

Expected-reward optimality and
stochastic-order optimality are distinct
notions. A result established for one
objective does not automatically apply
to the other.

\begin{remark}[Scope of PAC-optimality]
\label{rem:csa-pac}
Probably Approximately Correct (PAC)
guarantees concern the output of a
learning algorithm, rather than the
definition of a CSA policy.

Fix an expected-reward objective, a
model class, an admissible policy class
$\mathcal P$, an observation model,
and an initial-condition set
$\mathcal Z_0$. Suppose a learning
algorithm returns a possibly randomized
policy $\widehat\pi\in\mathcal P$.
For a specified initial condition
$z_0\in\mathcal Z_0$, a PAC guarantee
may take the form
\[
\Pr_{\mathrm{learn}}
\left(
V_{s,t}^{\widehat\pi}(z_0)
\geq
V_{s,t}^{*}(z_0)-\varepsilon
\right)
\geq 1-\delta,
\]
where $\varepsilon>0$, $0<\delta<1$,
and the probability is over the
learning algorithm's sampling and
internal randomization. The value
$V_{s,t}^{\widehat\pi}(z_0)$ is the
expected reward of the returned
policy under the specified CSA
execution semantics, including its
pre-race contingent-action assignments.

The displayed guarantee concerns
one specified initial condition.
A guarantee holding simultaneously
over $\mathcal Z_0$ requires the
stronger event
\[
\Pr_{\mathrm{learn}}
\left(
\forall z_0\in\mathcal Z_0:
V_{s,t}^{\widehat\pi}(z_0)
\geq
V_{s,t}^{*}(z_0)-\varepsilon
\right)
\geq 1-\delta
\]
and assumptions sufficient to
establish it.

A PAC statement must also specify
the learning protocol, the information
available to the learner, and the
sample or interaction complexity
required to achieve the guarantee.
The displayed inequalities alone
do not establish such a complexity
bound.

PAC results for finite-state
discrete-time MDPs do not
automatically extend to arbitrary
CSAs with general activity-time
distributions. Such an extension
requires a specified reduction or
learning procedure and assumptions
controlling the state space, timing
information, exploration, rewards,
and approximation error. In
particular, the learner must account
for the fact that provisional
control assignments may affect
the timed race itself.

Uniformization is exact only when the assumptions required by the
relevant Markovian model are satisfied. It is not a general exact
reduction for arbitrary non-exponential activity-time distributions.
\end{remark}
The reward formulation specifies the objective being optimized and the
information available to the controller. We next distinguish general
CSA dynamics from the Markovian subclass and develop its corresponding
rate-based representation.
\subsection{Controlled Markovian Automata}

\begin{definition}[State process]
\label{def:csa-state-process}
The state process of a controlled stochastic
model under an admissible policy $\pi$ is
\[
 X^\pi=\{X^\pi(t):t\geq0\},
\]
where $X^\pi(t)$ denotes its physical state
at time $t$, whenever the process is defined.

An augmented state process may additionally
record active activities, accumulated processing,
outstanding requirements, the currently assigned
contingent action, or other timing information.
Its state space and observation convention must
be specified separately.
\end{definition}

\begin{definition}[Stochastic equivalence]
\label{def:csa-stochastic-equivalence}
Let
\[
 X=\{X(t):t\geq0\},
 \qquad
 X'=\{X'(t):t\geq0\}
\]
be stochastic processes with state spaces
$Q$ and $Q'$, respectively. Let $\sim$ be
an equivalence relation on $Q\uplus Q'$.

The processes are \emph{stochastically
equivalent under $\sim$} if their projections
onto the equivalence classes of $\sim$
have the same finite-dimensional
distributions. Explicitly, for every
$n\geq0$, every
$0\leq t_0<\cdots<t_n$, and every sequence
of equivalence classes $D_0,\ldots,D_n$,
\[
\begin{aligned}
 &\Pr\bigl(
 X(t_i)\in D_i\cap Q,\;
 0\leq i\leq n
 \bigr)\\
 &\qquad=
 \Pr\bigl(
 X'(t_i)\in D_i\cap Q',\;
 0\leq i\leq n
 \bigr).
\end{aligned}
\]

The processes are \emph{stochastically
isomorphic} if there is a bijection
$f:Q\to Q'$ such that $f(X)$ and $X'$
have the same finite-dimensional
distributions.

Stochastic equivalence is understood
relative to specified initial conditions
and policies. Equality of state-process
finite-dimensional distributions does
not, by itself, establish equality of
activity-labeled histories or
completion-reward distributions.
\end{definition}
The exponential distribution plays a distinguished role in stochastic
models because of its memoryless property.  This motivates the following
Markovian subclass of controlled stochastic automata.
\begin{definition}[Markovian controlled stochastic automaton]
\label{def:markovian-csa}
A controlled stochastic automaton (CSA) is called \emph{Markovian} if
the activity-time distribution of every activity is exponential.
\end{definition}
For a Markovian CSA, the exponential activity-time distributions can
be represented by their completion rates.  This motivates the following
rate-based controlled automaton model.
\begin{definition}[Controlled Markovian automaton]
\label{def:controlled-markovian-automaton}
A \emph{controlled Markovian
automaton} (CMA) is a tuple
\[
 \mathcal W
 =
 (Q,\Sigma,C,P,\nu_0,\sigma),
\]
where
$(Q,\Sigma,C,P,\nu_0)$
is a CPA and
\[
 \sigma:
 Q\times\Sigma\times C
 \longrightarrow\mathbb R_{\geq0}
\]
assigns a completion rate
to each state--activity--action
combination.

Only values on admissible
state--activity--action combinations affect the
operational semantics. We set
\[
 \sigma(q,a,c)=0
 \qquad
 \text{whenever }
 P(q,a,c,q')=0
 \text{ for every }q'\in Q.
\]

For $q\in Q$, define
\[
 B(q)=
 \left\{
 b:\Sigma(q)\to C
 \;\middle|\;
 b(a)\in C(q,a)
 \text{ for every }a\in\Sigma(q)
 \right\}
\]
and, for $b\in B(q)$,
\[
 \Lambda(q,b)
 =
 \sum_{a\in\Sigma(q)}
 \sigma(q,a,b(a)).
\]
We require
\[
 \Lambda(q,b)<\infty
 \qquad
 \text{for every }q\in Q
 \text{ and }b\in B(q).
\]

At the beginning of a timed
race in state $q$, the
controller selects a
contingent action $b\in B(q)$.
The selected map determines
the completion rates
$\sigma(q,a,b(a))$ for all
enabled activities.

If $\Lambda(q,b)>0$,
the next completion occurs
after an exponential
waiting time of rate
$\Lambda(q,b)$ and has
identity $a$ with
probability
\[
 \frac{\sigma(q,a,b(a))}
      {\Lambda(q,b)}.
\]
The action $b(a)$ is then
executed, and the successor
state is sampled from
\[
 P(q,a,b(a),\cdot).
\]
A new contingent action
is selected before the
next race.

If $\Lambda(q,b)=0$,
no further observable
completion occurs while
the state and contingent
action remain unchanged.

A complete physical-time interpretation additionally specifies the
treatment of explosion, if it is possible.
\end{definition}
\begin{proposition}[Representation by a CMA]
\label{prop:csa-cma-representation}
Every Markovian controlled stochastic automaton in the sense of
Definition~\ref{def:markovian-csa} admits a CMA representation
with the same states, observable activities, control actions,
initial distribution, and controlled transition kernel.

Under corresponding admissible contingent-action policies, the
two models have identical conditional laws of successive completion
times, completed activity identities, executed actions, and successor
states. Consequently, under matching initial conditions and
non-explosion, their physical state processes are stochastically
isomorphic under the identity state correspondence.
\end{proposition}

\begin{proof}
Let
\[
    \mathcal V=(Q,\Sigma,C,P,\nu_0,F,\rho,\Pi)
\]
be a Markovian CSA. By
Definition~\ref{def:markovian-csa}, the activity-time distribution
of every activity is exponential. Let $\lambda_a$ denote the
corresponding exponential rate parameter of activity $a$.

Construct the CMA
\[
    \mathcal W=(Q,\Sigma,C,P,\nu_0,\sigma)
\]
using the same underlying CPA and define the completion-rate
function $\sigma$ according to the processing-rate semantics of
$\mathcal V$. In particular, for every admissible
state--activity--action combination $(q,a,c)$,
\[
    \sigma(q,a,c)
    =
    \lambda_a\,\rho(q,a,c),
\]
with $\sigma(q,a,c)=0$ on inadmissible combinations.

Fix corresponding admissible contingent-action policies. At each
race-decision epoch in state $q$, suppose that the same contingent
action $b\in B(q)$ is selected in the two models. By construction,
activity $a\in\Sigma(q)$ has completion rate
\[
    \sigma(q,a,b(a))
\]
in both models. Hence, with
\[
    \Lambda(q,b)
    =
    \sum_{a\in\Sigma(q)}
    \sigma(q,a,b(a)),
\]
the conditional law of the next completion time is exponential
with rate $\Lambda(q,b)$, and, whenever $\Lambda(q,b)>0$, the
conditional probability that activity $a$ completes next is
\[
    \frac{\sigma(q,a,b(a))}{\Lambda(q,b)}.
\]
Thus the two models have the same conditional law of the next
completion time and completed activity.

If activity $a$ completes, both models execute the preassigned
action $b(a)$ and sample the successor state according to the same
kernel
\[
    P(q,a,b(a),\cdot).
\]
The corresponding policies then select contingent actions for the
next race according to the same policy law.

Induction over successive completion epochs therefore establishes
equality of the finite activity-labeled timed-history distributions.
Under matching initial conditions and non-explosion, this equality
extends to the corresponding physical-time state processes, which
are stochastically isomorphic under the identity state
correspondence.
\end{proof}
\begin{remark}[Corresponding randomized policies]
\label{rem:cma-randomized-policies}
For randomized contingent-action policies, correspondence in
Proposition~\ref{prop:csa-cma-representation} requires matching
activity-specific control kernels and hence matching product
distributions over contingent-action maps.  Under this
correspondence, the activity-labeled timed execution laws established
in the proposition are preserved.
\end{remark}

The corresponding Markovian subclass can be defined directly at the
activity-network level.
\begin{definition}[Markovian controlled stochastic activity network]
A \emph{Markovian controlled stochastic 
activity network} is a
CSAN in which the activity-time distribution
of every timed activity is exponential.
\end{definition}
\begin{definition}[Controlled Markovian activity network]
\label{def:cman}
A \emph{controlled Markovian activity network} (CMAN) is a tuple
\[
    \mathcal N
    =
    (P,IA,TA,CA,\Gamma,IG,OG,IR,IOR,TOR,IP,\sigma),
\]
where
\[
    (P,IA,TA,CA,\Gamma,IG,OG,IR,IOR,TOR,IP)
\]
is a controlled probabilistic activity network and
\[
    \sigma:\mathbb{N}^{n}\times TA\times CA
    \longrightarrow \mathbb{R}_{\geq 0}
\]
is a control-dependent processing-rate function.

The behavior of a CMAN is that of the corresponding Markovian
CSAN in which every timed activity has an exponential activity-time
distribution with rate $1$, and $\sigma$ is its control-dependent
processing-rate function. Thus, for a marking
$q\in\mathbb{N}^{n}$, timed activity $a\in TA$, and admissible
control $c\in CA$, the effective completion rate of $a$ is
$\sigma(q,a,c)$.
\end{definition}
\begin{definition}[CMA realization of a CMAN]
\label{def:cman-cma-realization}
Let
\[
    \mathcal N
    =
    (P,IA,TA,CA,\Gamma,IG,OG,IR,IOR,TOR,IP,\sigma)
\]
be a controlled Markovian activity network, and let
\[
    (Q,\Sigma,C,P,\nu_0)
\]
be the CPA realized by its underlying controlled probabilistic
activity network with initial marking in the usual manner.

The \emph{CMA realized by $\mathcal N$} at that initial marking is
\[
    \mathcal W_{\mathcal N}
    =
    (Q,\Sigma,C,P,\nu_0,\sigma_{\mathcal W}),
\]
where the completion-rate function
\[
    \sigma_{\mathcal W}:Q\times\Sigma\times C
    \longrightarrow \mathbb{R}_{\geq 0}
\]
is induced by the CMAN rate function $\sigma$ under the usual
correspondence between reachable network markings, timed
activities, and admissible controls and the states, observable
activities, and control actions of the realized CPA.

In particular, whenever a realized state $q$ corresponds to a
reachable marking $m$, an observable activity $a$ corresponds to
a timed activity $\tau$, and $c$ is the corresponding admissible
control,
\[
    \sigma_{\mathcal W}(q,a,c)
    =
    \sigma(m,\tau,c).
\]
\end{definition}
\begin{proposition}[Consistency of CMAN realizations]
\label{prop:cman-realization-consistency}
Let $\mathcal N$ be a controlled Markovian activity network with an initial marking, and
let $\mathcal W_{\mathcal N}$ be its CMA realization.  Let
$\mathcal N^{\mathrm S}$ denote the corresponding Markovian CSAN
of Definition~\ref{def:cman} with the same initial marking, obtained by assigning every timed
activity an exponential activity-time distribution with rate $1$
and using the CMAN completion-rate function $\sigma$ as the
control-dependent processing-rate function.

Then the CMA $\mathcal W_{\mathcal N}$ and the Markovian CSA
realized by $\mathcal N^{\mathrm S}$ have the same states,
observable activities, admissible control actions, initial
distribution, and controlled transition kernel.  Moreover, under
corresponding policies, they have the same activity-labeled timed
execution law.
\end{proposition}

\begin{proof}
By construction, the underlying controlled probabilistic activity
network of $\mathcal N$ realizes the same CPA in both models.
Hence the two realizations have the same states, observable
activities, admissible control actions, initial distribution, and
controlled transition kernel.

In the corresponding Markovian CSAN, every timed activity has an
exponential activity-time distribution with rate $1$.  If the
activity is assigned control $c$ at marking $m$, its
control-dependent processing rate is $\sigma(m,a,c)$.  Therefore,
its effective completion rate is
\[
    1\cdot\sigma(m,a,c)
    =
    \sigma(m,a,c).
\]
By the definition of the CMA realization of a CMAN, the
corresponding activity in $\mathcal W_{\mathcal N}$ has the same
completion rate.

Thus, at every corresponding race-decision epoch and under
corresponding contingent-action selections, the two models have
the same exponential completion rates, the same distribution of
the completed activity, and the same controlled successor-state
kernel.  Induction over successive completion epochs therefore
gives equality of their activity-labeled timed execution laws.
\end{proof}

The rate-based semantics of controlled Markovian automata also provide
a direct connection with continuous-time Markov decision processes,
which we examine next.
\subsection{Continuous-Time Markov Decision Process (CTMDP)}

Continuous-time Markov decision processes
(CTMDPs) are a standard framework for
controlled systems with exponential
holding times
\cite{puterman1994markov,Howard71,
feinberg2002handbook}.

Under the pre-race semantics of a controlled
Markovian automaton (CMA), the controller
selects a contingent action before each
timed race. This selection determines the
completion rates of the enabled activities.
When an activity completes, its preassigned
action is executed, and the successor state
is sampled.

At the level of the physical-state process,
a CMA can be projected onto a CTMDP by
taking the CTMDP actions to be contingent-action
maps and marginalizing over the identity of
the completed activity. This projection
preserves the holding-time and successor-state
laws under the stated Markovian assumptions,
but in general does not preserve the
activity-labeled structure of the CMA.

\begin{definition}[Continuous-time Markov decision process]
\label{def:ctmdp}
A \emph{continuous-time Markov decision process} (CTMDP) is a
tuple
\[
    \mathcal D=(S,B,\eta_0,\lambda,K),
\]
where:
\begin{itemize}
    \item $S$ is a finite or countable state space;

    \item $B$ is a nonempty set of actions;

    \item $\eta_0\in\operatorname{Dist}(S)$ is the initial
    probability distribution;

    \item
    \[
        K:S\times B\rightharpoonup\operatorname{Dist}(S)
    \]
    is a partial probability kernel, called the
    \emph{embedded jump kernel};

    \item for each $s\in S$, the set of actions admissible at
    $s$ is defined by
    \[
        B(s):=
        \{\,b\in B:
        K(s,b,\cdot)\text{ is defined}\,\},
    \]
    and is required to be nonempty;

    \item
    \[
        \lambda:
        \operatorname{dom}(K)\longrightarrow
        \mathbb R_{\geq0}
    \]
    is the holding-time rate function, where
    \begin{equation}
        \operatorname{dom}(K)
        =
        \{\,(s,b)\in S\times B:b\in B(s)\,\}.
        \label{dom(K)}
    \end{equation}
\end{itemize}

Thus, for every $s\in S$ and $b\in B(s)$,
\[
    K(s,b,\cdot)\in\operatorname{Dist}(S),
\]
or equivalently,
\[
    K(s,b,s')\geq0
    \quad\text{and}\quad
    \sum_{s'\in S}K(s,b,s')=1.
\]

Under the \emph{jump-epoch interpretation}, upon entering
state $s$, an action $b\in B(s)$ is selected and remains in
force until the next jump. If $\lambda(s,b)>0$, then,
conditional on $(s,b)$, the holding time $\tau$ satisfies
\[
    \Pr\{\tau>t\mid s,b\}
    =
    e^{-\lambda(s,b)t},
    \qquad t\geq0,
\]
and the successor state is distributed according to
$K(s,b,\cdot)$, independently of the holding time conditional
on $(s,b)$.

If $\lambda(s,b)=0$, the holding time is infinite, so no jump
occurs while the state--action pair $(s,b)$ remains in force.
Although $K(s,b,\cdot)$ is defined because $b$ is admissible,
its value is behaviorally irrelevant in this case.

The kernel $K$ is an embedded \emph{jump} kernel rather than
a state-change kernel. In particular,
\[
    K(s,b,s)>0
\]
is permitted. Such a transition is a jump of the embedded
process even though the physical state before and after the
jump is the same.

Unless stated otherwise, the CTMDP is assumed to be
non-explosive under the policy class under consideration. If
non-explosion is not assumed, an explicit convention for the
process following an explosion must be specified.
\end{definition}
\begin{definition}[Reward structure for a CTMDP]
\label{def:ctmdp-reward-structure}
Let
\[
    \mathcal D=(S,B,\eta_0,\lambda,K)
\]
be a CTMDP. A reward structure for $\mathcal D$ is a pair
\[
    r^{\mathcal D}
    =
    \bigl(r_{\mathrm{rate}}^{\mathcal D},
          r_{\mathrm{jump}}^{\mathcal D}\bigr),
\]
where
\[
    r_{\mathrm{rate}}^{\mathcal D}
    :
    \operatorname{dom}(K)\longrightarrow\mathbb R
\]
is a bounded reward-rate function where $\operatorname{dom}(K)$ is defined as in (\ref{dom(K)}) and
\[
    r_{\mathrm{jump}}^{\mathcal D}
    :
    S\times B\times S\longrightarrow\mathbb R
\]
is a bounded jump-reward function, evaluated only on admissible
state--action pairs and successor states of positive probability.

If action $b\in B(s)$ is selected in state $s$, reward accumulates
at rate $r_{\mathrm{rate}}^{\mathcal D}(s,b)$ during the resulting
holding time, and a jump from $s$ to $s'$ yields the additional
reward
\[
    r_{\mathrm{jump}}^{\mathcal D}(s,b,s').
\]
Only values of $r_{\mathrm{jump}}^{\mathcal D}(s,b,s')$ for which
$K(s,b,s')>0$ affect execution rewards.
\end{definition}

Let
\[
 \mathcal W=(Q,\Sigma,C,P,\nu_0,\sigma)
\]
be a CMA as defined in
Definition~\ref{def:controlled-markovian-automaton}.
For each $q\in Q$ and $a\in\Sigma$, define
\[
 C(q,a)
 :=
 \left\{
 c\in C:
 \exists q'\in Q,\;
 P(q,a,c,q')>0
 \right\},
\]
and let
\[
 \Sigma(q):=\{a\in\Sigma:C(q,a)\neq\varnothing\}
\]
be the set of activities enabled at $q$.

A \emph{contingent action} admissible at $q$ is a function
that assigns an admissible control action to every activity
enabled at $q$. Accordingly, define
\[
 B(q):=
 \left\{
 b:\Sigma(q)\to C
 \;\middle|\;
 b(a)\in C(q,a)
 \text{ for every }a\in\Sigma(q)
 \right\}.
\]
Equivalently,
\[
 B(q)\cong
 \prod_{a\in\Sigma(q)} C(q,a).
\]
If $\Sigma(q)=\varnothing$, the empty-product convention is
used, so $B(q)$ consists of the unique empty function.

Define the global set of contingent actions by
\[
 B:=\bigcup_{q\in Q}B(q).
\]
A contingent action $b\in B$ is admissible at $q$ exactly
when $b\in B(q)$. Notice that the same function may be
admissible at more than one state; in that case it represents
the same element of $B$ at each such state.

For $q\in Q$, $b\in B(q)$, and $a\in\Sigma(q)$, define
\[
 \sigma_b(q,a)
 :=
 \sigma\bigl(q,a,b(a)\bigr),
\]
and define the total completion rate by
\[
 \Lambda(q,b)
 :=
 \sum_{a\in\Sigma(q)}
 \sigma_b(q,a).
\]
By the assumptions of the CMA,
\[
 0\leq\Lambda(q,b)<\infty.
\]

If $\Lambda(q,b)>0$, the waiting time to the next activity
completion is exponentially distributed with rate
$\Lambda(q,b)$, and
\[
 \Pr(A=a\mid q,b)
 =
 \frac{\sigma(q,a,b(a))}
      {\Lambda(q,b)},
 \qquad a\in\Sigma(q).
\]
If activity $a$ completes, the preassigned control $b(a)$ is
executed and the successor state is distributed according to
\[
 P(q,a,b(a),\cdot).
\]

Thus the selected contingent action $b$ may affect both the
holding-time distribution and the distribution of the
completed activity. The activity identity $a$ and the
executed control $b(a)$ are part of the CMA semantics. The
ordinary CTMDP projection defined below retains their effect
on the holding-time and successor-state distributions but
does not retain them as transition marks.

\begin{definition}[Contingent-action CTMDP projection of a CMA]
\label{def:cma-contingent-ctmdp}
Let
\[
 \mathcal W=(Q,\Sigma,C,P,\nu_0,\sigma)
\]
be a CMA with finite activity and control-action sets, and
let
\[
 B:=\bigcup_{q\in Q}B(q).
\]
Its \emph{contingent-action CTMDP projection} is
\[
 \mathcal M=(Q,B,\nu_0,\lambda,K),
\]
where
\[
 K:Q\times B\rightharpoonup\operatorname{Dist}(Q)
\]
is the partial embedded jump kernel whose domain is
\[
 \operatorname{dom}(K)
 =
 \{(q,b)\in Q\times B:b\in B(q)\}.
\]
Consequently, for every $q\in Q$,
\[
 B(q)
 =
 \{b\in B:(q,b)\in\operatorname{dom}(K)\}
 =
 \{b\in B:K(q,b,\cdot)\text{ is defined}\},
\]
in accordance with Definition~\ref{def:ctmdp}.

The holding-time rate function has the same state--action
domain,
\[
 \lambda:
 \operatorname{dom}(K)\to\mathbb R_{\geq0},
\]
and is defined by
\[
 \lambda(q,b):=\Lambda(q,b),
 \qquad (q,b)\in\operatorname{dom}(K).
\]

For $(q,b)\in\operatorname{dom}(K)$ with
$\Lambda(q,b)>0$, define
\begin{equation}
\label{eq:cma-contingent-kernel}
 K(q,b,q')
 :=
 \sum_{a\in\Sigma(q)}
 \frac{\sigma(q,a,b(a))}
      {\Lambda(q,b)}
 P(q,a,b(a),q'),
 \qquad q'\in Q.
\end{equation}
For $(q,b)\in\operatorname{dom}(K)$ with
$\Lambda(q,b)=0$, define
\[
 K(q,b,q')
 :=
 \mathbf 1_{\{q'=q\}},
 \qquad q'\in Q.
\]

For every $(q,b)\in\operatorname{dom}(K)$,
$K(q,b,\cdot)$ is a probability distribution on $Q$.
Indeed, if $\Lambda(q,b)>0$, then
\[
\begin{aligned}
 \sum_{q'\in Q}K(q,b,q')
 &=
 \sum_{a\in\Sigma(q)}
 \frac{\sigma(q,a,b(a))}
      {\Lambda(q,b)}
 \sum_{q'\in Q}P(q,a,b(a),q')\\
 &=
 \sum_{a\in\Sigma(q)}
 \frac{\sigma(q,a,b(a))}
      {\Lambda(q,b)}
 =1,
\end{aligned}
\]
whereas for $\Lambda(q,b)=0$ the assertion follows directly
from the definition of $K$.

When $\Lambda(q,b)=0$, we have
$\lambda(q,b)=0$, so the corresponding holding time is
infinite. Hence the value assigned to $K(q,b,\cdot)$ does
not affect the process dynamics; the choice
$K(q,b,\cdot)=\delta_q$ provides a probability kernel on the
entire admissible state--action domain, where $\delta_q\in\operatorname{Dist}(Q)$ is the Dirac
distribution concentrated at $q$, i.e.,
\[
 \delta_q(q')=
 \begin{cases}
 1, & q'=q,\\
 0, & q'\neq q.
 \end{cases}
\]

For $\Lambda(q,b)>0$, the kernel
\eqref{eq:cma-contingent-kernel} is obtained by marginalizing
over the completed activity. Therefore, the projected CTMDP
has the same conditional holding-time and successor-state
distributions as the CMA for each admissible pair $(q,b)$,
but it does not retain the identity of the completed activity
$a$ or the executed activity-specific control $b(a)$.
\end{definition}
\begin{remark}[Decision order and timed race]
\label{rem:ctmdp-decision-order}
The contingent-action CTMDP projection reflects the decision
order of the underlying CMA. At a decision epoch in state
$q$, the controller first selects an entire contingent action
\[
    b\in B(q),
\]
thereby provisionally assigning the control $b(a)$ to every
enabled activity $a\in\Sigma(q)$. These assignments determine
the activity-completion rates
\[
    \sigma(q,a,b(a))
\]
and hence both the total race rate
\[
    \Lambda(q,b)
    =
    \sum_{a\in\Sigma(q)}\sigma(q,a,b(a))
\]
and, when $\Lambda(q,b)>0$, the probability that activity $a$
wins the race,
\[
    \frac{\sigma(q,a,b(a))}{\Lambda(q,b)}.
\]
Only after the race is resolved is the winning activity $a$
known. The control $b(a)$ assigned to that activity is then
executed, and the successor state is sampled according to
\[
    P(q,a,b(a),\cdot).
\]

Thus the decision order is
\[
    Q_n
    \;\longrightarrow\;
    B_n
    \;\longrightarrow\;
    \text{timed race}
    \;\longrightarrow\;
    A_n
    \;\longrightarrow\;
    B_n(A_n)
    \;\longrightarrow\;
    Q_{n+1}.
\]
In particular, the controller does not observe the activity
that will complete and then choose its control; rather, it
commits to the complete contingent-action map before the race.
The CTMDP action is therefore the entire map $b$, not merely
the executed component $b(a)$. The remaining components are
provisional assignments that affect the race through their
rates but are not themselves executed.
\end{remark}

A stationary deterministic CMA policy selects
\[
    \varphi(q)\in B(q).
\]
For randomized contingent-action policies, the convention adopted
in this paper specifies activity-specific distributions
\[
    \kappa_a(\cdot\mid q)
    \in\operatorname{Dist}(C(q,a)),
    \qquad a\in\Sigma(q),
\]
whose conditional independence induces the product-form CTMDP policy
\begin{equation}
\label{eq:ctmdp-product-policy}
    \varphi(b\mid q)
    =
    \prod_{a\in\Sigma(q)}
    \kappa_a\bigl(b(a)\mid q\bigr),
    \qquad b\in B(q).
\end{equation}
The same correspondence applies to history-dependent policies by
conditioning the activity-specific kernels on the available history.
An unrestricted randomized CTMDP policy may instead assign an
arbitrary distribution in $\operatorname{Dist}(B(q))$ and is therefore,
in general, more general than the randomized contingent-action CMA
policy class adopted here.
\begin{proposition}[Relationship between CTMDPs and CMAs]
\label{prop:ctmdp-cma-relationship}
Every CTMDP admits a single-activity CMA representation that,
under corresponding policies and matching initial conditions,
preserves the physical-time state-process law.

Conversely, the contingent-action CTMDP projection of a CMA
given in Definition~\ref{def:cma-contingent-ctmdp} preserves,
under corresponding deterministic policies and under corresponding
randomized policies satisfying the product-form convention in
\eqref{eq:ctmdp-product-policy}, the physical-time state-process
law, provided the corresponding processes are non-explosive.

The projection generally does not preserve the identity of
the completed activity, the executed activity-specific control,
activity-labeled histories, or other quantities that depend on
the completed activity.
\end{proposition}

\begin{proof}
Let
\[
 \mathcal D=(S,B,\eta_0,\lambda,K)
\]
be a CTMDP as in Definition~\ref{def:ctmdp}. Construct a CMA
\[
 \mathcal W=(Q,\Sigma,C,P,\nu_0,\sigma)
\]
by setting
\[
 Q=S,\qquad
 \Sigma=\{a_0\},\qquad
 C=B,\qquad
 \nu_0=\eta_0,
\]
where $a_0$ is a dummy activity. For each $q\in Q$, define
\[
 C(q,a_0):=B(q),
\]
where
\[
 B(q)=
 \{c\in B:K(q,c,\cdot)\text{ is defined}\}.
\]
For every $c\in C(q,a_0)$ and $q'\in Q$, set
\[
 \sigma(q,a_0,c):=\lambda(q,c),
 \qquad
 P(q,a_0,c,q'):=K(q,c,q').
\]

Since $a_0$ is the only enabled activity, each admissible CMA
contingent action $b$ at $q$ is uniquely determined by
$b(a_0)=c$ for some $c\in B(q)$. Hence admissible CTMDP
actions and admissible CMA contingent actions are in
one-to-one correspondence. Because only one activity is enabled,
every randomized distribution over the CMA contingent-action set
is trivially of product form. For corresponding $c$ and $b$,
\[
 \Lambda(q,b)
 =
 \sigma(q,a_0,b(a_0))
 =
 \lambda(q,c).
\]
If this rate is positive, $a_0$ completes with probability one and
\[
 P(q,a_0,c,q')=K(q,c,q'),
\]
so the two models have identical conditional holding-time and
successor-state distributions. If the rate is zero, the holding
time is infinite in both models. Thus, under corresponding policies
and matching initial conditions, the two models induce the same
physical-time state-process law.

Conversely, let $\mathcal W$ be a CMA and let $\mathcal M$ be its
contingent-action CTMDP projection from
Definition~\ref{def:cma-contingent-ctmdp}. For every admissible
pair $(q,b)$ with $\Lambda(q,b)>0$, both models have an
exponential holding time with rate
\[
 \lambda(q,b)=\Lambda(q,b),
\]
and the CTMDP successor-state kernel is
\[
 K(q,b,q')
 =
 \sum_{a\in\Sigma(q)}
 \frac{\sigma(q,a,b(a))}
      {\Lambda(q,b)}
 P(q,a,b(a),q'),
\]
which is precisely the CMA successor-state distribution obtained
by marginalizing over the completed activity. If
$\Lambda(q,b)=0$, both models have infinite holding time.
Therefore, corresponding deterministic policies, and corresponding
randomized policies whose distributions over $B(q)$ have the
product form in \eqref{eq:ctmdp-product-policy}, have identical
conditional holding-time and successor-state laws at successive
decision epochs. Under matching initial conditions and
non-explosion, they consequently induce the same physical-time
state-process law.

The marginalization over the completed activity removes its
identity and, consequently, the identity of the executed component
$b(a)$. Hence the ordinary CTMDP projection does not in general
preserve activity-labeled histories, activity-specific controls,
or other quantities depending on the completed activity.
\end{proof}
\begin{proposition}[Reward correspondence]
\label{prop:ctmdp-reward-correspondence}
Let
\[
    r=(r_{\mathrm{rate}},r_{\mathrm{act}})
\]
be the bounded reward model of a CMA, where
\[
    r_{\mathrm{rate}}:\mathsf{Dom}(B)\to\mathbb{R}
\]
is the state--contingent-action reward rate where $\mathsf{Dom}(B)$ is defined as in (\ref{dom(B)}) and
\[
    r_{\mathrm{act}}:Q\times\Sigma\times C\times Q
    \to\mathbb{R}
\]
is the activity-completion reward.

Equip the contingent-action CTMDP projection with the reward structure
\[
    r^{\mathcal D}
    =
    \bigl(r_{\mathrm{rate}}^{\mathcal D},
          r_{\mathrm{jump}}^{\mathcal D}\bigr)
\]
of Definition~\ref{def:ctmdp-reward-structure}, defined as follows.
For every $(q,b)\in\mathsf{Dom}(B)$,
\[
    r_{\mathrm{rate}}^{\mathcal D}(q,b)
    =
    r_{\mathrm{rate}}(q,b).
\]
For every $(q,b)\in\mathsf{Dom}(B)$ with
$\Lambda(q,b)>0$ and every $q'\in Q$ such that
$K(q,b,q')>0$, define
\begin{equation}
\label{eq:ctmdp-induced-jump-reward}
r_{\mathrm{jump}}^{\mathcal D}(q,b,q')
=
\frac{
\displaystyle
\sum_{a\in\Sigma(q)}
\sigma(q,a,b(a))
P(q,a,b(a),q')\,
r_{\mathrm{act}}(q,a,b(a),q')
}{
\displaystyle
\Lambda(q,b)K(q,b,q')
}.
\end{equation}
The values of $r_{\mathrm{jump}}^{\mathcal D}(q,b,q')$ when
$K(q,b,q')=0$, and for $(q,b)\in\mathsf{Dom}(B)$ with
$\Lambda(q,b)=0$, are immaterial and may be set to zero.

For $\Lambda(q,b)>0$, the corresponding conditional expected
jump reward is
\begin{equation}
\label{eq:ctmdp-expected-jump-reward}
\begin{split}
\bar r_{\mathrm{jump}}(q,b)
&=
\sum_{q'\in Q}
K(q,b,q')\,
r_{\mathrm{jump}}^{\mathcal D}(q,b,q')
\\
&=
\sum_{a\in\Sigma(q)}
\frac{\sigma(q,a,b(a))}{\Lambda(q,b)}
\sum_{q'\in Q}
P(q,a,b(a),q')\,
r_{\mathrm{act}}(q,a,b(a),q').
\end{split}
\end{equation}

Under corresponding deterministic policies, or corresponding
randomized policies satisfying the product-form convention in
\eqref{eq:ctmdp-product-policy}, and provided the corresponding
processes are non-explosive, the CMA and its contingent-action
CTMDP projection have identical expected finite-horizon and
discounted accumulated rewards.
\end{proposition}

\begin{proof}
Fix $(q,b)\in\mathsf{Dom}(B)$ with $\Lambda(q,b)>0$.
The conditional probability that activity $a$ completes next is
\[
    \frac{\sigma(q,a,b(a))}{\Lambda(q,b)}.
\]
Conditional on completion of $a$, the successor state has
distribution
\[
    P(q,a,b(a),\cdot).
\]
Hence the joint conditional probability that activity $a$
completes and the successor state is $q'$ is
\[
    \frac{\sigma(q,a,b(a))}{\Lambda(q,b)}
    P(q,a,b(a),q').
\]
By Definition~\ref{def:cma-contingent-ctmdp}, 
marginalizing over the completed activity gives
\[
    K(q,b,q')
    =
    \sum_{a\in\Sigma(q)}
    \frac{\sigma(q,a,b(a))}{\Lambda(q,b)}
    P(q,a,b(a),q').
\]
Therefore, whenever $K(q,b,q')>0$,
\[
    r_{\mathrm{jump}}^{\mathcal D}(q,b,q')
\]
defined in \eqref{eq:ctmdp-induced-jump-reward} is precisely the
conditional expected CMA activity-completion reward given the
current state $q$, the selected contingent action $b$, and the
successor state $q'$.

Averaging over the successor state therefore gives
\[
\begin{split}
\bar r_{\mathrm{jump}}(q,b)
&=
\sum_{q'\in Q}
K(q,b,q')\,
r_{\mathrm{jump}}^{\mathcal D}(q,b,q')
\\
&=
\sum_{a\in\Sigma(q)}
\frac{\sigma(q,a,b(a))}{\Lambda(q,b)}
\sum_{q'\in Q}
P(q,a,b(a),q')\,
r_{\mathrm{act}}(q,a,b(a),q'),
\end{split}
\]
which is exactly the conditional expected CMA
activity-completion reward given $(q,b)$. Thus the CMA and its
CTMDP projection have the same conditional expected reward at
each jump.

Moreover,
\[
    r_{\mathrm{rate}}^{\mathcal D}(q,b)
    =
    r_{\mathrm{rate}}(q,b).
\]
Under the policy correspondence of
Proposition~\ref{prop:ctmdp-cma-relationship}, the CMA and its
CTMDP projection have the same conditional holding-time and
successor-state laws, together with matching
state--contingent-action laws. Hence their expected accumulated
reward-rate contributions coincide. Combining this with equality
of the conditional expected jump rewards yields equality of
expected finite-horizon accumulated rewards.

Finally, boundedness of the reward functions, non-explosion, and
discounting justify the corresponding equality of expected
discounted accumulated rewards.
\end{proof}
\begin{remark}[Significance of the CTMDP--CMA relationship]
The preceding propositions distinguish representation from
projection. Every CTMDP admits a single-activity CMA
representation that preserves its physical-time state-process
law. Conversely, the contingent-action CTMDP projection of a
CMA preserves this law under the stated policy-correspondence
and non-explosion conditions, while marginalizing over the
completed activity. Thus a CMA may retain activity-labeled
structure that is absent from its ordinary CTMDP projection.
The reward construction of
Proposition~\ref{prop:ctmdp-reward-correspondence} further shows
that expected finite-horizon and discounted accumulated rewards
can be preserved even when activity identity is marginalized,
provided the induced CTMDP reward structure is used.
\end{remark}

\begin{remark}[Optimization and learning from the CTMDP projection]
\label{rem:ctmdp-optimal-control}
The contingent-action CTMDP projection provides a bridge from
CMA semantics to established methods for continuous-time stochastic
control and policy synthesis. Classical CTMDP results may therefore
be applied whenever their hypotheses are satisfied and the relevant
objective and admissible policy class are preserved.

An unrestricted randomized CTMDP policy may use an arbitrary
distribution over $B(q)$, whereas the randomized contingent-action
CMA policy convention adopted here corresponds to the product-form
subclass in \eqref{eq:ctmdp-product-policy}. For finite state spaces
with finite admissible action sets and suitable boundedness
assumptions, infinite-horizon expected discounted-reward objectives
admit stationary deterministic optimal policies. Standard solution
techniques include continuous-time Bellman equations and, under
appropriate assumptions, uniformization followed by discrete-time
dynamic programming
\cite{puterman1994markov,Howard71,feinberg2002handbook}.
For these objectives, when the CTMDP admits an optimal
deterministic policy within the policy class corresponding to the
CMA, that policy transfers directly to the CMA and yields the same
optimal value under the preservation results established above.
No product-form restriction is involved in this deterministic
case.  The product-form convention becomes relevant only when
transferring genuinely randomized CTMDP policies to randomized
contingent-action policies of the CMA.

As the state space becomes too large for exact methods, the CTMDP
projection may also provide a basis for approximate dynamic
programming and reinforcement-learning methods. Value- or
action-value-function approximation, sampling and simulation, and
direct policy approximation can be used to avoid explicit enumeration
of the entire state space. When a generative model of the underlying CSAN or CMAN is available,
it may additionally support model-based reinforcement learning,
planning, and simulation without requiring explicit construction of
the full CMA or CTMDP state space.  For large structured systems, the
activity-network representation may also provide useful information
for state aggregation, feature construction, guided exploration, and
structure-aware policy synthesis.

These observations suggest a natural scalability progression from
exact stochastic dynamic programming to approximate methods and
learning-based approaches as explicit state-space analysis becomes
infeasible. They do not, by themselves, imply optimality,
convergence, or sample-efficiency guarantees; such guarantees depend
on the objective, approximation architecture, learning algorithm,
information model, and assumptions on the underlying continuous-time
process.

For the projection developed above, the action at state $q$ is the
entire contingent-action map $b\in B(q)$, selected before the timed
race and retained throughout that race. If activity $a$ completes,
the already assigned component $b(a)$ is executed. Under
corresponding policies, the projection preserves the conditional
holding-time and successor-state laws, subject to the conditions
stated in Proposition~\ref{prop:ctmdp-cma-relationship}, while
marginalizing over the identity of the completed activity.
Consequently, transfer of optimization or learning results back to
the CMA requires preservation of the relevant policy information
and objective and, for randomized policies, compatibility with the
product-form contingent-action convention. Objectives that depend
essentially on activity identities, executed activity-specific
controls, activity-labeled histories, or activity-dependent reward
distributions require a separate preservation argument or an
information-retaining state augmentation. For expected
finite-horizon and discounted accumulated rewards, the construction
of Proposition~\ref{prop:ctmdp-reward-correspondence} provides such
a preservation result.

Related discounted-reward results extend to countable state spaces
under appropriate boundedness and regularity assumptions, whereas
unbounded rates, undiscounted objectives, or more general state
spaces require additional conditions and techniques
\cite{feinberg2002handbook,HernandezLerma1996}. Time-bounded
objectives may require time-dependent policies, and their
computational complexity depends on the representation of rates,
horizons, thresholds, and numerical precision; discrete-time
complexity results do not automatically carry over to CTMDPs
\cite{Mundhenk2000}.

Finally, a CSA whose activity-time distributions are not all
exponential is not Markovian in the sense of
Definition~\ref{def:markovian-csa} and need not admit the direct 
CMA representation used here.  In such cases, timing-state
augmentation, approximation, or more general stochastic-process
representations may be required.  Corresponding convergence and
policy-transfer guarantees then depend on the timing semantics,
objective, and policy class
\cite{Howard71,YounesSimmons2004}.  Likewise, PAC-style guarantees
do not follow automatically from a CTMDP projection; they require
a specified learning setting together with appropriate statistical
and model-approximation assumptions.
\end{remark}

\section{Conclusion}

This paper has developed Controlled Stochastic Activity Networks
(Controlled SANs) and their automata-theoretic semantics as a
formal framework for integrating concurrency, nondeterminism,
probabilistic branching, stochastic timing, and policy-driven
control in discrete-event and real-time systems, with particular
emphasis on reinforcement learning and safety-critical
applications. The framework proceeds hierarchically through
controlled automata, controlled probabilistic automata, and
controlled stochastic automata. Within this hierarchy, we have
formalized policy classes with different memory and computational
resources, studied their expressive power through associated
language classes, and developed behavioral equivalences to
support abstraction and compositional reasoning.

At the nondeterministic level, we have distinguished structurally
available control actions from behaviorally admissible controlled
transitions and established an operational semantics relating
controlled activity networks to controlled automata. We have
also investigated the modeling power of computable controlled
activity networks. In particular, Theorem~\ref{thm:modeling-power}
establishes a constructive realization result for decidable,
effectively finitely branching controlled automata with finite
activity and control sets and a finite, nonempty initial-state
set. Under these conditions, a well-behaved computable controlled
activity network can realize the reachable behavior of the
corresponding controlled automaton up to effective isomorphism.
This result provides a formal connection between the
activity-network representation and the computational
expressiveness of controlled automata.

At the probabilistic level, an activity-generation law together
with contingent activity-indexed control assignments yields a
DTMDP representation of a controlled probabilistic automaton.
The controller first selects a contingent action; the next
activity is then generated conditionally on the current state
and that action, and the control already assigned to the
generated activity is executed. For randomized contingent-action
policies, the provisional activity-specific assignments are
conditionally independent given the information available to the
controller and therefore induce a product-form distribution over
contingent-action maps. Conversely, every conventional DTMDP
admits a single-activity controlled probabilistic automaton
representation. This relationship provides a semantic bridge
from the probabilistic framework to established methods for
discrete-time stochastic control, reward optimization, and
policy synthesis, subject to preservation of the relevant
policy class and objective.

At the stochastic level, the same contingent-action principle
is combined with timed race semantics. Before each race, the
controller provisionally assigns an admissible control to every
enabled activity. These assignments may affect activity
completion rates and hence the race itself. When an activity
completes, its already assigned control is executed, and the
probabilistic successor state is sampled; a new contingent
action is selected before the next race. For the Markovian
subclass, this semantics yields a conventional CTMDP projection
in which contingent-action maps serve as CTMDP actions. Under
the stated assumptions, the projection preserves the
conditional holding-time and successor-state laws at the
physical-state level under corresponding deterministic
policies and under randomized policies satisfying the adopted
product-form convention, while marginalizing over the identity
of the completed activity. Conversely, every conventional
CTMDP admits a single-activity controlled Markovian automaton
representation.

The DTMDP representation and CTMDP projection also provide
corresponding reward representations, allowing expected
finite-horizon and discounted accumulated rewards to be
preserved under the conditions established above. These
relationships connect the formal semantics developed here
with established methods for optimization and policy synthesis.
For models that permit explicit state-space analysis, classical
dynamic-programming and stochastic-control methods can be
applied under their usual assumptions. As state spaces become
too large for exact methods, the same semantic models provide
a basis for approximate dynamic programming, sampling and
simulation, and reinforcement-learning approaches. When an
explicit model is available, it can support model-based
learning, planning, and simulation. Moreover, the underlying
activity-network structure suggests opportunities for state
aggregation, feature construction, guided exploration,
compositional analysis, and structure-aware policy synthesis.
Thus, the semantic framework provides a foundation for a
progression from exact optimization to approximation and
learning as explicit state-space analysis becomes infeasible.

These connections do not by themselves provide automatic
optimality, convergence, complexity, or learning guarantees.
Such guarantees depend on the objective, controller information,
state and action spaces, admissible policy class, reward model,
and assumptions of the particular optimization or learning
method. In particular, unrestricted randomized DTMDP or CTMDP
policies may be more general than the randomized
contingent-action policies adopted here, whose activity-specific
assignments induce product-form distributions. For objectives
admitting deterministic optimal policies within the
corresponding policy class under the relevant hypotheses,
this distinction does not affect the corresponding optimal
value; when randomization is essential, however, the
policy-class restriction must be respected explicitly.

For general non-exponential stochastic timing, a conventional
CTMDP projection need not exist; residual timing information
or other state augmentation, approximation, or more general
stochastic-process representations may be required.
Accordingly, results from DTMDPs or CTMDPs transfer to
Controlled SANs only when the corresponding representation
or projection preserves the relevant objective, information
structure, and admissible policy class.

An important motivation for these formal foundations is the
development of learning-enabled controllers for safety-critical
systems. In such applications, reinforcement-learning methods
must operate in environments characterized by concurrency,
uncertain outcomes, stochastic timing, and potentially complex
control constraints. Controlled SANs provide a structured
semantic framework in which the behavior of the controlled
system, the admissible decisions, and the information available
to the controller can be specified explicitly. This separation
of system dynamics and control policies offers a basis for
combining learning-based optimization with formal verification
and quantitative analysis. However, the establishment of
application-specific safety guarantees requires additional
verification or synthesis results beyond the semantic
foundations developed here.

Another potential application concerns agentic artificial
intelligence (AI) systems, in which autonomous agents select
actions, interact with their environments, and execute
possibly concurrent workflows under uncertainty and timing
constraints. Controlled SANs suggest a formal modeling
approach in which agent activities and environmental
interactions are represented by controlled state transitions,
while explicit control policies govern the selection of
admissible actions. Such models could provide a basis for
investigating action admissibility, reachability, stochastic
execution behavior, and temporal properties of agent
workflows. In multi-agent settings, the activity-network
structure may also support the representation of interacting
agents and their shared resources. Establishing suitable
modeling abstractions, compositional verification techniques,
and practical implementations for agentic AI systems remains
a direction for future research.

These results position Controlled SANs as a semantic foundation
connecting structured stochastic modeling with formal
verification, performance analysis, optimization, and
reinforcement learning. They also suggest a path toward
automated tools in which a Controlled SAN model is used to
derive suitable decision-process representations, synthesize
exact or approximately optimal policies, and evaluate or
verify the resulting behavior against specified quantitative
and qualitative objectives.

Future work includes policy-preserving abstraction,
temporal-logic and probabilistic verification, scalable
controller synthesis, quantitative approximation bounds,
compositional and structure-aware optimization, and
reinforcement-learning methods tailored to the timing,
contingent-action, and information semantics of Controlled SANs.
Particular priorities are the development of learning-based
controllers for safety-critical systems, the integration of
formal safety constraints into learning and policy synthesis,
and the establishment of statistical or PAC-style guarantees
under explicitly stated assumptions. Another important
direction is integrated tool support for modeling, simulation,
verification, policy synthesis, and performance analysis of
large controlled stochastic systems, including exploratory
applications to agentic AI workflows.

\section*{AI Assistance Statement}
During the preparation of this work, the author used ChatGPT (OpenAI) to assist with language polishing, technical phrasing, and improving the clarity of exposition. After using this tool, the author reviewed and edited the content as needed and takes full responsibility for the content of the publication.

\bibliographystyle{elsarticle-num}
\bibliography{references}
\begin{figure}[htbp]
  \centering
  \vspace*{-3cm} 
  \hspace*{-2cm} % Moves the image up
  \includegraphics[scale=1]{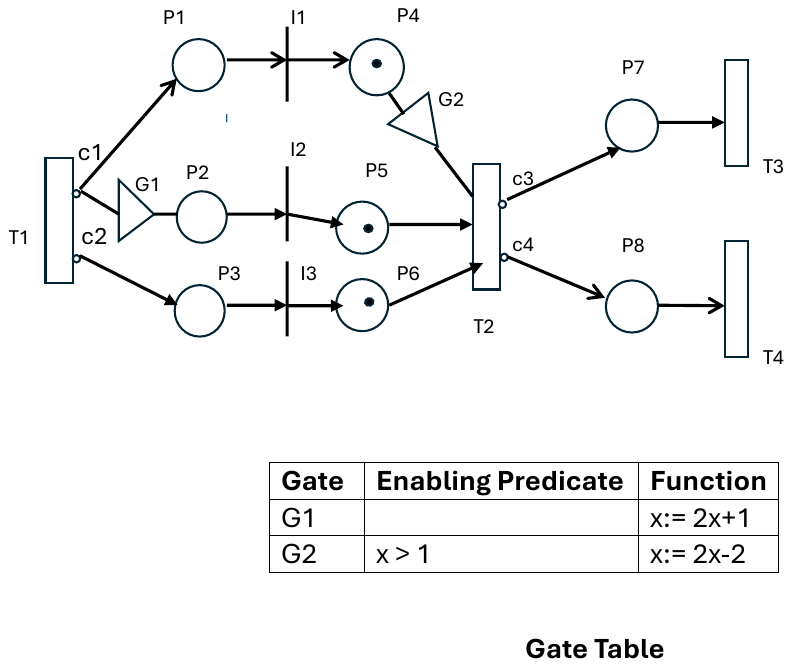}
  \captionsetup{labelfont={large}, font=large}
  \caption{An example of a simple controlled activity network with a marking.}
  \label{diagram}
\end{figure}
\begin{figure}[htbp]
  \centering
  \vspace*{-3cm} 
  \hspace*{-2cm} % Moves the image up
  \includegraphics[scale=1]{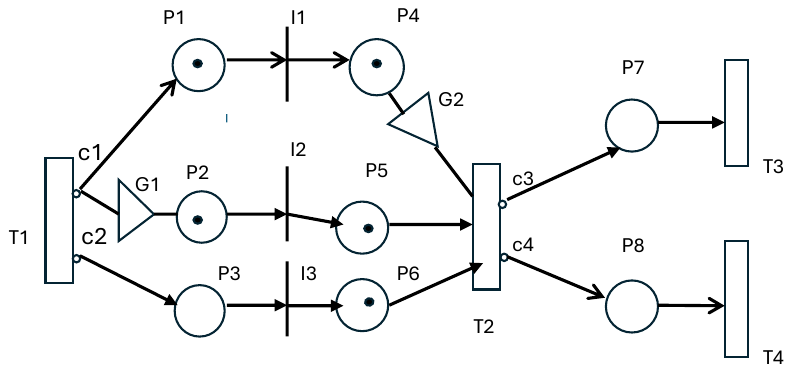}
  \captionsetup{labelfont={large}, font=large}
  \caption{Marking of the model after $T1$ completes via control action $c1$ in the model of Figure 1.}
  \label{diagram}
\end{figure}
\begin{figure}[htbp]
  \centering
  \vspace*{-3cm} 
  \hspace*{-2cm} % Moves the image up
  \includegraphics[scale=1]{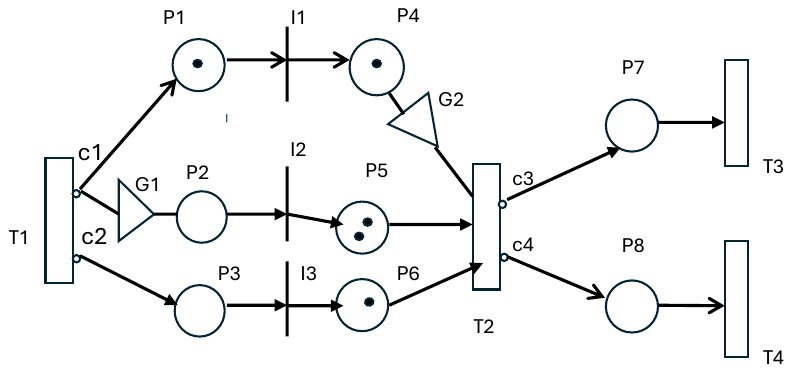}
  \captionsetup{labelfont={large}, font=large}
  \caption{Marking of the model after $I2$ completes in the model of Figure 2.}
  \label{diagram}
\end{figure}
\begin{figure}[htbp]
  \centering
  \vspace*{-3cm} 
  \hspace*{-2cm} % Moves the image up
  \includegraphics[scale=1]{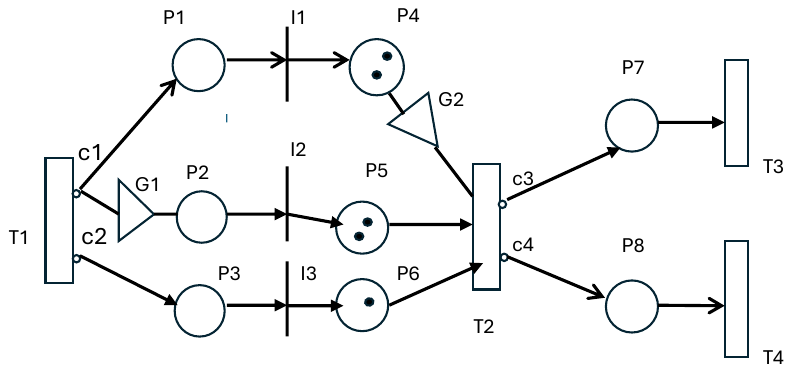}
  \captionsetup{labelfont={large}, font=large}
  \caption{Marking of the model after $I1$ completes in the model of Figure 3.}
  \label{diagram}
\end{figure}
\newpage
$\;$
\\
$\;$
% Replace the existing Figure 5 environment with this block.
% Requires \usepackage{graphicx} and \usepackage{caption} (already in CSAN).
\begin{figure}[!htbp]
\centering
\setlength{\unitlength}{1.3pt}
\begin{picture}(350,235)(0,-230)
% Timed activity a, control c
\put(160,-43){\rule{3mm}{36pt}}
\put(175,-40){$a$}
\put(170,-25){\circle{5}}
\put(170,-18){$c$}
% Triangular input gate G_a
\put(112,-37){\line(0,1){24}}
\put(112,-13){\line(2,-1){24}}
\put(112,-37){\line(2,1){24}}
\put(115,-28){\makebox(16,10){$G_a$}}
\put(136,-25){\vector(1,0){24}}
% Timed completion and successor module
\put(170,-25){\vector(1,0){40}}
\put(220,-25){\circle{20}}
\put(205,-50){$P_{ac}$}
\put(230,-25){\vector(1,0){50}}
\put(280,-50){\framebox(50,50){\large $K_{R_{ac}}$}}
% State place
\put(160,-150){\circle{20}}
\put(150,-175){$P_A$}
% Successor computation reads and updates P_A
\put(280,-42){\vector(-1,-1){110}}
\put(168,-157){\vector(3,2){160}}
% Gate reads state marking without consuming it
\put(150,-150){\line(-1,0){58}}
\put(92,-150){\line(0,1){125}}
\put(92,-25){\vector(1,0){20}}
% Completion token: K_Rac -> P_done
\put(310,-50){\vector(0,-1){30}}
\put(310,-90){\circle{20}}
\put(323,-93){$P_{\mathrm{done}}$}
% Standard input arc: P_done -> instantaneous cleanup activity tau
\put(310,-100){\vector(0,-1){22}}
\put(296,-125){\rule{28\unitlength}{1 \unitlength}}
\put(325,-125){$\tau$}
% Activity -> single triangular output gate H_cleanup (downward)
\put(310,-129){\vector(0,-1){15}}
\put(298,-145){\line(1,0){24}}
\put(298,-145){\line(1,-2){12}}
\put(322,-145){\line(-1,-2){12}}
\put(301,-153){\makebox(18,9){ \small $H$}}
% Output gate acts on remaining temporary/synchronization places
\put(310,-169){\vector(0,-1){16}}
\put(310,-195){\makebox(0,0){\scriptsize cleanup of auxiliary places}}
\end{picture}
\captionsetup{font=large}
\caption{Schematic representation of the
successor-computation subnetwork $K_{R_{ac}}$
associated with timed activity $a$ and control
action $c$. For each stable marking $\mu$
encoding a state $q$ such that
$c\in\Gamma(\mu,a)=C(q,a)$, the subnetwork
computes and nondeterministically selects a
successor state
$q'\in\operatorname{Succ}(q,a,c)$,
whose code is stored in $P_A$.
The instantaneous cleanup activity $\tau$,
triggered by $P_{\mathrm{done}}$, resets the
auxiliary places.}
\end{figure}
$\;$
\\
$\;$
\vspace{3cm}
\\

\begin{figure}[!htbp]
\centering
\setlength{\unitlength}{1.3pt}
\begin{picture}(350,145)(0,-125)

% Initial place P_0
\put(50,-10){\circle{20}}
\put(50,-37){$P_{0}$}
\put(60,-10){\vector(1,0){80}}

% Initialization module K_Q0
\put(140,-40){\framebox(60,60){\large $K_{Q_{0}}$}}

% State place P_A
\put(280,-10){\circle{20}}
\put(280,-37){$P_{A}$}
\put(200,0){\vector(1,0){80}}
\put(280,-20){\vector(-1,0){80}}

% Completion place P_done
\put(280,-90){\circle{20}}
\put(293,-95){$P_{\mathrm{done}}$}

% Directed edge from K_Q0 to P_done
\put(200,-40){\vector(2,-1){78}}

\end{picture}

\captionsetup{font=large}
\caption{Schematic representation of the initialization subnetwork \(K_{Q_0}\). 
Starting from the initial marking, the subnetwork nondeterministically selects a state \(q_0\in Q_0\), 
stores its code in \(P_A\), and signals completion through \(P_{\mathrm{done}}\).}
\end{figure}

\end{document}